\documentclass[11pt]{article}

\usepackage{authblk} % Author affiliation
\usepackage{hyperref} % Hyperlink
\usepackage{abstract} % Abstract formatting

\usepackage{fullpage} % To use less margin to save pages, and allowing us to use longer equations

\usepackage{xcolor}

\usepackage{amsfonts,amsmath,amsthm,mathtools,amssymb} % Basic math tools
\usepackage{physics} % Basic physics tools
\usepackage{bbm} %Identity
\usepackage{appendix} % To use Appendix

\usepackage[utf8]{inputenc}
\usepackage[T1]{fontenc}

\renewcommand{\braket}[2]{\left\langle #1 , #2 \right\rangle}
\newcommand{\sbra}[1]{\langle #1 |}
\newcommand{\sket}[1]{| #1 \rangle}

\newcommand{\snorm}[1]{\| #1 \|}

\newcommand{\bAR}{\mathbf{A}^{\!R}}

\newcommand{\HNR}{H_{N,R}}

\newcommand{\hc}{\mathrm{h.c.}}

\renewcommand{\d}[1]{\mathrm{d}#1}

\def\HS{\mathrm{HS}}
\def\op{\mathrm{op}}

\newcommand{\one}{\mathbbm{1}}

\def\ii{\mathrm{i}}
\def\im{\mathrm{i}}

\def\bA{\mathbf{A}}

\def\bD{\mathbf{D}}

\def\bbR{\mathbb{R}}
\def\R{\mathbb{R}}

\theoremstyle{plain} 
\newtheorem{theorem}{Theorem}[section]
\newtheorem{assumption}[theorem]{Assumption}
\newtheorem{proposition}[theorem]{Proposition}
\newtheorem{lemma}[theorem]{Lemma}
\newtheorem{corollary}[theorem]{Corollary}

\theoremstyle{remark} 
\newtheorem{remark}[theorem]{Remark}

\theoremstyle{definition} 
\newtheorem{definition}{Definition}

\numberwithin{equation}{section}

\def\b{|}

\newcommand{\rateR}{(\log N)^{-\frac12 + \varepsilon}}
\title{
Derivation of the Chern--Simons--Schr\"odinger equation\\
from the dynamics of an almost-bosonic-anyon gas
}

\author{
Th\'eotime Girardot%
\thanks{GSSI, Via communale, L’Aquila, Italy
	\href{mailto:theotime.girardot@gssi.it}{theotime.girardot@gssi.it}}
\quad
Jinyeop Lee%
\thanks{Department of Applied Mathematics,
	Kyung Hee University,
	1732 Deogyeong-daero, Giheung-gu,
	Yongin-si, Gyeonggi-do, South Korea
	\href{mailto:jinyeop.lee@khu.ac.kr}{jinyeop.lee@khu.ac.kr}}
}

\date{\today}

\begin{document}

\maketitle

\begin{abstract}
	We study the time evolution of an initial product state in a system of almost-bosonic-extended-anyons in the large-particle limit. 
	We show that the dynamics of this system can be well approximated, in finite time, by a product state evolving under the effective Chern--Simons--Schr\"odinger equation. 
	Furthermore, we provide a convergence rate for the approximation in terms of the radius $R = \rateR$ of the extended anyons. 
	These results establish a rigorous connection between the microscopic dynamics of almost-bosonic-anyon gases and the emergent macroscopic behavior described by the Chern--Simons--Schr\"odinger equation.
\end{abstract}

%\tableofcontents

\section{Introduction}
\subsection{General Introduction}

Anyons are two-dimensional quasi-particles interacting through an Aharonov--Bohm magnetic flux. By gauging this interaction directly into the wave function, they can be regarded as free particles with a new symmetry that differs from the bosonic or fermionic cases. Anyons possess a statistical parameter, a real number denoted by $\alpha \in [0,1]$, which interpolates between bosons ($\alpha = 0$) and fermions ($\alpha = 1$). We refer to the survey \cite{Lun_24} for a broader overview of various mathematical developments concerning anyon gases.

This paper focuses on the dynamics of an almost-bosonic anyon gas, where the statistical parameter $\alpha$ tends to zero as the number of particles becomes large. The ground state energy per particle for such systems has been previously studied in \cite{CorDubLunRou-19, CorLunRou-16, Gir_20, LunRou,Vis_25,AEGGLN_26,Ata_Lun_Gir_25}. Also note that anyons can be based on fermionic wave functions giving rise to the study of almost-fermionic anyons \cite{Rou_Lun_Lev_26,Rou_Gir_21,Lam_Lun_Rou_23,Lun_17,Lun_Rou_16}.  
It has been shown that the ground-state energy per particle can be approximated by a one-body functional.
The associated ground state of the $N$-body system then converges to a product state based on the minimizer of the one-body functional. 
The aim of the present paper is to show that, for a sufficiently large number of particles, such a product state will remain approximately in a product form over time, and to present an effective governing equation for this behavior. 

The dynamical properties of a product state can be analyzed using various methods, including the BBGKY hierarchy \cite{Erdos2010Derivation-of-the-GP,ErdYau-01,GinVel-79,GinVel-79b,Hepp1974,Spohn1980}, semi-classical analysis \cite{Am_Ni_08,Am_Ni_09,Am_N_11,Gol_Mou_Paul_16,Go_Pa_19,Gol_Paul_22} and the coherent state approach \cite{Benedikter2015,Boccato2018CompleteBEC,BocCenSch-15,ChenLeeLee2018,ChenLeeSchlein2011,Rodnianski2009}. 
In this paper, neither of these two strategies will be followed. We will apply a third strategy inspired by Fock space methods consisting of counting the number of particles in the condensate state, which has been developed in \cite{KnoPic-10,MitrouskasPetratPickl2019,Pickl2011}. This approach will allow us to demonstrate the convergence of a product state $\varphi_0^{\otimes N}$ to a product state $\varphi_t^{\otimes N}$, where the evolution of $\varphi_t$ is governed by the Chern--Simons--Schr\"odinger (CSS) equation with initial data $\varphi_0$.
This is in agreement with what has been shown by physicists in the early 1990s, i.e., that many-particle anyon systems can be described using the CSS equation \cite{IL1992PR,JP1990PRD,JP1990PRL}.
In the following sections, we introduce the Hamiltonian describing the anyon gas, define the dynamical CSS equation and present the main result followed by discussions of its limitations and implications.

\subsection{The Anyon Gas}
\subsubsection{Point-Like Anyons}
To describe anyons, the wave function of the system has to formally behave as
\begin{equation}\label{def:sym}
	\tilde{\Psi}(x_{1},\dots,x_{j},\dots,x_{k},\dots,x_{N})=e^{\im\alpha\pi}\tilde{\Psi}(x_{1},\dots,x_{k},\dots,x_{j},\dots,x_{N})
\end{equation}
when we exchange particles $x_j$ and $x_k$. If we apply the same exchange twice, following the rule of \eqref{def:sym}, we observe that $\tilde{\Psi}$ is multi-valued. To overcome this issue, we obtain this symmetry by writing
\[
\tilde{\Psi}(x_{1},\dots,x_{N})=\prod_{j<k}e^{\im\alpha\phi_{jk}}\Psi(x_{1},\dots,x_{N}) \quad\text{where}\quad
\phi_{jk}=\arg\frac{x_{j}-x_{k}}{|x_{j}-x_{k}|}
\]
with $\Psi$ a bosonic wave function, i.e., symmetric under particle exchange. In that way, we can encode the symmetry \eqref{def:sym} into a product playing the role of the phase of a bosonic wave function. Acting on $\tilde{\Psi}$ with the Laplacian is then equivalent to acting on $\Psi$ with a magnetic Laplacian in which an $\alpha$-dependent magnetic potential appears. This is called the gauge picture of the anyonic problem and consists in working with the Hamiltonian
\begin{equation}\label{def:HN}
	H_{N}=\sum_{j=1}^{N}(-\im\nabla_j+\alpha \bA_j)^{2}\quad\text{with}\quad \bA_j=\sum_{k\neq j}\frac{(x_{j}-x_{k})^{\perp}}{|x_{j}-x_{k}|^{2}}
\end{equation}
acting on $L^2_{\mathrm{sym}}(\R^{2N})$. Here, we have denoted $(x,y)^{\perp}=(-y,x)$.

The operator $\bA_j(x_1, x_2,\dots,x_N)=\bA_j$ is the statistical gauge vector potential felt by particle $j$ due to the influence of all the other particles. 
In this picture, anyons are described as bosons, each of them carrying an Aharonov--Bohm magnetic flux of strength $\alpha$.

\subsubsection{Extended Anyons}
The Hamiltonian defined in \eqref{def:HN} is too singular to be an operator acting on a pure tensor product state $u^{\otimes N} \in{}\bigotimes^N_{\mathrm{sym}} L^2(\R^2)$. In order to circumvent this problem, we introduce a length $R$ over which the equivalent magnetic flux is smeared. In our approach, $R=f(N)$ for some function $f(N)\to 0$ when $N\to\infty$. In that limit, we will recover the point-like anyon system.
We consider the 2D Coulomb potential generated by a unit charge smeared over the disc of radius $R$:
\begin{equation}\label{wr}
	w_{R}(x)=\left(w_0\ast \frac{\one_{B(0,R)}}{\pi R^2}\right)(x),\:\:\text{with }\:\:w_{0}(x)=\log |x|.
\end{equation}
We remark that the Aharonov--Bohm magnetic
potential and field may be recovered as
\begin{align*}
	\nabla^{\perp}w_{0}(x)=\frac{x^{\perp}}{|x|^{2}},\quad\:\text{and}\quad\:\mathbf{B}_{0}(x)=\nabla^{\perp}\cdot\nabla^{\perp}w_{0}=\Delta w_{0}=2\pi\delta_{0}.
\end{align*}
We can then define the regularized Hamiltonian
\begin{equation}
	\HNR=\sum_{j=1}^{N}(-\im \nabla_j+\alpha \bAR_j)^{2}\quad\text{where}\quad  \bAR(x_{j}):=\sum_{k\neq j}\nabla^{\perp}w_{R}(x_{j}-x_{k})
	\label{def:HRN}
\end{equation}
acting on $L^{2}_{\mathrm{sym}}(\mathbb{R}^{2N})$ as a self-adjoint operator with the same form domain as the non-interacting bosonic Hamiltonian where $\bAR=0$ due to the boundedness of the interaction $\bAR$ for $R>0$; see \cite{AvrHerSim,ReedSimonII1975}. We may also use the formulation
\begin{equation}
	\nabla^{\perp}w_{R}(x_{j}-x_{k})=\frac{(x_{j}-x_{k})^{\perp}}{|x_{j}-x_{k}|_R^2}\quad\text{where}\quad|x|_R=\max(R,|x|).
\end{equation}
Notice that in this picture, $R$ can be interpreted as the radius of the particles. In the limit $R=f(N)\to 0$, any specific choice of $f$ may a priori define a specific type of extended anyons \cite{Co_Fer_20,Am_Cam_Bak_95}. These different gases may, a priori, have different mean-field limits \cite{Vis_25,AEGGLN_26,Ata_Lun_Gir_25}. An interesting threshold is $R={N}^{-\frac12}$ which corresponds to the average distance between anyons in the plane. In the present paper, we consider a slower convergence rate $R=\rateR$. See Remark~\ref{rem:R} for more details on that point. 

\subsubsection{The Mean-Field Limit of Almost-Bosonic Anyons}

By expanding \eqref{def:HRN}, we interpret our anyons as an interacting bosonic system, described by the following Hamiltonian, which includes a two-body and a three-body term:
\begin{align}
	\HNR=&\sum_{j=1}^{N}-\Delta_{j}  \;\tag{Kinetic term}\nonumber\\
	&+\alpha\sum_{j\neq k}\left((-\im\nabla_{j})\cdot\nabla^{\perp}w_{R}(x_{j}-x_{k})+\nabla^{\perp}w_{R}(x_{j}-x_{k})\cdot(-\im\nabla_{j})\right)\;\tag{Mixed two-body term}\nonumber\\
	&+\alpha^{2}\sum_{j\neq k\neq l}\nabla^{\perp}w_{R}(x_{j}-x_{k})\cdot\nabla^{\perp}w_{R}(x_{j}-x_{l})\;\tag{Three-body term}\nonumber\\
	&+\alpha^{2}\sum_{j\neq k}\left|\nabla^{\perp}w_{R}(x_{j}-x_{k})\right|^{2}\;\tag{Singular two-body term}.\\
	\label{expanded_H}
\end{align}

Notice here that our gauge formulation of the anyon gas presents a two-body term mixing momentum and position.
The mean-field limit is defined by taking 
\begin{equation}
	\alpha:=\frac{\beta}{N-1}\nonumber
\end{equation}
where $\beta$ is a real number. 
As shown in \cite[Theorem 1.1]{Gir_20}, for some trapping potential $V := \sum_{j=1}^N v(x_j)$ and for all $0 < \eta < 1/4$ such that $R = N^{-\eta}$, the minimizer of $\HNR + V$ is well approximated by the state $u_{\mathrm{af}}^{\otimes N}$, where the function $u_{\mathrm{af}} \in L^2(\R^2)$ minimizes the average-field energy
\begin{equation}\label{def:Eaf}
	\mathcal{E}_R^{\mathrm{af}}[u]:=\int_{\R^2}|(-\im\nabla +\beta \bAR[|u|^2])u|^2 +\int_{\R^2}v|u|^2
\end{equation}
under the constraint $\norm{u_{\mathrm{af}}}_{L^2(\R^2)}=1$ and where $\bAR[|u|^2]:=\nabla^{\perp}w_R\ast |u|^2$. In the above, we omitted the measure. We will omit it whenever it coincides with the Lebesgue measure.
Moreover, it holds that 
\[
N^{-1}\inf_{L^{2}_{\mathrm{sym}}(\mathbb{R}^{2N})} (\HNR+V)\to \inf_{\norm{u}_{L^2(\R^2)}=1}\mathcal{E}_0^{\mathrm{af}}[u]
\quad\text{as}\quad
N\to +\infty.
\]
This is comparable to Hartree theory but with a self-generated magnetic field $\bAR[|u|^2]$ instead of a smeared two-body electric field $V_e\ast |u|^2$.

\subsection{The Dynamics}
\subsubsection{The Hartree Dynamics}

If the potential $V$ is removed at $t = 0$, letting the system evolve, the dynamics of the $N$ almost-bosonic-extended anyons are governed by the Schrödinger equation
\begin{equation}\label{def:schro}
	\ii \partial_t \Psi_N(t) = \HNR \Psi_N(t).
\end{equation}
If $\Psi_N(0)\in L^2(\R^{2N})$ is the initial state of the system at time $t=0$, we can write
\begin{equation}
	\Psi_N(t)=e^{-\im t \HNR }\Psi_N(0).\nonumber
\end{equation}
Let $\sket{f}\sbra{f}$ denote the orthogonal projection onto $\operatorname{span}\{f\}$. The Schrödinger equation can also be written in terms of the density matrix and takes the form
\begin{align}\label{def:Schro2}
	\im \partial_{t}\gamma_{N}(t)=\left [\HNR,\gamma_{N}(t)\right ]\quad \text{where}\quad\gamma_N(t)=|\Psi_N(t)\rangle\langle\Psi_N(t)|
\end{align}
with the associated $k$-particle reduced density matrices defined as
\begin{equation}
	\gamma_{N}^{(k)}(t):=\Tr_{k+1,\dots, N} \left [\gamma_{N}(t)\right ].\nonumber
\end{equation}
Our aim is to show that if the initial datum of \eqref{def:schro} is a product state $\Psi_N(0)=\varphi_0^{\otimes N}$ for a given $\varphi_0\in L^2(\R^2)$, at later time, the state $\Psi_N(t)=e^{-\im \HNR t}\Psi_N(0)$ can be approximated by an uncorrelated state $\varphi_t^{\otimes N}$ where $\varphi_t$ is driven by a one-body nonlinear equation with initial data $\varphi_0$.

To infer what could be the one-body effective equation, we can write the BBGKY hierarchy of the nonlinear system, i.e., expand \eqref{def:Schro2} in terms of $\gamma_N^{(k)}(t)$ and formally take the limit 
\[
N\to \infty ,\quad\alpha=\beta(N-1)^{-1}\to 0,\quad R\to 0.
\]
Under the assumption that, in this limit, $\gamma_N^{(k)}(t)\to |\varphi(t)\rangle\langle\varphi(t)|^{\otimes k}$ we obtain that $\varphi_t$ is driven by the Chern--Simons--Schr\"odinger equation defined in Definition~\ref{def:css}. A detailed formal calculation of the BBGKY hierarchy is presented in Appendix~\ref{app:BBGKY}.
In our case, showing the existence of a solution $\gamma_{\infty}(t) \simeq \varphi_t^{\otimes (N=\infty)}$ without an explicit convergence rate to control the inverse powers of $R$ is a challenging task that we cannot accomplish using only the general BBGKY approach.

\subsubsection{The Chern--Simons--Schr\"odinger Equation}

Another way to guess what should be the correct effective equation is to compute
\begin{align}
	\im \partial_{t}u &=\left.\frac{\delta \mathcal{E}_R^{\mathrm{af}}[v,\overline{v}]}{\delta\overline{v}}\right|_{u}. \nonumber
\end{align}
In the case of the electric Hartree equation, the calculation gives
\begin{align}
	\im \partial_{t}u &=\left.\frac{\delta (\int |\nabla v|^2 +(V_e\ast |v|^2) |v|^2)}{\delta\overline{v}}\right|_{u}=-\Delta u +2(V_e\ast |u|^2) u. \nonumber
\end{align}
In the present magnetic case, the nonlinearity of $\bAR[|u|^2]$ in \eqref{def:Eaf} is squared.
Carrying out the calculation, we derive the following equation.
\begin{definition}[\textbf{Chern--Simons--Schr\"odinger equation}]\label{def:css}
	We denote by $\text{CSS}(R,u_0)$ the differential problem whose unknown is     a function $u:(\R_+\times \R^2)\mapsto \mathbb{C}$ satisfying 
	\begin{align}\label{eq:pilotu}
		\im \partial_{t}u 
		&=\left (-\im\nabla +\beta \bAR\left [| u|^{2}\right ]\right )^{2}u\nonumber\\
		&\quad- \beta\left [\nabla^{\perp}w_{R}\ast\left (2\beta \bAR\left [| u|^{2}\right ]| u|^{2}+\im\left (u\nabla\overline{u}-\overline{u}\nabla u \right ) \right )\right ]u
	\end{align}
	with initial condition $u(0,x)=u_0(x)\in H^2(\R^{2})$.
	We define $\text{CSS}(u_0):=\text{CSS}(R=0,u_0)$.
\end{definition}

The well-posedness of $\text{CSS}(R,u_0)$ is discussed in Section~\ref{sec:css}. General references on this topic can be found in \cite{Ber_Bou_Sau_95,Liu_Smi_13,Liu_Smi_Tat_14,Oh_Pus_13}, 
with particular attention given to its connection with the fractional quantum Hall effect in \cite{Esp_18,Raj_Sig_20}.
For later convenience, we denote $\mathbf{J}\left [u\right]:=\im(u\nabla \overline{u} -\overline{u}\nabla u)$.

\begin{remark}
	Note that $\text{CSS}(u_0)$ is sometimes written in a gauge-covariant form:
	\begin{equation}\label{eq:CSSform}
		\left\{
		\begin{aligned}
			\bD_{t} u &=\; \ii \bD_{\ell} \bD_{\ell} u + \ii g \abs{u}^{2} u, \\
			F_{12} &= - \frac{1}{2} \abs{u}^{2}, \\
			F_{01} &=  - \frac{\ii}{2} ( u \overline{\bD_{2} u} - (\bD_{2} u) \overline{u} ), \\
			F_{02} &=\; \frac{\ii}{2} ( u \overline{\bD_{1} u} - (\bD_{1} u) \overline{u} )
		\end{aligned}
		\right.
	\end{equation}
	except that in our case $g=0$.
	Here, $A = A_{0} \d x^{0} + A_{1} \d x^{1} + A_{2} \d x^{2}$ is a real-valued 1-form (gauge potential), $F_{\mu\nu} = \partial_{\mu} A_{\nu} - \partial_{\nu} A_{\mu}$ (curvature 2-form), $\bD_{\mu} = \partial_{\mu} + \ii A_{\mu}$ (covariant derivative), and $g \in \bbR$. 
	Here, by choosing the Coulomb gauge, $\partial_1 A_1 + \partial_2 A_2 = 0$, Eq. \eqref{eq:CSSform} can be reformulated as
	\begin{equation}\label{eq:CSS-Coulomb-gauge}
		\left\{
		\begin{aligned}
			\ii \partial_t u &=  - \Delta u + A_0 u - 2\ii \sum_{\ell=1}^{2} A_\ell \partial_\ell u + \sum_{\ell=1}^{2} A_\ell A_\ell u - g |u|^2 u, \\
			\Delta A_1 &= \frac{1}{2} \partial_2 |u|^2, \\
			\Delta A_2 &= -\frac{1}{2} \partial_1 |u|^2, \\
			\Delta A_0 &= \frac{\ii}{2} \partial_1 (u \overline{\bD_2 u} - (\bD_2 u) \overline{u}) - \frac{\ii}{2} \partial_2 (u \overline{\bD_1 u} - (\bD_1 u) \overline{u}).
		\end{aligned}
		\right. \,
	\end{equation}
	By noting that the fundamental solution of the Laplacian in $\R^2$ is the Coulomb potential, we see that \eqref{eq:CSS-Coulomb-gauge} is equivalent to \eqref{eq:pilotu} with $\beta=1$, $R=0$, and the extra nonlinear term $-g|u|^2u$. We refer to \cite{Ata_24,Lim_Min_17} for recent results considering $g\neq 0$. We also notice that the Lagrangian formulation of this equation corresponds to a usual Chern--Simons equation coupled to the density $|u|^2$, see for instance \cite[Eq. (1.4)]{JP1990PRD}.
\end{remark}

\subsection{Main Theorem}
Before stating the main theorem of the paper, we explicitly outline the following important assumption.
\begin{assumption}\label{assumption}
	Throughout the paper, we assume that the initial data satisfies $\varphi_0 \in H^2(\mathbb{R}^2)$ and that the parameter $|\beta|$ is bounded by a constant $c > 0$, which depends on $\|\varphi_0\|_{H^1}$ but not on $\|\varphi_0\|_{H^2}$. These conditions will ensure that $\|\varphi_t\|_{H^2}\leq C(t)$ which is a crucial quantity our method cannot avoid.
\end{assumption}

\begin{theorem}[\textbf{Convergence from the many-body Schr\"odinger dynamics to CSS}.]\mbox{}
\label{thm:main}
	Let $\varphi_t$ be the solution of $\mathrm{CSS}(\varphi_0)$ defined in \eqref{eq:pilotu} with initial data $\varphi_0 \in H^2(\mathbb{R}^2)$. 
	Let $\Psi_N(t)$ denote the solution of the Schrödinger equation \eqref{def:schro} with Hamiltonian $\HNR$ for any $R\to 0$ such that $R\gtrsim (\log N)^{-\frac{1}{2}+\varepsilon}$ and $\Psi_N(0) = \varphi_0^{\otimes N}$. 
	Choose any $k \in \mathbb{N}$, and let $\gamma_N^{(k)}$ denote the $k$-particle reduced density matrix associated with $\Psi_N(t)$. 
	Then, there exist constants $c, T>0$ depending on $\|\varphi_0\|_{H^1}$ and $C > 0$  depending on $\|\varphi_0\|_{H^2}$ and $\|\varphi_0\|_{H^1}$ such that for any $|\beta| \leq c$ and for all time $0 \leq t \leq T$, we have, for sufficiently large $N$,
	\begin{equation}\label{eq:main}
		\Tr\left|\gamma_N^{(k)}(t) - \sket{\varphi_t^{\otimes k}} \sbra{\varphi_t^{\otimes k}}\right| 
		\leq C (\log N)^{-\frac{1}{2}+\varepsilon}
	\end{equation}
	for any choice of $\varepsilon >0$.
\end{theorem}

The result is new in several aspects. This is, as far as we know, the first time that the CSS equation is derived from a many-body problem and more generally, that the dynamics of anyonic particles is studied. Although similar results already exist for bosons with three-body interaction, for instance in \cite{ChenPavlovic2011JFA,XChen201ARMA,Lee_21,NS2020CMP}, our treatment of the non-positive mixed two-body term is entirely new. We also bring a modest improvement to well-known bounds for anyonic terms, such as those in Corollary~\ref{cor:nabwRLp} and Lemma~\ref{lem:A-infty}, \ref{lem:sing2bd} and \ref{lem:mix2bd}. We discuss more specific points in the following remarks.

\begin{remark}
	In \cite[Theorem 2.1]{Ber_Bou_Sau_95}, the well-posedness of the CSS equation is established in $H^2(\R^2)$ for times up to some $T'$ due to the existence of blow-up, \cite[Theorem 3.1]{Ber_Bou_Sau_95}. 
	In our case, $g = 0$, and blow-up might not occur. 
	Nevertheless, in our approach, we need to control $\snorm{\varphi_t}_{H^2}$ and only manage to bound it uniformly in $R$ on finite time intervals or must allow it to diverge in $R$ on unbounded time intervals; see also the next Remark~\ref{rem:beta}. 
	The CSS equation is also known to be globally well-posed for small $H^1(\R^2)$ initial data, \cite[Theorem 2.3]{Ber_Bou_Sau_95} or for zero energy initial data as shown in \cite[Theorem 1.2]{Ata_24}. 
\end{remark}

\begin{remark}\label{rem:beta}
	The assumption on $\beta$ is necessary to obtain the uniform control $\snorm{\varphi_t^R}_{H^2} \leq C$ in Theorem~\ref{thm:wellpose} for all $R$. 
	It is important to note that this uniform control in $R$ only holds for finite times $0 \leq t < T$. 
	Without this assumption, when considering all $\beta \in \R$, the best control we can achieve is $\snorm{\varphi_t^R}_{H^2} \leq R^{-Ct}$, as shown in Lemma~\ref{thm:u-GWP}, which is valid for all $t \geq 0$. 
	However, this bound diverges (as $R\to 0$) too rapidly for our proof to apply. 
	The proof could still work with a control of the form $\snorm{\varphi_t^R}_{H^2} \leq C|\log R|$. 
	It is worth mentioning that, although Lemmas~\ref{lem:Vm1/2} and \ref{lem:kinetic} are stated with the constraint $|\beta| \leq c$, their proofs could be adapted to avoid any dependence on $\snorm{\varphi_t^R}_{H^2}$ and remain valid for all $\beta$. 
	The only additional effort would come from the estimate \eqref{ine:v'}, where $\snorm{\varphi_t^R}_{H^2}$ could be replaced by $|\log R|$. 
	The crucial obstacle to relaxing the assumption on $\beta$ is located in the proof of Lemma~\ref{lem:Conv_pilot} that we cannot obtain without using $\snorm{\varphi_t^R}_{H^2}\leq C$.
\end{remark}

\begin{remark}
	The initial data factorization assumption $\Psi_N(0)=\varphi_0^{\otimes N}$ could be relaxed to include a class of almost condensed states $\Psi_N(0)\to\varphi_0^{\otimes N}$ as $N\to \infty$ in trace norm. Our choice was to simplify the exposition. 
\end{remark}

\begin{remark}
	We could certainly include a repulsive two-body interaction. This potential extension is discussed in Appendix~\ref{sec:twobody}.
\end{remark}

\begin{remark}
	It is likely that we could also include an external magnetic interaction by deriving bounds similar to those in \cite[Lemma 2.3]{Gir_20} to replace the estimates in Lemma~\ref{lem:mix2bd}. 
	The key challenge would then be located in Lemma~\ref{lem:kinetic}, where it would be necessary to control the evolution of the projection of the magnetic gradient operator, rather than the pure gradient.
\end{remark}

\subsection{Strategy of the Proof}
\subsubsection{Structure of the Paper}\label{sec:strategy}
The idea of the proof follows the techniques developed in \cite{KnoPic-10,Pickl2011}. We will stick to the notation introduced in \cite{KnoPic-10}, where the indicators of the particles outside the condensate are measured using the quantities
\begin{equation}\label{def:EandR}
	\mathcal N_+^{(k)}(t):=1-\braket{\varphi_t^{\otimes k}}{\gamma_N^{(k)}(t)\varphi_t^{\otimes k}}\quad\text{and}\quad R_N^{(k)}(t):=\Tr\left|\gamma_N^{(k)}(t)-\sket{\varphi_t^{\otimes k}}\sbra{\varphi_t^{\otimes k}}\right|.
\end{equation} 
The method relies on controlling $R_N^{(k)}(t)$ using that
\begin{equation}\label{ine:EAB}
	(R_N^{(k)}(t))^2\lesssim \mathcal N_+^{(k)}(t)\lesssim k\mathcal N_+(t)
\end{equation}
as shown in \cite[Lemma 2.1 and Lemma 2.3]{KnoPic-10} where we have denoted $\mathcal N_+(t):=\mathcal N_+^{(1)}(t)$\footnote{Note that we can also write $\mathcal N_+(t)=N^{-1}(N-\langle a^*_{\varphi_t}a_{\varphi_t}\rangle_{\Psi_N(t)})$ where  $a^*$ and $a$ are the usual creation and annihilation operators.}.
The control of $\mathcal N_+(t)$ relies on Gr\"onwall's lemma~\ref{lem:Gron} providing
\begin{equation}\label{Gron1}
	\mathcal N_+(t)\lesssim \mathcal N_+(0)+\mathcal K_{+}(t)+o_N(1)
\end{equation}
where $\mathcal K_{+}(t)$ is the kinetic energy of the excited particles properly defined and controlled in Lemma \ref{lem:kinetic}. In this lemma, the control of $\mathcal K_{+}(t)$ is made through
\begin{equation}\label{kin}
	\mathcal K_+(t)\lesssim \sqrt{ \mathcal N_+}+o_N(1)
\end{equation}
where $\sqrt{ \mathcal N_+}(t)$ is defined in \eqref{def:N+}. We conclude the proof by controlling $\sqrt{ \mathcal N_+}(t)$ in Theorem \ref{thm:K}
\begin{equation}\label{Gron2}
	\sqrt{\mathcal N_+}(t)\lesssim \sqrt{\mathcal N_+}(0)+\mathcal K_{+}(t)+o_N(1)
\end{equation}
which can be closed using of \eqref{kin}. Using that $\mathcal N_+(t)\leq \sqrt{\mathcal N_+}(t)$ we realize that the first Gr\"onwall of \eqref{Gron1} is superfluous and the proof only consists of proving \eqref{Gron2} and \eqref{kin}.

To make the one-body time evolution meaningful, the time $t$ will be bounded by a given $T$ until which $\text{CSS}(\varphi_0)$ stays well-posed as shown in Theorem~\ref{thm:wellpose}. 
A difficult part of the calculation consists in controlling the divergences in the parameter $R$ appearing everywhere. In Section~\ref{sec:bounds}, we derive a collection of estimates to control the different operators of $\HNR$ defined in \eqref{expanded_H}. The bounds are Hardy-type inequalities and provide a control in terms of $|\log R|(1-\Delta)$, i.e., a slightly diverging kinetic energy plus a constant. That section also contains bounds of several quantities appearing later in the calculation. In Section~\ref{sec:css}, we establish the well-posedness of $\text{CSS}(R,\varphi_0)$ as well as the conservation of the energy. This allows us to control $\norm{\varphi_t}_{H^1}$ in terms of $\norm{\varphi_0}_{H^1}$ and to control $\norm{\varphi_t}_{H^2}$ in two different ways,     by a constant $C$ independent of $R$ on finite time intervals in Lemma~\ref{thm:wellpose} and by Theorem~\ref{thm:u-GWP}. 
The derivation of \eqref{ine:EAB} is carried out in Section~\ref{sec:pickl}. In that calculation, most of the $-\nabla $ can be moved to act on a $\varphi_t$ and provide an $\norm{\varphi_t}_{H^1}\leq C$ via the bounds of Section \ref{sec:bounds}. Some do not, which requires to estimate the time growth of the kinetic energy of the non-condensed part of $\Psi_N(t)$, $\mathcal K_+(t)$. This is the content of Section~\ref{sec:kinetic} and is achieved via the second Gr\"onwall-type lemma. In the last section, Section~\ref{sec:con}, we gather all the intermediate results and conclude by proving Theorem~\ref{thm:main}.

\subsubsection{Notations and Basic Inequalities}
The main Hilbert space in our paper is $L^2_{\mathrm{sym}}(\R^{2N})$. For that reason, we will always mean $\norm{\Psi}_{L^2_{\mathrm{sym}}(\R^{2N})}$ when writing $\norm{\Psi}$ and the $L^p_{\mathrm{sym}}(\R^{2})$ norms will be denoted by $\norm{\varphi}_p$. For the Sobolev norms $H^p_{\mathrm{sym}}(\R^{2})$, we write $H^p$ in the subscript. The symbol $C$ will denote any finite strictly positive constant possibly depending on fixed parameters such as $\beta$, $\norm{\varphi_0}_{H^1}$ or $\norm{\varphi_0}_{H^2}$ but never on $R$, $N$, or quantities depending on them. It can be different from line to line. We will also use the notation
\begin{equation}
	\braket{\Psi_N}{A\Psi_N}:=\left\langle A\right\rangle_{\Psi_N}\nonumber
\end{equation}
for self-adjoint operators $A$. As above, we often drop the $t$ variable in $\Psi_N(t)$ and $\gamma^{(k)}_N(t)$ or even $\varphi_t$.
As we have already used to describe \eqref{def:Schro2} and the main theorem, we use the bra-ket representation $\sket{f}\sbra{f}$ of the projection operator onto a function $f$.

For an operator $A$, we denote by $\norm{A}_\op$ and $\norm{A}_\HS$ the operator norm and the Hilbert-Schmidt norm, respectively.

\bigskip

All along the paper, we will make use of well-known inequalities %in Appendix \ref{app:Inq} 
that can be found in the literature, e.g. \cite[Section 8.1, Eq. (3)]{LieLos-01} and \cite[Theorem 4.3]{LieLos-01}, respectively.
\begin{lemma}[\textbf{Sobolev inequality}]
	Let $f\in H^1(\R^2)$. Then, for any $2\leq q <\infty$, we have that
	\begin{equation}\label{eq:S}\tag{S}
		C_q^{-1}\norm{f}_q^2\leq \norm{f}_2^2 +\norm{\nabla f}_2^2
	\end{equation}
	where
	\begin{equation}
		C^{-1}_q\leq\left(q^{1-\frac{2}{q}}(q-1)^{\frac{1}{q}-1}\left(\frac{q-2}{8\pi}\right)^{\frac{1}{2}-\frac{1}{q}}\right)^{2}.\nonumber
	\end{equation}
	Note that, for large $q$, we have
	\begin{equation}\label{ine:Cq}
		C_q \leq 2\pi q^{1-\frac{2}{q}}.
	\end{equation}
\end{lemma}
This inequality will be used systematically to control $\snorm{\varphi_t}_q$ and $\snorm{\nabla\varphi_t}_q$, when $2<q<\infty$. We will sometimes consider $q\simeq |\log R|$ and will then have to keep track of $C_q$.

\begin{lemma}[\textbf{Sobolev inequality for $q=\infty$}]\label{lem:sobinfty}
	Let $f\in H^{s}(\R^2)$ with $s>1$, then there exist a constant $C$ depending on $s$ but not $f$ such that
	\begin{equation}
		\snorm{f}_\infty \leq C \snorm{f}_{H^s}^{\frac{1}{s}} \snorm{f}_2^{1-\frac{1}{s}}\leq C \snorm{f}_{H^s} .
	\end{equation}
\end{lemma}
The proof can be found in \cite[Theorem 5.8, 5.9]{AF2003Sobolev}.
We also need to define the weak $L^p$ norms denoted $L^{p}_{w}(\R^2)$ as
\begin{equation}
	\snorm{f}_{p,w}:=\sup_{t>0}\left[t \left(\mu \{x\in \R^2 : |f(x)|>t \} \right)^{\frac{1}{p}}\right]
\end{equation}
where $\mu\{ X\}$ is the Lebesgue measure of the measurable set $X\subset \R^2$.

\begin{lemma}[\textbf{Weak-Young inequality}]
	Let $f\in L^p(\R^2)$ and $g\in L^{r}_{w}(\R^2)$, we have for any ${1< p,q,r <\infty}$ that
	\begin{equation}\label{eq:WY}\tag{WY}
		\snorm{g*f}_q\leq C_{p,q,r}\snorm{g}_{r,w}\snorm{f}_p
	\end{equation}
	provided that
	\[
	1+\frac{1}{q}=\frac{1}{p}+\frac{1}{r}.
	\] 
\end{lemma}

Note that Young's convolution inequality is identical to \eqref{eq:WY} except that the former allows the endpoint valuess $1\leq p,q,r \leq\infty$. Also note that 
\begin{equation}
	\snorm{\nabla^{\perp}w_R}_{2,w}\leq\snorm{\nabla^{\perp}w_0}_{2,w}\leq C.
\end{equation}
For that reason, throughout the paper, $\snorm{\nabla^{\perp}w_R}_{2,w}$ will be considered as a constant. As already mentioned, we will make use of Gr\"onwall's lemma.

\begin{lemma}[\textbf{Gr\"onwall Lemma}]\label{lem:Gron}
	Let $f$ be a function satisfying $|\partial_t f(t)|\leq C(t)f(t) + \varepsilon(t)$, then
	\begin{align}\label{eq:Gron}
		f(t)\leq f(0)\,e^{\int_0^t C(\tau)\dd \tau} + \int_0^t \varepsilon(s) \, e^{\int_s^t C(\tau) \dd \tau} \dd s.
	\end{align}
\end{lemma}

\section{Useful Properties}\label{sec:bounds}

\subsection{Operator Bounds for the Interaction Terms}

We gather some intermediate estimates in the following lemmas. 
The main feature of this section is the appearance of the $|\log R|$ divergences appearing in Lemmas~\ref{lem:sing2bd},  \ref{lem:mix2bd} and Corollary \ref{cor:nabwRLp}. These results provide modest improvements of \cite[Lemmas 2.1, 2.2, and 2.3]{LunRou}. Additional results, such as Lemmas~\ref{lem:estnew} and \ref{lem:A-infty},  will be needed when discussing the well-posedness of the $\mathrm{CSS}$ equation in Section~\ref{sec:css}. 
In the remainder of this section, we bound several quantities appearing in the core of the proof in Section~\ref{sec:pickl} and \ref{sec:kinetic}. 

\begin{lemma}\label{lem:estnew}
	For any $u,v \in H^{1}(\bbR^2)$ and any $R\geq 0$, the following bounds hold
	\begin{equation}\label{eq:N4}
		\norm{ \bAR\left [| u|^{2}\right ]}_p\leq C
	\end{equation}
	as well as
	\begin{equation}\label{eq:diff4}
		\norm{(\bA\left [| u|^{2}\right ]-\bAR\left [| v|^{2}\right ])}_p\leq CR +C\norm{u-v}_2
	\end{equation}
	for any $2< p <\infty$. Note that the above $C$ depend on $\norm{u}_{H^{1}}$ and $\norm{v}_{H^{1}}$.
\end{lemma}

We note that throughout this paper, Lemma~\ref{lem:estnew} is used exclusively in the case  $s = 0$. However, for possible future applications, we present the lemma in a more general form than is strictly required here. 

\begin{proof}
	We first have via weak Young inequality \eqref{eq:WY}, that for $p>2$:
	\begin{equation}
		\norm{ \bAR\left [|u|^{2}\right ]}_p=\norm{\nabla^{\perp}w_R \ast \vert u\vert^2}_p\leq C\norm{\nabla^{\perp}w_0}_{2,w}\norm{ |u|^2}_{\frac{2p}{p+2}}\leq C\norm{u}_{H^{1}(\bbR^2)}
	\end{equation}
	using Sobolev inequality (S).
	For the second norm we write
	\begin{align}
		\bA\left [| u|^{2}\right ]-\bAR\left [| v|^{2}\right ]&=(\nabla^{\perp}w_0 -\nabla^{\perp}w_R)\ast\vert v\vert^2+\nabla^{\perp}w_0\ast(\vert u\vert^2 -\vert v\vert^2)\nonumber\\
		&=(\nabla^{\perp}w_0 -\nabla^{\perp}w_R)\ast\vert v\vert^2+\nabla^{\perp}w_0\ast((u-v)\overline{u} +(\overline{u}-\overline{v})v).
	\end{align}
	We apply the triangle inequality and Young inequality \eqref{eq:WY} for the convolution to get
	\begin{align}
		&\norm{(\bA\left [| u|^{2}\right ]-\bAR\left [| v|^{2}\right ])}_p\nonumber\\ 
		&\leq \norm{\left [(\nabla^{\perp}w_0 -\nabla^{\perp}w_R)\ast\vert v\vert^2\right ] }_p\nonumber\\
		&\quad +\norm{(\nabla^{\perp}w_0\ast (u-v)\overline{u})}_p+\norm{(\nabla^{\perp}w_0\ast(\overline{u}-\overline{v})v)}_p \nonumber\\
		&\leq \norm{\nabla^{\perp}w_0 -\nabla^{\perp}w_R }_1\norm{|v|^2}_p \nonumber\\
		&\quad +\norm{\nabla^{\perp}w_0}_{2,w}(\norm{(\overline{u}-\overline{v})v}_{\frac{2p}{p+2}}+\norm{(u-v)\overline{u}}_{\frac{2p}{p+2}})\nonumber\\
		&\leq \norm{|v|^2}_{p}\int_{B(0,R)}|x|^{-1}\mathrm{d}x +\norm{\nabla^{\perp}w_0}_{2,w}\norm{u-v}_{2}(\norm{v}_p+\norm{u}_{p})\nonumber\\
		&= C\, \norm{v}_{H^{1}}R+C\, (\norm{u}_{H^{1}}+\norm{v}_{H^{1}})\norm{u-v}_{2}.
	\end{align}
	This gives the desired lemma.
\end{proof}

\begin{lemma}[\textbf{The smeared Coulomb potential}]\label{lem:smCp}
	Let $w_R$ be defined as in \eqref{wr}. There is a constant $C > 0$ such that
	\begin{equation}\label{eq:sup-nabwR}
		\sup_{\bbR^2} |\nabla w_R| \leq \frac{C}{R}. \quad
	\end{equation}
	Moreover, for any $2<p<\infty$,
	\begin{equation}\label{eq:nabwRLp}
		\norm{\nabla w_R}_{p} \leq \left(\frac{4\pi p}{p^2-4}\right)^{\frac1p} R^{\frac{2}{p}-1}.
	\end{equation}
\end{lemma}
\begin{proof}
	Eq. \eqref{eq:sup-nabwR} follows from Newton's theorem as established in \cite[Lemma 2.1]{LunRou}.
	To obtain \eqref{eq:nabwRLp}, we compute
	\begin{align*}
		\norm{\nabla w_R}_{p}^p &= 2\pi \int_0^R \frac{r^p}{R^{2p}} r \, \d{r} + 2\pi \int_R^\infty r^{-p} r \, \d{r} \\
		&= 2\pi \left( \frac{1}{R^{2p}} \frac{1}{2+p}R^{2+p} + \frac{1}{p-2} R^{2-p}\right)
		= \frac{4\pi p}{p^2 - 4} R^{2-p}.
	\end{align*}
\end{proof}

\begin{corollary}\label{cor:nabwRLp}
	For $p(R)=2+{|\log R|}^{-1}$ with $0<R<e^{-1}$, we have
	\begin{equation}\label{eq:log-nabwRLp}
		\norm{\nabla w_R}_{p(R)} \leq C \sqrt{ |\log R| }.
	\end{equation}
\end{corollary}
\begin{proof}
	The proof follows by substituting $p(R) = 2 + |\log R|^{-1}$ into \eqref{eq:nabwRLp}, using the monotonicity of the logarithmic function, and noting that \(x^{(\log x)^{-1}} = e\) for all \(x > 0\). Additionally, the monotonically increasing property of \(a^x\) with respect to \(x\) for all \(a > 1\) is applied.
\end{proof}

\begin{lemma}\label{lem:A-infty}
	For any $u\in H^1(\R^2)$ and any $0 < R < e^{-1}$, the following holds
	\begin{align}
		\norm{\bAR[|u|^2]}_{\infty}
		&\leq C \sqrt{|\log R|}\norm{u}_{H^1}^2.
	\end{align}
\end{lemma}

\begin{proof}
	By using Young inequality with $\varepsilon=|\log R|^{-1}$, combined with the bound \eqref{ine:Cq} for the constant of Sobolev embedding, and Corollary \eqref{cor:nabwRLp}, one gets
	\begin{align}
		\norm{\bAR[|u|^2]}_{\infty}&=\norm{\nabla^{\perp}w_R\ast|u|^2}_{\infty}
		\leq \norm{\nabla^{\perp}w_R}_{2+\varepsilon}\norm{|u|^2}_{\frac{2+\varepsilon}{1+\varepsilon}}
		\leq C \sqrt{|\log R|} \norm{u}_{H^1}^2.
	\end{align}
\end{proof}

\begin{lemma}[\textbf{Singular two-body term}]\label{lem:sing2bd}
	We have that, as operators on $L^2(\bbR^4)$ or $L^2_{\mathrm{sym}}(\bbR^4)$,
	\begin{equation}\label{eq:sig2bd}
		|\nabla w_R(x-y)|^2 \leq C |\log R|^2 \left( 1 - \Delta_x \right)
	\end{equation}
	for any $0<R<e^{-1}$.
\end{lemma}

\begin{proof}
	Here we repeat the proof of \cite[Lemma 2.2]{LunRou} and only modify the last step. We start with an application of Hölder's and Sobolev's inequalities.
	For any $W : \bbR^2 \to \bbR$ and $f \in C_c^\infty(\bbR^4)$
	\begin{align*}
		\langle f , W(x - y) f \rangle 
		&= \int_{\bbR^2 \times \bbR^2} \overline{f(x, y)} W(x - y) f(x, y) \, \d{x} \d{y} \\
		&\leq \|W\|_{p} \int_{\bbR^2} \left( \int_{\bbR^2} |f(x, y)|^{2q} \d{x} \right)^{1/q} \d{y} \\
		&\leq C_q \|W\|_{p} \int\int_{\bbR^2 \times \bbR^2} \left( |f(x, y)|^2 + |\nabla_x f(x, y)|^2 \right) \d{x} \d{y} \\
		&= C_q \|W\|_{p} \langle f | (1-\Delta_x) \otimes {\bf 1} | f \rangle
	\end{align*}
	where we take any $p > 1, \, q = \frac{p}{p-1} \in (1, +\infty)$.
	Finally, we plug $W=|\nabla w_R|^2$, $p=1+\varepsilon$, $q=\frac{p}{p-1}=\frac{1+\varepsilon}{\varepsilon }$ with  $\varepsilon=\frac{1}{2}|\log R|^{-1}$, and we use \eqref{ine:Cq} for $C_q$ and \eqref{eq:log-nabwRLp} for $\|W\|_p$ to obtain
	\begin{align}
		C_{q}\|W\|_{p}&
		\leq C_{q}\|\nabla w_{R}\|_{2p}^{2}
		\leq C |\log R|^2.
	\end{align}
\end{proof}

\begin{remark}
	Lemma \ref{lem:sing2bd} is optimal. The counter example is in Appendix \ref{app:counterexample}.
\end{remark}

\begin{lemma}[\textbf{Mixed two-body term}]\label{lem:mix2bd}
	For any $0< R < e^{-1}$, we have that
	\begin{align}
		\left( \nabla_x \cdot \nabla^\perp w_R(x-y) + \nabla^\perp w_R(x-y)\cdot \nabla_x \right)^2
		&\leq  \left\{
		\begin{array}{ll}
			C {|\log R|}^2  (1-\Delta_x)^2 \\
			C {|\log R|}^2 (1 -\Delta_x)(1-\Delta_y) \label{eq:mix2-2}
		\end{array}
		\right. 
	\end{align}
	as operators on $L^2_{\mathrm{sym}}(\bbR^2 \times \bbR^2)$.
\end{lemma}

\begin{proof}

	Note that the first inequality of the lemma can be found in \cite[Lemma 2.5]{LunRou}. We use the same approach to derive the second one. We start by noticing that
	\begin{equation}\label{eq:commutator}
		\nabla_x \cdot \nabla^{\perp} w_R (x-y) = \nabla^{\perp} w_R (x-y) \cdot \nabla_x
	\end{equation}
	because $\nabla_x \cdot \nabla^{\perp} w_R(x-y) = 0$. 
	We can then square the expression we want to estimate, obtaining 
	\[
	\left( \nabla_x \cdot \nabla^{\perp} w_R (x-y) + \nabla^{\perp} w_R (x-y)  \cdot \nabla_x \right) ^2 
	= 4 \nabla_x \cdot \nabla^{\perp} w_R (x-y) \nabla^{\perp} w_R (x-y)  \cdot \nabla_x.
	\]
	Consequently, for any $f = f(x,y) \in C^{\infty}_c(\R^4)$,
	\begin{align*}
		&\left| \left\langle f \big| \left( \nabla_x \cdot \nabla^{\perp} w_R (x-y) + \nabla^{\perp} w_R (x-y)  \cdot \nabla_x \right) ^2 \big| f \right\rangle \right| \\
		&= 4\left| \iint_{\R^2 \times \R^2} \left( \nabla_x \bar{f} (x,y) \cdot \nabla^{\perp} w_R (x-y) \right) \left(\nabla_x f (x,y) \cdot \nabla^{\perp} w_R (x-y) \right)\d{x}\d{y}\right| \\
		&\leq 4\iint_{\R^2 \times \R^2} \left| \nabla_x f (x,y)\right| ^2 \left| \nabla^{\perp} w_R (x-y) \right| ^2 \d{x}\d{y}. 
	\end{align*}
	Inserting \eqref{eq:sig2bd} with $ \left| \nabla^{\perp} w_R (x-y) \right| = \left| \nabla^{\perp} w_R (y-x) \right| $, we get 
	\begin{align*}
		&\left| \left\langle f \big| \left( \nabla_x \cdot \nabla^{\perp} w_R (x-y) + \nabla^{\perp} w_R (x-y)  \cdot \nabla_x \right) ^2 \big| f \right\rangle \right| 
		\leq C |\log R|^2 \iint_{\R^2 \times \R^2} ( \left| \nabla_x f \right| ^2 +  \left| \nabla_y \nabla_x f \right| ^2 ) \d{x}\d{y}
	\end{align*}
	and thus, adding $-\Delta_y \geq 0$, we get 
	\begin{equation*}
		\left( \nabla_x \cdot \nabla^{\perp} w_R (x-y) + \nabla^{\perp} w_R (x-y)  \cdot \nabla_x \right) ^2 
		\leq C {|\log R|}^2  (1 -\Delta_x)(1 -\Delta_y)
	\end{equation*}
	concluding the proof.
\end{proof}

We will also need the operator bound from \cite[Lemma 2.4]{LunRou} based on Hardy's inequality.
\begin{lemma}[\textbf{Three-body term}]\label{lem:3body}
	We have that, as operators on $L^{2}_{\mathrm{sym}}(\mathbb{R}^{6})$,
	\begin{equation}\label{eq:W3C}
		0\leq \nabla^{\perp}w_{R}(x-y)\cdot\nabla^{\perp}w_{R}(x-z)\leqslant C\left(1-\Delta_x\right).
	\end{equation}
\end{lemma}
This concludes the derivation of the main operator bounds of the operators contained in the Hamiltonian $\HNR$ defined in \eqref{expanded_H}. 
We now proceed to the analysis of the various quantities that play an important role in the proofs of Sections~\ref{sec:pickl} and \ref{sec:kinetic}.

\subsection{A Priori Bounds for Mean-Field Terms}

\begin{definition}\label{def:vw}
	For any function $\varphi_t\in H^2(\R^2)$, we define the operators:
	\begin{align}
		v'(x)&:=-\im \nabla_x\cdot \bAR\left [|\varphi_t|^{2}\right ](x
		)+\hc -(\nabla^{\perp}w_R\ast \mathbf{J}[\varphi_t])(x),\label{def:v'}\\
		w'(x)&:=\bAR\left [|\varphi_t|^{2}\right ]^2(x)+2(\nabla^{\perp}w_R\ast\bAR\left [|\varphi_t|^{2}\right ]|\varphi_t|^2)(x),\label{def:w'}\\
		v(x-y)&:= -\im \nabla_x\cdot\nabla^{\perp}w_R(x -y)+\hc, \label{def:v}\\
		w(x-y, x-z)&:=\nabla^{\perp}w_R(x -y)\cdot\nabla^{\perp}w_R(x-z)\label{def:w}
	\end{align}
	as well as the projectors $p(x):=|\varphi_t(x)\rangle\langle\varphi_t(x)|$ and $q(x):=\one_{L^2(\R^2)}-p(x)$. 
\end{definition}
\begin{lemma} 
	We then have the following estimates:
	\begin{align}
		\norm{v' \varphi_t}_2&\leq C(\norm{ \varphi_t}^3_{H^1}+\norm{ \varphi_t}_{H^2}\norm{ \varphi_t}_{H^1}),\label{ine:v'}\\
		\norm{v(x -y)\varphi_t(x)\varphi_t(y)}_{L^2(\R^4)}^2&\leq C|\log R|^2(\norm{ \varphi_t}^2_{H^1}+\norm{ \varphi_t}_{H^1}+1),\label{ine:vphiphi}\\
		\norm{w' \varphi_t}_2&\leq C\norm{ \varphi_t}_{H^1}^3(\norm{ \varphi_t}_{H^1}+1),\label{ine:w'}\\
		\norm{w(x-y, x-z)\varphi_t(x)\varphi_t(y)\varphi_t(z)}_{L^2(\R^6)}&\leq C|\log R|\norm{ \varphi_t}_{H^1}^3.\label{ine:wphiphiphi}
	\end{align}
\end{lemma}
\begin{proof}
	We start with the first estimate
	\begin{align}
		\norm{v' \varphi_t}_2&\leq 2\norm{ \bAR\left [|\varphi_t|^{2}\right ]\cdot\nabla\varphi_t}_2+\norm{\nabla^{\perp}w_R\ast \mathbf{J}[\varphi_t]\varphi_t}_2\nonumber\\
		&\leq 2\norm{\nabla\varphi_t}_4\norm{ \bAR\left [|\varphi_t|^{2}\right ]}_{4}+\norm{\nabla^{\perp}w_R\ast \mathbf{J}[\varphi_t]}_4\norm{\varphi_t}_4\nonumber\\
		&\leq C\norm{\nabla\varphi_t}_4\norm{ \varphi_t}_{H^1}+2\snorm{\nabla^{\perp}w_R}_{2,w}\norm{\varphi_t\nabla \varphi_t}_{\frac43}\norm{\varphi_t}_4\nonumber
	\end{align}
	where we used \eqref{eq:N4} as well as \eqref{eq:WY} and H\"older inequality to treat the convolution product. For the second estimate of the lemma we write
	\begin{align}
		\norm{v(x_1 -x_2)\varphi_t(x_1)\varphi_t(x_2)}^2_{L^2(\R^4)}&=\braket{\varphi_t(x_1)\varphi_t(x_2)}{(v(x_1 -x_2))^2\varphi_t(x_1)\varphi_t(x_2)}\nonumber\\
		&\leq C|\log R|^2\braket{\varphi_t(x_1)\varphi_t(x_2)}{(1-\Delta_{x_1})(1-\Delta_{x_2})\varphi_t(x_1)\varphi_t(x_2)}\nonumber\\
		&\leq C|\log R|^2(\norm{ \varphi_t}^2_{H^1}+\norm{ \varphi_t}_{H^1}+1)
	\end{align}
	by Lemma~\ref{lem:mix2bd}, using the operator bound of \eqref{eq:mix2-2}. The next term to treat is
	\begin{align}
		\norm{w' \varphi_t}_2&\leq \norm{ \bAR\left [|\varphi_t|^{2}\right ]^2\varphi_t}_2+2\norm{( \nabla^{\perp}w_R\ast\bAR\left [|\varphi_t|^{2}\right ]|\varphi_t|^2)\varphi_t}_2\nonumber\\
		&\leq    \norm{ \bAR\left [|\varphi_t|^{2}\right ]}^2_{8} \norm{\varphi_t}_4+C\norm{ \nabla^{\perp}w_R\ast\bAR\left [|\varphi_t|^{2}\right ]|\varphi_t|^2}_4\norm{\varphi_t}_4\nonumber\\
		&\leq C\norm{ \varphi_t}_{H^1}^3+C\snorm{ \nabla^{\perp}w_R}_{2,w}\norm{\bAR\left [|\varphi_t|^{2}\right ]|\varphi_t|^2}_{\frac43}\norm{\varphi_t}_{4}\nonumber\\
		&\leq C\norm{ \varphi_t}_{H^1}^3+C\snorm{ \nabla^{\perp}w_R}_{2,w}\norm{\bAR\left [|\varphi_t|^{2}\right ]}_{4}\norm{|\varphi_t|^2}_{2}\norm{\varphi_t}_{4}\nonumber\\
		&\leq C\norm{ \varphi_t}_{H^1}^3+C\norm{ \varphi_t}_{H^1}^4
	\end{align}
	where we used \eqref{eq:N4} twice with $s=0$. Once with $p=8$ and once with $p=4$ combined with \eqref{eq:WY}. 

	Finally, the last estimate is
	\begin{align}
		&\norm{w(x-y, x-z)\varphi_t(x)\varphi_t(y)\varphi_t(z)}_{L^2(\R^6)}^2\nonumber\\
		&\leq\braket{\varphi_t(x)\varphi_t(y)\varphi_t(z)}{|\nabla^{\perp}w_R(x-y)|^2 |\nabla^{\perp}w_R(x-z)|^2\varphi_t(x)\varphi_t(y)\varphi_t(z)}\nonumber\\
		&= \norm{(|\nabla^{\perp}w_R|^2\ast |\varphi_t|^2)^2|\varphi_t|^2}_1\nonumber\\
		&\leq \norm{\varphi_t}_4^2\snorm{\nabla^{\perp}w_R}^4_{2+2\varepsilon}\norm{\varphi_t}_{\frac{8(1+\varepsilon)}{1+5\varepsilon}}^4\nonumber\\
		&\leq C|\log R|^2\norm{ \varphi_t}_{H^1}^6.
	\end{align}
	Here, we chose $\varepsilon =\frac{1}{2|\log R|}$ in the last inequality to apply Corollary~\ref{cor:nabwRLp}.
\end{proof}
The next two lemmas involve the projectors $p(x)$ and $q(x)$ and show what happens in the cases when $-\Delta$ acts on a $\varphi_t$ or does not. 
\begin{lemma}\label{lem:vpq} 
	For any function $\Psi_N\in L_{\mathrm{sym}}^2(\R^{2N})$, we have
	\begin{align}\label{ine:vpq}
		\norm{v(x_1-x_2)p_1q_2\Psi_N}&\leq  C|\log R|(\norm{\varphi_t}_{H^1}+1)(\norm{\nabla_2 q_2\Psi_N}+\norm{ q_2\Psi_N}).
	\end{align}
\end{lemma}
\begin{proof}
	The main ingredient is \eqref{eq:mix2-2} of Lemma~\ref{lem:mix2bd} that we use to get
	\begin{align}
		\norm{v(x_1-x_2)p_1q_2\Psi_N}^2&=\braket{p_1q_2\Psi_N}{(v(x_1 -x_2))^2p_1q_2\Psi_N}\nonumber\\
		&\leq C|\log R|^2\braket{p_1q_2\Psi_N}{(1-\Delta_1)(1-\Delta_2)p_1q_2\Psi_N}\nonumber\\
		&\leq C|\log R|^2(\norm{\varphi_t}_{H^1}^2+1)\braket{p_1q_2\Psi_N}{(1-\Delta_2)p_1q_2\Psi_N}\nonumber\\
		&\leq C|\log R|^2(\norm{\varphi_t}_{H^1}^2+1)(\norm{\nabla_2 q_2\Psi_N}^2+\norm{ q_2\Psi_N}^2).
	\end{align}
	Note that the above holds the same for $v(y-x)$ even if $v(x-y)\neq v(y-x)$.
\end{proof}

\begin{lemma}\label{lem:sqrtw}
	For any function $\Psi_N\in L_{\mathrm{sym}}^2(\R^{2N})$, we have for $u\in\{p,q\}$ that
	\begin{align*}
		\norm{\sqrt{w(x_1-x_2,x_1-x_3)}u_1u_2u_3\Psi_N}&=
		\begin{cases}
			C(\norm{\varphi_t}_{H^1}+1)\norm{u_1u_2u_3\Psi_N}\ & \text{if } u_1=p_1 \\
			C(\norm{u_1u_2u_3\Psi_N}+\norm{\nabla_1 u_1\Psi_N}) & \text{if }u_1=q_1
		\end{cases}
	\end{align*}
	This also holds true for any circular permutation of $(x,y,z)$.
\end{lemma}

\begin{proof}
	We have by Lemma~\ref{lem:3body}
	\begin{align}
		\braket{u_1u_2u_3}{w(x_1-x_2,x_1-x_3)u_1u_2u_3}_{\Psi_N}\leq C\braket{u_1u_2u_3}{(1-\Delta_1) u_1u_2u_3}_{\Psi_N}\nonumber
	\end{align}
	such that if $u_1=p_1=|\varphi_t(x_1)\rangle\langle\varphi_t(x_1)|$
	\begin{align}
		\braket{u_1u_2u_3}{w(x_1-x_2,x_1-x_3)u_1u_2u_3}_{\Psi_N}\leq C(\norm{\varphi_t}^2_{H^1}+1)\norm{u_1u_2u_3\Psi_N}^2\nonumber
	\end{align}
	and if $u_1=q_1$ we get
	\begin{align}
		\braket{u_1u_2u_3}{w(x_1-x_2,x_1-x_3)u_1u_2u_3}_{\Psi_N}&\leq C\norm{u_1u_2u_3\Psi_N}^2+C\norm{\nabla_1 u_1u_2u_3\Psi_N}^2\nonumber\\
		&\leq C\norm{u_1u_2u_3\Psi_N}^2+C\norm{\nabla_1 u_1\Psi_N}^2\nonumber
	\end{align}
	using that $u\leq \one_{L^2(\R^2)}$.
\end{proof}

\section{Chern--Simons--Schr\"odinger Equation}\label{sec:css}
In this section, we establish the well-posedness of the Chern--Simons--Schr\"odinger equation $\mathrm{CSS}(R,\varphi_0)$ for initial data $\varphi_0\in H^2(\R^2)$. We then demonstrate the convergence of the solution of $\mathrm{CSS}(R,\varphi_0)$ to the solution of $\mathrm{CSS}(\varphi_0)$ when $R\to 0$.

\subsection{Well-posedness}
We base our approach on results for $\mathrm{CSS}(\varphi_0)$ derived in \cite{Ber_Bou_Sau_95} to prove, among other things, the local well-posedness in $H^2(\R^2)$ and the conservation of the energy. To make more effective use of the existing results, we rewrite $\mathrm{CSS}(R,\varphi_0)$ using the notation system introduced in \cite{Ber_Bou_Sau_95}. 

The equation can be expressed in the covariant derivative form by setting the parameters $g$ and $\kappa$ in \cite{Ber_Bou_Sau_95} as $\kappa^{-1} = -2\pi \beta$, $g = 0$, and by rewriting \eqref{eq:pilotu} as follows:
\begin{equation}\label{eq:pilotuB}
	(\im \partial_{t}- \mathcal{A}_R^0 )u =-\left (\nabla -\im \mathcal{A}_R\right )^{2}u
\end{equation}
where we define
\begin{align}
	\mathcal{A}_R
	&=-\beta \bAR\left [| u|^{2}\right ]\nonumber
\end{align}
and
\begin{align}
	\mathcal{A}_R^0
	&:=-2\beta \nabla^{\perp}w_{R}\ast\Im [\overline{u}(\nabla -\im \mathcal{A}_R)u]=-2\beta \nabla^{\perp}w_{R}\ast\mathbf{j}_R\nonumber
\end{align}
where $\mathbf{j}_R:=\left (\beta \bAR\left [| u|^{2}\right ]| u|^{2}+\frac{1}{2}\mathbf{J}[u] \right ).$
Note that, in this translation, the factor $2$ differs from that in \cite[Equation 2.6]{Ber_Bou_Sau_95} due to their choice of kinetic energy $-\frac{1}{2}\Delta$. With that notation, we have the property
\begin{equation}\label{eq:nicepropjr}
	\int_{\R^2}\mathbf{j}_R\cdot \partial_t\mathcal{A}_R= -\beta\int_{\R^2}\mathbf{j}_R\cdot(\nabla^{\perp}w_R\ast \partial_t|u_t|^2)=\beta\int_{\R^2}(\mathbf{j}_R\ast\nabla^{\perp}w_R )\partial_t|u_t|^2=-\frac{1}{2}\int_{\R^2}\mathcal{A}^0_R \partial_t|u_t |^2.
\end{equation}

\begin{theorem}[\textbf{Conservation of the energy}]\label{Thm:cons}
	There exists a constant $c>0$ such that for any $0\leq R \leq c$ and any solution $u_t$ of $\mathrm{CSS}(R, u_0)$, the energy is conserved 
	\begin{equation}\label{eq:conservation}
		\mathcal{E}_R^{\mathrm{af}}[u_t]=\mathcal{E}_R^{\mathrm{af}}[u_0]
		\quad\text{for all }t\geq 0.
	\end{equation}
	Moreover, we deduce that 
	\begin{equation}\label{ine:normH1u}
		\norm{u_t}_{H^1} \leq C
	\end{equation}
	where the constant $C$ depends only on the fixed parameters $\beta$ and $\norm{u_0}_{H^1}$ and not on $R$.
\end{theorem}

\begin{proof}
	We start by computing
	\begin{align}
		\frac{\mathrm{d}}{\mathrm{d}t} \int_{\R^2} |(\nabla - \im \mathcal{A}_R)u_t|^2 
		&= 2 \Re \int_{\R^2} \big((\nabla - \im \mathcal{A}_R)u_t\big)^* (\nabla - \im \mathcal{A}_R)\partial_t u_t 
		- 2 \int_{\R^2} \mathbf{j}_R \cdot \partial_t \mathcal{A}_R \nonumber \\
		&= 2 \Re \int_{\R^2} \big((\nabla - \im \mathcal{A}_R)u_t\big)^* (\nabla - \im \mathcal{A}_R)\partial_t u_t 
		+  \int_{\R^2} \mathcal{A}^0_R \partial_t |u_t|^2.
	\end{align}
	Next, we multiply \eqref{eq:pilotuB} by $\partial_t \overline{u}_t$ and take the real part, obtaining
	\begin{equation}
		0 = \Re \left\langle \im\partial_t u_t,\partial_t u_t\right\rangle. = \Re \int_{\R^2} \big((\nabla - \im \mathcal{A}_R)u_t\big)^* (\nabla - \im \mathcal{A}_R)\partial_t u_t 
		+\frac{1}{2} \int_{\R^2} \mathcal{A}^0_R \partial_t |u_t|^2, \nonumber
	\end{equation}
	which concludes the proof of \eqref{eq:conservation}. 
	The $H^1$-norm of $u_t$ can then be estimated using the following argument:
	\begin{align}
		\norm{\nabla u_t}_2 
		&= \norm{\nabla u_t - \im \mathcal{A}_R u_t + \im \mathcal{A}_R u_t}_2 \nonumber \\
		&\leq \mathcal{E}_R^{\mathrm{af}}[u_0]^{1/2} + |\beta| \norm{\bAR[|u_t|^2]u_t}_2 \nonumber \\
		&\leq \mathcal{E}_R^{\mathrm{af}}[u_0]^{1/2} + C |\beta| \norm{u_t}_2^2 \norm{\nabla |u_t|}_2 \nonumber \\
		&\leq \big(1 + C |\beta|\big)\mathcal{E}_R^{\mathrm{af}}[u_0]^{1/2}.
	\end{align}
	Here, we used Hardy's inequality~\ref{lem:af_three_body} and the diamagnetic inequality~\ref{lem:af_smooth_ineqs} to reconstruct $\mathcal{E}_R^{\mathrm{af}}[u_t]$ and applied the conservation of energy \eqref{eq:conservation}. 
	Using Proposition~\ref{prop:conv}, we replace $\mathcal{E}_R^{\mathrm{af}}[u_0]$ with $\mathcal{E}^{\mathrm{af}}[u_0]$ in the above. 
	For $R$ small enough, we then obtain
	\begin{equation}\label{ine:normH1u2}
		\norm{u_t}_{H^1}^2 \leq C\big(1 + \mathcal{E}^{\mathrm{af}}[u_0]\big).
	\end{equation}
	Here, $\mathcal{E}^{\mathrm{af}}[u_0]$ is also a fixed parameter, and we have
	\begin{equation}\label{ine:EafH1}
		\mathcal{E}^{\mathrm{af}}[u_0] \leq 2\snorm{\nabla u_0}_2^2 + 2\beta^2\snorm{\bA[|u_0|^2]u_0}_2^2 \leq C\snorm{u_0}_{H^1}^2,
	\end{equation}
	where we used the triangle inequality and \eqref{eq:N4}. 
	This concludes the proof.
\end{proof}

\begin{theorem}[\textbf{Local well-posedness in $H^2$}]\label{thm:wellpose}
	There exists a constant $c>0$ such that for any $R<c$ and any initial data $u_0\in H^2$ there exists a time $T>0$ not depending on $R$ but only on $\norm{u_0}_{H^2}$ and $\beta$ such that $\mathrm{CSS}(R,u_0)$ has a unique solution 
	\begin{equation}\label{eq:u-LWP}
		u_t\in C([0,T],H^2(\R^2)).
	\end{equation} 
	Moreover, there exist three constants $C,c', T>0$ depending on $\norm{u_0}_{H^1}$ but not on $R$ such that for any $|\beta|<c'$ and any $0\leq t\leq T$, we have the control
	\begin{align}\label{ine:H2norm2}
		\snorm{-\Delta u_t}_2&\leq C.
	\end{align}
\end{theorem}

\begin{remark}
	We can, in fact, prove \eqref{eq:u-LWP} with $T(R) = +\infty$ for any $R > 0$, as demonstrated in Theorem~\ref{thm:u-GWP}. However, Theorem~\ref{thm:u-GWP} also provides a global-in-time bound for $\norm{-\Delta u_t}_2$ that diverges as $R$ approaches zero.
\end{remark}

\begin{proof}[Proof of Theorem~\ref{thm:wellpose}]
	\textbf{Abstract well-posedness.} 
	The proof for $R=0$ can be found in \cite[Theorem 2.1]{Ber_Bou_Sau_95} and provides a $T_1>0$ for which the well-posedness of $\mathrm{CSS}(0,u_0)$ holds for time $0\leq t\leq T_1$.
	When $0<R<1$, the well-posedness follows by the same argument of  \cite[Theorem 2.1]{Ber_Bou_Sau_95} and provides a $T(R)>0$. However, with so little information, we could have $T(R)<T_1$ and $\inf_{R\geq 0}T(R)$ might, a priori, be zero. To make our proof meaningful for any $0<R\leq c$, we need to show that $T(R)>C$ for some $C>0$ independent of $R$. This is the aim of the next paragraph. \newline
	\textbf{Control of the $R$-dependent $H^1$-norm.} Let $u_t$ be the solution of $\mathrm{CSS}(R,u_0)$. By Theorem~\ref{Thm:cons}, $\norm{u_t}_{H^1}\leq C$ where $C$ does not depend on $R$.\newline
	\textbf{Control of the $R$-dependent $H^2$-norm.}  We will prove $T(R)>C>0$ via the control of  $\norm{-\Delta u_t}_2$. We first express it in terms of $\norm{\partial_t u_t}_2$ and bound this last norm via Gr\"onwall's lemma~\ref{lem:Gron}. To simplify the calculation, we introduce the notation $\mathcal{A}:=(\nabla -\im \mathcal{A}_R )^{2}+ \mathcal{A}_R^0$ such that
	\begin{equation}
		\im\partial_t u_t=\mathcal{A}u_t.
	\end{equation}
	Note that $\mathcal{A}$ is self-adjoint because $\mathcal{A}_R^0$ is real-valued. 
	Before starting, we prove a useful inequality
	\begin{align}
		\norm{(-\im\nabla+\beta\bAR[|u_t|^2])^2u_t}_2
		&=\snorm{(\im\partial_{t}-\mathcal{A}_{R}^{0})u_{t}}_{2}\nonumber\\
		&\leq  \norm{\partial_t u_t}_2+ \norm{\mathcal{A}_R^0u_t}_2 \nonumber\\
		&\leq  \norm{\partial_t u_t}_2+C|\beta|\norm{u_t}_4\norm{\mathcal{A}_R^0}_4\nonumber\\
		&\leq \norm{\partial_t u_t}_2+C|\beta|\norm{u_t}_4\snorm{\nabla^{\perp}w_R}_{2,w}\norm{\bAR[|u_t|^2]|u_t|^2+\mathbf{J}[u_t]}_{\frac43}\nonumber\\
		&\leq \norm{\partial_t u_t}_2+C|\beta|\norm{u_t}^3_{H^1}.\label{ine:magtopartu1}
	\end{align}
	Using \eqref{eq:S} to get  $\norm{\nabla u_{t}}_{4}^2\leq \norm{\nabla u_{t}}^2_{2}+\norm{\Delta u_{t}}^2_{2}$ and \eqref{eq:N4} of Lemma~\ref{lem:estnew} we now express
	\begin{align}
		\norm{-\Delta u_{t}}_{2}&=\snorm{(\nabla-\im\mathcal{A}_{R}+\im\mathcal{A}_{R})^{2}u_{t}}_{2}\nonumber\\
		&\leq\snorm{(\nabla-\im\mathcal{A}_{R})^{2}u_{t}}_{2}+\snorm{2\beta\bAR[|u_{t}|^{2}](-\im\nabla+\beta\bAR[|u_{t}|^{2}])u_{t}}_{2}+|\beta|^{2}\snorm{\left|\bAR[|u_{t}|^{2}]\right|^{2}u_{t}}_{2}\nonumber\\
		&\leq\snorm{(\im\partial_{t}-\mathcal{A}_{R}^{0})u_{t}}_{2}+|2\beta|\snorm{\bAR[|u_{t}|^{2}]}_{4}\snorm{(-\im\nabla+\beta\bAR[|u_{t}|^{2}])u_{t}}_{4}+|\beta|\snorm{\bAR[|u_{t}|^{2}]}^2_{8}\snorm{u_{t}}_{4}\nonumber\\
		&\leq\snorm{(\im\partial_{t}-\mathcal{A}_{R}^{0})u_{t}}_{2}+C|\beta|\,\|u_{t}\|_{H^1}\left(\norm{\nabla u_{t}}_{4}+\snorm{(\beta\bAR[|u_{t}|^{2}])u_{t}}_{4}+C|\beta|\snorm{u_{t}}^2_{H^1}\right)\nonumber\\
		&\leq\snorm{(\im\partial_{t}-\mathcal{A}_{R}^{0})u_{t}}_{2}+C|\beta|\,\|u_{t}\|_{H^1}\left(\norm{\Delta u_{t}}_{2}+|\beta|\snorm{(\bAR[|u_{t}|^{2}])}_{8}\snorm{u_{t}}_{8}+|\beta|\snorm{u_{t}}^2_{H^1}+C\snorm{u_{t}}_{H^{1}(\R^2)}\right)\nonumber\\
		&\leq\snorm{(\im\partial_{t}-\mathcal{A}_{R}^{0})u_{t}}_{2}+C|\beta|\,\|u_{t}\|_{H^1}\left(\norm{\Delta u_{t}}_{2}+C|\beta|\snorm{u_{t}}^2_{H^1}+C\snorm{u_{t}}_{H^1}\right).\nonumber\\
		\intertext{Then we use the estimate \eqref{ine:magtopartu1} to have}
		&\leq\snorm{\partial_{t}u_{t}}_{2}+C|\beta|\,\|u_{t}\|_{H^1}\norm{\Delta u_{t}}_{2}+\snorm{\mathcal{A}_{R}^{0}u_{t}}_{2}+C|\beta|^2\norm{u_{t}}_{H^{1}(\R^{2})}^{3}+C|\beta|\norm{u_{t}}^2_{H^{1}(\R^{2})}\nonumber\\
		&\leq\snorm{\partial_{t}u_{t}}_{2}+C|\beta|\,\|u_{t}\|_{H^1}\norm{\Delta u_{t}}_{2}+C|\beta|^2\snorm{u_{t}}_{H^{1}(\R^{2})}^{3}+C|\beta|\snorm{u_{t}}_{H^{1}(\R^{2})}^{2}\label{ine:H2ofu1}
	\end{align}
	We observe in \eqref{ine:H2ofu1} that if $C|\beta|\,\|u_{t}\|_{H^1}<1$, then
	\begin{equation}	
		\norm{-\Delta u_{t}}_{2}\leq C \left(\norm{\partial_{t}u_{t}}_{2}+\norm{u_t}_{H^1}\right).
	\end{equation}
	This explains the Assumption~\ref{assumption} that $|\beta|<c$ for some $c>0$.\footnote{It is possible to avoid this constraint at the price of divergences in $R$; see Theorem~\ref{thm:u-GWP}.}
	The next step is to compute
	\begin{equation}
		\partial_t \norm{\partial_t u_t}_2^2=\partial_t \norm{\mathcal{A} u_t}_2^2=2\im\Im\braket{(\im\partial_t \mathcal{A})u_t}{\mathcal{A}u_t}+2\im\Im\braket{\mathcal{A}^2u_t}{\mathcal{A}u_t}\nonumber
	\end{equation}	
	in order to bound $ \norm{\partial_t u_t}_2$. The second term in the above is the imaginary part of a real number and cancels. This cancellation is the main reason why we apply a Gr\"onwall argument to $\norm{\partial_t u_t}_2$ and not directly on the magnetic term $ \norm{(\nabla  -\im \mathcal{A}_R )^2u_t}_2$. 
	We are then left with
	\begin{align}
		2\im\Im\braket{(\im\partial_t \mathcal{A})u_t}{\mathcal{A}u_t}&=2\im \Im \braket{\im\partial_t (-\im\nabla+\beta\bAR[|u_t|^2])^2u_t+(\im\partial_t\mathcal{A}_R^0)u_t}{\mathcal{A}u_t}
		:=E^{(1)}+E^{(2)}.\nonumber
	\end{align}
	To estimate $E^{(1)}$, we first calculate
	\begin{equation}
		\im\partial_t \bAR[|u_t|^2]=2\nabla^{\perp}w_R\ast\Re[\overline{u}_t\im\partial_t u_t]\nonumber
	\end{equation}	
	and get
	\begin{equation}
		|E^{(1)}|=8\left|\Im \braket{\beta\nabla^{\perp}w_R\ast\Re[\overline{u}_t\im\partial_t u_t]\cdot(-\im\nabla+\beta\bAR[|u_t|^2])u_t}{\partial_tu_t}\right|
		\nonumber
	\end{equation}
	with
	\begin{align}
		|E^{(1)}|&\leq C\snorm{\nabla^{\perp}w_R\ast\Re[\overline{u}_t\im\partial_t u_t]}_{4}\snorm{(-\im\nabla+\beta\bAR[|u_t|^2])u_t}_{4}\snorm{\partial_tu_t}_2\nonumber\\
		&\leq C\snorm{\partial_tu_t}_2\snorm{u_t}_{H^1}\snorm{\nabla^{\perp}w_R}_{2,w}\snorm{\overline{u}_t\im\partial_t u_t}_{\frac{4}{3}}\nonumber\\
		&\leq C\snorm{\partial_tu_t}_2\snorm{u_t}_{H^1} \norm{u_t}_{H^2} \snorm{\im\partial_t u_t}_{2}\snorm{u_t}_{4}\nonumber\\
		&\leq C\left(\norm{\partial_tu_t}^3_2 + \snorm{\partial_tu_t}^2_2\right)\snorm{u_t}^2_{H^1}.\label{ine:EAR11}
	\end{align}
	Here, we used the inequality \eqref{ine:magtopartu1} to get the last line but also Young inequality for the convolution, the conservation of energy \eqref{eq:conservation} and the estimate \eqref{ine:normH1u}. 
	We now treat
	\begin{align*}
		|E^{(2)}|
		&=2|\im \Im \braket{(\im\partial_t\mathcal{A}_R^0)u_t}{\mathcal{A}u_t}|\\
		&\leq C\norm{\partial_tu_t}_2\norm{u_t}_4\norm{\im\partial_t\mathcal{A}_R^0}_4\\
		&\leq C\norm{\partial_tu_t}_2\norm{u_t}_4\snorm{\nabla^{\perp}w_R}_{2,w}\norm{4\beta\bAR[\Re(\overline{u_t}\im\partial_t u_t)]|u_t|^2+4\beta\bAR[|u_t|^2]\Re(\overline{u_t}\im\partial_t u_t)}_{\frac43}\\
		&\qquad+ C|\snorm{\partial_tu_t}_2\snorm{u_t}_4\snorm{\nabla^{\perp}w_R\ast \mathbf{J}[\im\partial_t u_t]}_4,
	\end{align*}
	where we have two $\tfrac{4}{3}$-norms to treat and the last term. The first gives 
	\begin{align}
		\norm{4\beta\bAR[\Re(\overline{u_t}\im\partial_t u_t)]|u_t|^2}_{\frac43}&\leq C\norm{|u_t|^2}_2\norm{\nabla^{\perp}w_R\ast \Re(\overline{u_t}\im\partial_t u_t)}_4\nonumber\\
		&\leq  C\norm{|u_t|^2}_2\snorm{\nabla^{\perp}w_R}_{2,w}\norm{\overline{u_t}\im\partial_t u_t}_{\frac43}\nonumber\\
		&\leq C\norm{|u_t|^2}_2\snorm{\nabla^{\perp}w_R}_{2,w}\norm{\im\partial_t u_t}_{2}\norm{u_t}_4.
	\end{align}
	Here we used \eqref{eq:WY}.
	For the second term, we obtain
	\begin{align}
		\norm{4\beta\bAR[|u_t|^2]\Re(\overline{u_t}\im\partial_t u_t)}_{\frac43}&\leq C\norm{\partial_tu_t}_2\norm{\bAR[|u_t|^2]}_8\norm{u_t}_8
		\leq  C\norm{\partial_tu_t}_2\norm{u_t}^2_{H^1}.
	\end{align}
	The last term to treat is
	\begin{align}
		\norm{\nabla^{\perp}w_R\ast \mathbf{J}[\im\partial_t u_t]}_4&= \norm{\nabla^{\perp}w_R\ast \im[(\im\partial_tu_t)\nabla\overline{u}_t+u_t\nabla(\im\partial_t\overline{u}_t)-(\im\partial_t\overline{u}_t)\nabla u_t-\overline{u}_t\nabla (\im\partial_tu_t)]}_4\nonumber\\
		&\leq  4\norm{\nabla^{\perp}w_R\ast (\im\partial_tu_t)\nabla\overline{u}_t}_4\nonumber\\
		&\leq C\snorm{\nabla^{\perp}w_R}_{2,w}\norm{ (\im\partial_tu_t)\nabla\overline{u}_t}_{\frac43}\nonumber\\
		&\leq C\snorm{\nabla^{\perp}w_R}_{2,w}\norm{ \partial_tu_t}_{2}\norm{u_t}_4
	\end{align}
	where we first used that $\nabla^{\perp}\cdot\nabla =0$ to integrate by parts in the convolution product the two $u_t\nabla(\im\partial_t\overline{u}_t)$ terms. We then used \eqref{eq:WY}.
	We conclude that
	\begin{align}\label{ine:E21}
		|E^{(2)}|
		&\leq C\norm{\partial_tu_t}^2_2(\norm{u_t}^4_{H^1}+\norm{u_t}^3_{H^1}+\norm{u_t}^2_{H^1}).
	\end{align}
	Combining the estimates \eqref{ine:EAR11} and \eqref{ine:E21} for $E^{(1)}$ and $E^{(2)}$ yields
	\begin{align}\label{ine:dtu}
		\partial_t \snorm{\partial_t u_t}_2^2&\leq  C\left(\norm{\partial_t u_t}_2^2 + \snorm{\partial_t u_t}_2^3 \right).
	\end{align}
	We can apply Gr\"onwall's lemma~\ref{lem:Gron} to the function $f_t:=\snorm{\partial_t u_t}_2+1$ and obtain
	\begin{align}
		\snorm{\partial_t u_t}_2&\leq \frac{1}{C_1 - C_2t}
	\end{align}
	for some $C_2,C_1>0$ depending on $\beta$ and $\norm{u_0}_{H^1}$.
	We bound $\norm{\mathcal{A}u_0}_2$ using the same techniques as earlier in the proof. Substituting this last inequality into \eqref{ine:H2ofu1}, and choosing $T := \min(T_1, 0.99 \frac{C_1}{C_2})$, allows us to conclude the proof.
\end{proof}

\subsection{Convergence of Solutions}

\begin{lemma}[\textbf{Convergence of the effective equation}]\label{lem:Conv_pilot}
	Let $\varphi_0 \in H^2(\R^2)$ and define $\varphi^R_t$ to be the solution of $\mathrm{CSS}(R,\varphi_0)$ and $\varphi_t$ to be the solution of $\mathrm{CSS}(\varphi_0)$. Then, there exist $C,c,c'>0$ such that for any $0<R\leq c$, $|\beta|\leq c'$ and $t< T$ (where the $T$ given by Theorem~\ref{thm:wellpose}):
	\begin{equation}
		\partial_t\norm{\varphi_t -\varphi_t^R}_2^2\leq  CR^2+C\norm{\varphi_t-\varphi^R_t}_2^2
	\end{equation}
	which implies
	\begin{equation}
		\norm{\varphi_t -\varphi_t^R}_2^2\leq C R^2 e^{Ct}
	\end{equation}
	by Gr\"onwall's lemma~\ref{lem:Gron}.
\end{lemma}

\begin{proof}
	Let us denote $\mathcal{E}_{R} = (-\im\nabla + \beta\bAR[|\varphi_t^R|^2])^2$. Then we have
	\begin{align}
		\partial_t \norm{\varphi_t - \varphi_t^R}_2^2 
		&= 2\Im \left[\bra{(\mathcal{E}_0 + \mathcal{A}^0_0)\varphi_t}\ket{\varphi_t^R} 
		- \bra{\varphi_t}\ket{(\mathcal{E}_R + \mathcal{A}^0_R)\varphi_t^R}\right]\nonumber \\
		&= 2\Im \left[\bra{(\mathcal{E}_0 + \mathcal{A}^0_0)\varphi_t}\ket{(\varphi_t^R - \varphi_t)} 
		- \bra{\varphi_t}\ket{(\mathcal{E}_R + \mathcal{A}^0_R)(\varphi_t^R - \varphi_t)}\right],\label{eq:1bd-difference}
	\end{align}
	where we used the fact that $\bra{(\mathcal{E}_R + \mathcal{A}^0_R)\varphi_t}\ket{\varphi_t}$ is real.
	We calculate separately the $\mathcal{E}$ part and the $\mathcal{A}^0$ part. 

	We begin with \textbf{the $\mathcal{E}$ part} and compute
	\begin{align}
		E:&=\bra{\mathcal{E}_0\varphi_t}\ket{(\varphi_t^R-\varphi_t)}-\bra{\varphi_t}\ket{\mathcal{E}_R(\varphi_t^R-\varphi_t)}\nonumber\\
		&=2\beta\bra{-\im\nabla \varphi_t}\ket{(\bA\left [\b\varphi_t\b^{2}\right ]-\bAR\left [\b\varphi^R_t\b^{2}\right ])(\varphi_t^R-\varphi_t)}+\beta^2\bra{(\bA\left [\b\varphi_t\b^{2}\right ]^2-\bAR\left [\b\varphi^R_t\b^{2}\right ]^2)\varphi_t^R}\ket{(\varphi_t^R-\varphi_t)}\nonumber\\
		&\leq C\norm{\nabla\varphi_t}_{4}\norm{\bA\left [\b\varphi_t\b^{2}\right ]-\bAR\left [\b\varphi^R_t\b^{2}\right ]}_4\norm{\varphi_t^R-\varphi_t}_{2}\nonumber\\
		& \quad+C \norm{\varphi^R_t}_{8}\norm{\bA\left [\b\varphi_t\b^{2}\right ]-\bAR\left [\b\varphi^R_t\b^{2}\right ]}_4\norm{\bA\left [\b\varphi_t\b^{2}\right ]+\bAR\left [\b\varphi^R_t\b^{2}\right ]}_8\norm{\varphi_t^R-\varphi_t}_2\nonumber\\
		&\leq C\norm{\nabla\varphi_t}_{4}\norm{\bA\left [\b\varphi_t\b^{2}\right ]-\bAR\left [\b\varphi^R_t\b^{2}\right ]}_4\norm{\varphi_t^R-\varphi_t}_{2}\nonumber\\
		& \quad+C \norm{\varphi^R_t}_{8}\norm{\bA\left [\b\varphi_t\b^{2}\right ]-\bAR\left [\b\varphi^R_t\b^{2}\right ]}_4\snorm{\nabla^{\perp}w_0}_{2,w}\norm{|u|^2}_{\frac{8}{5}}\norm{\varphi_t^R-\varphi_t}_2.
	\end{align}
	Here, we used that $\nabla \cdot\bAR\left [\b u\b^{2}\right ]=\nabla\cdot \nabla^{\perp}w_R \ast\b u\b^{2}= 0$, that $a^2 -b^2=(a+b)(a-b)$ and the triangle inequality together with \eqref{eq:WY}.
	We apply \eqref{eq:N4} and \eqref{eq:diff4} of Lemma~\ref{lem:estnew} combined with the Sobolev estimate \eqref{eq:S} on $\|\varphi^R_t\|_{8}$ and $\|\nabla\varphi_t\|_{4}$ to get
	\begin{align}\label{ine:E}
		|E|&\leq C(\|\varphi_t\|_{H^2}+1)(R^2 +\|\varphi_t^R-\varphi_t\|_2^2).
	\end{align}

	\medskip

	For \textbf{the $\mathcal{A}^0$ part,} we have to compute
	\begin{align}
		M:&=\bra{(\mathcal{A}^0_0-\mathcal{A}^0_R)\varphi_t}\ket{(\varphi_t^R-\varphi_t)}\nonumber\\
		&=\beta \bra{(\nabla^{\perp}w_0 -\nabla^{\perp}w_R)\ast(2\beta\bA\left [\b\varphi_t\b^{2}\right ] \b\varphi_t\b^{2}+\mathbf{J}\left [\varphi_t\right ])}\ket{\overline{\varphi_t}(\varphi_t^R-\varphi_t)}\nonumber\\
		&\quad -\beta \bra{\nabla^{\perp}w_R\ast 2\beta(\bAR\left [\b\varphi^R_t\b^{2}\right ] \b\varphi^R_t\b^{2}-\bA\left [\b\varphi_t\b^{2}\right ] \b\varphi_t\b^{2})}\ket{\overline{\varphi_t}(\varphi_t^R-\varphi_t)}\nonumber\\
		&\quad -\beta \bra{\nabla^{\perp}w_R\ast\left [\mathbf{J}\left [\varphi^R_t\right ]-\mathbf{J}\left [\varphi_t\right ]\right ]}\ket{\overline{\varphi_t}(\varphi_t^R-\varphi_t)}\nonumber\\
		&:= \beta (M_1 +M_2+M_3).\label{eq:defM}
	\end{align}
	We treat $M_1$, $M_2$ and $M_3$ separately.
	\begin{align}
		|M_1|&\leq C\norm{(\nabla^{\perp}w_0 -\nabla^{\perp}w_R)\ast  \mathbf{J}\left [\varphi_t\right ]}_{2}\norm{\overline{\varphi_t}(\varphi_t^R-\varphi_t)}_{2}\nonumber\\
		&\quad+ \norm{(\nabla^{\perp}w_0 -\nabla^{\perp}w_R)\ast 2\beta\bA\left [\b\varphi_t\b^{2}\right ] \b\varphi_t\b^{2}}_2 \norm{\overline{\varphi_t}(\varphi_t^R-\varphi_t)}_{2} \nonumber\\
		&\leq C\norm{\varphi_t}_{\infty}\norm{\varphi_t -\varphi^R_t}_{2}\norm{\nabla^{\perp}w_0 -\nabla^{\perp}w_R}_1\left(\norm{\nabla\varphi_t}_{4}\norm{\varphi^R_t}^2_{4} +\norm{\bA\left [\b\varphi_t\b^{2}\right ] \b\varphi_t\b^{2}}_2\right)\nonumber\\
		&\leq CR\norm{\varphi_t}_{H^2}^2\norm{\varphi_t -\varphi^R_t}_{2}\left(C+\norm{\bA\left [\b\varphi_t\b^{2}\right ] }_4\norm{\varphi_t}_8^2\norm{\varphi^R_t}_4\right)\nonumber\\
		&\leq CR^2+C\norm{\varphi_t -\varphi^R_t}_{2}^2\label{ine:M1}
	\end{align}
	where we used Young inequality for the convolution, Sobolev inequality \eqref{eq:S}, the embedding $\norm{\varphi_t}_{\infty}\leq \norm{\varphi_t}_{H^2}$  and \eqref{eq:N4} for the last inequality. 

	We continue by bounding $M_2$
	\begin{align}
		|M_2|&\leq  C\norm{\varphi_t^R-\varphi_t}_2\norm{\varphi_t}_4\norm{\nabla^{\perp}w_R\ast 2\beta(\bAR\left [\b\varphi^R_t\b^{2}\right ] \b\varphi^R_t\b^{2}-\bA\left [\b\varphi_t\b^{2}\right ] \b\varphi_t\b^{2})}_4\nonumber\\
		&\leq C\norm{\varphi_t^R-\varphi_t}_2(\norm{\varphi_t^R}_8\norm{\bAR\left [\b\varphi^R_t\b^{2}\right ]}_{8}\norm{\varphi_t^R-\varphi_t}_2 +\norm{|\varphi_t|^2}_2\norm{\bAR\left [\b\varphi^R_t\b^{2}\right ] -\bA\left [\b\varphi_t\b^{2}\right ] }_{4})\nonumber\\
		&\leq CR\norm{\varphi_t^R-\varphi_t}_2 +C\norm{\varphi_t^R-\varphi_t}_2^2\nonumber\\
		&\leq CR^2+C\norm{\varphi_t^R-\varphi_t}_2^2\label{ine:M2}
	\end{align}
	where we used \eqref{eq:WY},  \eqref{eq:N4} and \eqref{eq:diff4} of Lemma~\ref{lem:estnew}. We finish with $M_3$ and to do so, we write
	\begin{equation*}
		\mathbf{J}\left [\varphi^R_t\right ]- \mathbf{J}\left [\varphi_t\right ]=(\varphi_t-\varphi^R_t)\nabla \overline{\varphi_t}+\varphi^R_t\nabla (\overline{\varphi_t}-\overline{\varphi^R_t})+(\overline{\varphi^R_t}-\overline{\varphi_t})\nabla \varphi^R_t +\overline{\varphi_t}\nabla (\varphi^R_t-\varphi_t)
	\end{equation*}
	to get
	\begin{align*}
		|M_3|&\leq C\left |\bra{\nabla^{\perp}w_R\ast\left [ \mathbf{J}\left [\varphi^R_t\right ]- \mathbf{J}\left [\varphi_t\right ]\right ]}\ket{\overline{\varphi_t}(\varphi_t^R-\varphi_t)}\right |\\
		&= C\left |\bra{ \mathbf{J}\left [\varphi^R_t\right ]- \mathbf{J}\left [\varphi_t\right ]}\ket{\nabla^{\perp}w_R\ast\overline{\varphi_t}(\varphi_t^R-\varphi_t)}\right |\\
		&\leq 2C\left |\bra{(\varphi_t-\varphi^R_t)\nabla \overline{\varphi_t}+(\overline{\varphi^R_t}-\overline{\varphi_t})\nabla \varphi^R_t }\ket{\nabla^{\perp}w_R\ast\overline{\varphi_t}(\varphi_t^R-\varphi_t)}\right |
	\end{align*}
	where we used that $\nabla\cdot \nabla^{\perp}w_0=0$ to integrate by part the two terms $\varphi^R_t\nabla (\overline{\varphi_t}-\overline{\varphi^R_t})+\overline{\varphi_t}\nabla (\varphi^R_t-\varphi_t) $.

	We can conclude with 
	\begin{align}
		|M_{3}|
		&\leq C \snorm{\varphi_{t}^{R} - \varphi_{t}}_{2} 
		\big(\snorm{\nabla\varphi_{t}}_{4} + \snorm{\nabla\varphi_{t}^{R}}_{4}\big) 
		\snorm{\nabla^{\perp}w_{R} \ast \overline{\varphi_{t}}(\varphi_{t}^{R} - \varphi_{t})}_{4} \nonumber \\
		&\leq C \snorm{\varphi_{t}^{R} - \varphi_{t}}_{2} 
		\snorm{\varphi_{t}^{R}}_{H^2} 
		\snorm{\nabla^{\perp}w_{R}}_{2,w} 
		\snorm{\overline{\varphi_{t}}(\varphi_{t}^{R} - \varphi_{t})}_{\frac{4}{3}} \nonumber \\
		&\leq C \snorm{\varphi_{t}^{R} - \varphi_{t}}_{2} 
		\snorm{\varphi_{t}^{R}}_{H^2} 
		\norm{\varphi_{t}}_{4} 
		\snorm{(\varphi_{t}^{R} - \varphi_{t})}_{2} \nonumber \\
		&\leq C \snorm{\varphi_{t}^{R}}_{H^2} 
		\snorm{\varphi_{t}^{R} - \varphi_{t}}_{2}^{2}, \label{ine:M3}
	\end{align}
	where we used Young inequality for the convolution and the estimate \eqref{ine:normH1u}. 

	Combining \eqref{ine:M1}, \eqref{ine:M2}, and \eqref{ine:M3} with \eqref{eq:defM}, we obtain that 
	\begin{equation}
		|M| \leq CR^2 + C \snorm{\varphi_{t}^{R}}_{H^2} \snorm{\varphi_t^R - \varphi_t}_2^2.\label{ine:M}
	\end{equation}

	Noting that for $t < T$, we have $\norm{\varphi_t}_{H^2} \leq C$ by Theorem~\ref{thm:wellpose}.
	Finally, gathering \eqref{eq:1bd-difference}, \eqref{ine:E}, and \eqref{ine:M}, we conclude the proof.
\end{proof}

\section{Number of Particles in the Condensate}\label{sec:pickl}

\subsection{Bound on the time derivative of $\mathcal N_+(t)$}
As mentioned in the proof strategy outlined in Section~\ref{sec:strategy}, our aim is to follow the method in \cite{KnoPic-10, Pickl2011} and apply Grönwall's lemma to the quantity
\begin{equation}
	\mathcal N_+(t)=1-\braket{\varphi_t}{\gamma_N^{(1)}\varphi_t}\nonumber
\end{equation}
where $\gamma_N^{(1)}(t)=\Tr_{2,\dots, N}|\Psi_N(t)\rangle\langle\Psi_N(t)|$ with $\Psi_N(t)$ being the solution of the Schr\"odinger equation \eqref{def:schro} with initial data $\Psi_N(0)=\varphi_0^{\otimes N}$ and where $\varphi_t$ is the solution of $\text{CSS}(R,\varphi_0)$ defined in \eqref{eq:pilotu}.
We can rewrite
\begin{equation}\label{eq:gammatop}
	\braket{\varphi_t}{\gamma_N^{(1)}\varphi_t} = \int_{\R^{2(N-1)}}\left|\int_{\R^2}\varphi_t(x_1)\Psi_N(x_1,\dots,x_N)\,\dd x_1\right|^2\dd x_2\dots\dd x_N=\braket{\Psi_N}{p_1(t)\Psi_N}
\end{equation}
where we define the projectors
\begin{equation}\label{def:pq}
	p_1(t):=|\varphi_t(x_1)\rangle\langle\varphi_t(x_1)|\quad\text{and}\quad q_1(t):=\one_{L^2(\R^2)}-p_1(t).
\end{equation}
In the following, the notations $p_j$ and $q_j$ will refer to the same time-dependent projectors, acting on the variable $j$. Using \eqref{eq:gammatop}, we express
\begin{equation}
	\mathcal N_+(t)=\braket{\Psi_N}{q_1(t)\Psi_N}\nonumber
\end{equation}
and using the symmetry of $\Psi_N$, we equivalently have $\mathcal N_+(t)=\braket{\Psi_N}{q_j(t)\Psi_N}$ for any $j\in [1,N]$.
Using these properties, we can compute
\begin{align}\label{eq:BigSP}
	\partial_t \mathcal N_+(t)=-\im \bra{\Psi_N (t)}\ket{\left[(\HNR-H_{N,R}^{\mathrm{H}}),\frac{1}{N}\sum_{j=1}^N q_j\right]\Psi_N (t)}
\end{align}
where $[A,B]$ denotes the commutator of operators $A$ and $B$ and where the Hartree Hamiltonian is
\begin{align}
	H_{N,R}^{\mathrm{H}}:=&\sum_{j=1}^N\left (-\im\nabla_j +\beta \bAR\left [|\varphi_t|^{2}\right ](x_j)\right )^{2}\nonumber\\
	&- \beta\left [\nabla^{\perp}w_{R}\ast\left (2\beta \bAR\left [| \varphi_t|^{2}\right ]|\varphi_t|^{2}+\im\left (\varphi_t\nabla\overline{\varphi_t}-\overline{\varphi_t}\nabla\varphi_t \right ) \right )\right ](x_j).\label{eq:hartreeH}
\end{align}
As we have explained in the strategy, our proof will not directly focus on $ \mathcal N_+(t)$ but on $ \sqrt{\mathcal N}^{(1)}_+(t)$ that we define in the following paragraph.

\subsection{Properties of the Projectors}
Following \cite[Section 3.3]{KnoPic-10}, we write
\begin{equation}\label{eq:ide}
	\one_{L^2(\mathbb{R}^{2N})} = \prod_{j=1}^N (p_j + q_j),\quad\text{to define}\quad P_k := \sum_{\substack{a \in \{0,1\}^N \\ \sum_j a_j = k}} \prod_{j=1}^N p_j^{1-a_j} q_j^{a_j}
\end{equation}
where $P_k$ is obtained by expanding the product and collecting all summands containing exactly $k$ factors of $q$ operators. 
We use the convention that $P_k:=0$ whenever $k\notin \{0,\dots,N\}$.
Note that $P_k$ is an orthogonal projector, that $P_kP_l=\delta_{kl}P_k$, that $\sum_{k=0}^{+\infty}P_k=\one_{L^2(\R^{2N})}$ and that $P_0$ only contains projectors onto $\varphi_t$. For any weight function $f:\{0,\dots,N\}\to \mathbb{C}$ we define the operator
\begin{equation}\label{def:TF}
	\hat{f}:=\sum_{k=0}^{+\infty}f(k)P_k.
\end{equation}
It follows that $\hat{f}$ commutes with any $q_j$ and $P_k$. That way we can define the operator $\hat{m}$ with the real parameter
$\xi\geq 0$ as
\begin{equation}\label{def:mxi}
	\hat{m}(\xi):=\sum_{k=1}^N\left(\frac{k}{N}\right)^{\xi}P_k
\end{equation}
associated to the weight $m^{\xi}(k):=\left(\frac{k}{N}\right)^{\xi}$. These operators will be used extensively in the proof. Let us first note that we have the relations
\begin{equation}\label{eq:m11}
	\hat{m}(1)=\frac{1}{N}\sum_{j=1}^N\sum_{k=1}^Nq_j P_k=\frac{1}{N}\sum_{j=1}^Nq_j=\mathcal N_+^{(1)}
\end{equation}
that allows us to recover the term in the right-hand-side of the commutator in the expression \eqref{eq:BigSP}. We then see that $\hat{m}(1)=\mathcal N_+$ and generalization to $\hat{m}(\xi)$ allows to define
\begin{equation}\label{def:N+}
\sqrt{\mathcal N_+}(t):=\bra{\Psi_N(t)}\ket{\hat{m}(1/2)\Psi_N(t)}.
\end{equation}

\begin{lemma}\label{lem:mxi}
	The operator $\hat{m}(\xi)$ has the following properties. For any $\xi_1, \xi_2 \geq 0$, we have
	\begin{align}
		\hat{m}(\xi_1)\hat{m}(\xi_2)=\hat{m}(\xi_1 +\xi_2)\label{eq:m1}
	\end{align}
	and in particular 
	\begin{equation}
		\hat{m}(\xi_1)\hat{m}(-\xi_1 )=\one_{L^2(\R^{2N})} -P_0 \,.\label{eq:m2}
	\end{equation}
	Moreover, for any $\Psi \in L_{\mathrm{sym}}^2(\R^{2N})$ we have that
	\begin{equation}\label{ine:symmhat}
		N(N-1)\dots(N-n+1)\braket{\Psi}{\hat{f}q_1q_2\dots q_n\Psi}\leq N^n\braket{\Psi}{\hat{f}\hat{m}(n)\Psi}.
	\end{equation}
	for any $n\in \{0,\dots, N\}$, where the operator $\hat{f}$ is defined in \eqref{def:TF}.
\end{lemma}

\begin{proof}
	We can directly compute
	\begin{equation}
		\hat{m}(\xi_1)\hat{m}(\xi_2)=\sum_{k=1}^N\sum_{j=1}^N\left(\frac{k}{N}\right)^{\xi_1}\left(\frac{j}{N}\right)^{\xi_2}P_k P_j=\sum_{k=1}^N\left(\frac{k}{N}\right)^{\xi_1 +\xi_2}P_k\nonumber
	\end{equation}
	where we used that $P_k P_j =\delta_{kj}$. This proves \eqref{eq:m1} and  \eqref{eq:m2}. The third identity follows from the symmetry of $\Psi$ and \eqref{eq:m11}. We have
	\begin{align}
		N(N-1)\dots (N-n&+1)\braket{\Psi}{\hat{f}q_1q_2\dots q_n\Psi}\nonumber\\
		&=\sum_{\substack{j_1, j_2\dots j_n=1\\j_k\neq j_l}}^N\braket{\Psi}{\hat{f}q_{j_1}q_{j_2}\dots q_{j_n}\Psi}\nonumber\\
		&\leq N^n\braket{\Psi}{\hat{f}\hat{m}^n(1)\Psi}\nonumber
	\end{align}
	and we conclude using \eqref{eq:m1}.
\end{proof}
The last tool we need is the shift $\tau_n$ defined as
$(\tau_n f):=f(k+n)$ with the property of \cite[Lemma 3.10]{KnoPic-10} that we repeat here.
\begin{lemma}\label{lem:tau}
	Let $r\geq 1$, $A$ be a $r$-particles operator on $L^2(\R^{2r})$ and $\hat{f}$ an operator defined as in \eqref{def:TF}. Let $Q_i$, $i=1,2$, be two projectors of the form $Q_i=u_1\dots u_r$ where each $u$ stands for either $p$ or $q$. Then
	\begin{equation}
		Q_1 A\hat{f}Q_2=Q_1\widehat{\tau_n f}AQ_2,\nonumber
	\end{equation}
	where $n=n_1 -n_2$ and $n_i$ is the number of factors $q$ in $Q_i$.
\end{lemma}
In the following, we will make use of the two simple relations 
\begin{equation}\label{eq.diffm}
	\left(\hat{m}(1)-\widehat{\tau_{-n}m}(1)\right)q_j=\frac{n}{N}\sum_{k=1}^NP_k q_j=\frac{n}{N}q_j
\end{equation}
for any $n\in\mathbb{N}$ and $j\in\{1,\dots,N\}$ as well as
\begin{align}
	\left(\hat{m}(\tfrac12)-\widehat{\tau_{-n}m}(\tfrac12)\right)&=\sum_{k=1}^{n-1}\sqrt{\frac{k}{N}}P_k +\frac{n}{N}\sum_{k=n}^N\frac{\sqrt{N}}{\sqrt{k}+\sqrt{k-n}}P_k\nonumber\\
	& \leq \frac{n}{N}\hat{m}(-\tfrac12).\label{ine:diffn}
\end{align}

%%%%%%%%%%%%%%%%%%%%MN%%%%%%%%%%%%%%%%%%%%
\subsection{Control of $\sqrt{\mathcal N_+}(t)$}

In order to prove Theorem~\ref{thm:K}, we split the work into three parts that we will treat in the Lemmas~\ref{lem:Vm1/2}, \ref{lem:Wm12} and  \ref{lem:X}. To this end we define
\begin{align}
	H^R_N-H_{N,R}^{\mathrm{H}}:=V+W+X
\end{align}
where $V$ and $W$ are respectively the two-body and three-body terms given by
\begin{align}
	V&:=\beta \sum_{j=1}^{N}\left(\frac{1}{N} \sum_{k\neq j}v(x_j-x_k)-v'(x_j)\right)\label{def:V}\\
	W&:=\beta^2\sum_{j=1}^{N}\left (\frac{1}{N^2} \sum_{\substack{k\neq j\\ l\neq k\neq j}}w(x_j-x_k, x_j-x_l)-w'(x_j)\right)\label{def:W}\\
	X&:=\frac{\beta^2}{N^2}\sum_{j=1}^{N}\sum_{k\neq j}\left| \nabla^{\perp}w_R(x_j -x_k)\right|^2\label{def:X}
\end{align}
and $X$ a two-body error term small in itself and which does not appear in the Hartree energy. Recall that $v$, $v'$, $w$ and $w'$ have been defined in Definition~\ref{def:vw}.
 Let us remark that $\hat{m}(\tfrac12)$ evolves as
\begin{equation}
\im\partial_t \hat{m}(\tfrac12)= [\hat{m}(\tfrac12), H_{N,R}^{\mathrm H}]\nonumber
\end{equation}
providing, for $\Psi_N(t)$ the solution to the Schr\"odinger equation \eqref{def:schro}, the identity
\begin{align}
	\im\partial_t\sqrt{\mathcal N_+}&= \bra{\Psi_N (t)}\ket{\left[(\HNR-H_{N,R}^{\mathrm{H}}),\hat{m}(\tfrac12)\right]\Psi_N (t)}\nonumber\\
	&= \bra{\Psi_N (t)}\ket{\left[(V+W+X),\hat{m}(\tfrac12)\right]\Psi_N (t)}\label{eq:BigSPm172}
\end{align}
with the Hartree Hamiltonian defined in \eqref{eq:hartreeH}. We will treat the terms $V$, $W$ and $X$ one by one in separate lemmas in order to prove the following theorem.

\begin{theorem}[\textbf{Control of the root of excited particles}.]\mbox{}\label{thm:K}
	Let $\Psi_N(t)$ denote the solution of the Schrödinger equation \eqref{def:schro} with the initial data $\Psi_N(0) = \varphi_0^{\otimes N}$, where $\varphi_0 \in H^2(\mathbb{R}^2)$.
	Then there exist four constants $C,C',c,c'>0$ such that for any $|\beta| \leq c'$, $0<R\leq c$ and any $0 \leq t\leq T$ we have
	\begin{equation}
		\partial_t\sqrt{\mathcal N_+}(t)\leq C|\log R|R^{-2}\sqrt{\mathcal N_+}(t)+CN^{-1/2}R^{-2}\nonumber
	\end{equation}
	and one deduces that
	\begin{equation}
		\sqrt{\mathcal N_+}(t)\leq C'N^{-1/2}R^{-2}e^{CT\frac{|\log R|}{R^2}}.\nonumber
	\end{equation}
\end{theorem}
\begin{proof}
We start from Eq. \eqref{eq:BigSPm172}.	We apply Lemma~\ref{lem:Vm1/2} and  \ref{lem:Wm12} and Lemma~\ref{lem:X} with $\xi =1/2$ to treat each term of \eqref{eq:BigSPm172}. We obtain
	\begin{equation}
		\partial_t\sqrt{\mathcal N_+}\leq C|\log R|^4\sqrt{\mathcal N_+}+C|\log R|\norm{\nabla_1 q_1\Psi_N (t)}^2+CN^{-1/2}R^{-2}\nonumber
	\end{equation}
	and conclude using Lemma~\ref{lem:kinetic} to control $\norm{\nabla_1 q_1\Psi_N (t)}^2$ and close the Gr\"onwall argument.
\end{proof}

\subsection{The Two-Body Term $V$}
This section is dedicated to the computation of the $V$ term of \eqref{eq:BigSPm172} defined in \eqref{def:V}. 

\begin{lemma}\label{lem:Vm1/2}Let $\Psi_N(t)$ denote the solution of the Schrödinger equation \eqref{def:schro} with the initial data $\Psi_N(0) = \varphi_0^{\otimes N}$, where $\varphi_0 \in H^2(\mathbb{R}^2)$. Let $V$ be as defined in \eqref{def:V}, and let $\hat{m}(\tfrac12)$ be as defined in \eqref{def:mxi}. Then there exist three constants $C,c,c'>0$ such that for any $|\beta| \leq c'$, $0 \leq R\leq c$ and any $0 \leq t\leq T$ we have
	\begin{equation}
		\left|\bra{\Psi_N (t)}\ket{\left[V, \hat{m}(\tfrac12)\right]\Psi_N (t)}\right|\leq C|\log R|(\sqrt{\mathcal N_+}+\norm{\nabla_1 q_1\Psi_N (t)}^2+\frac{1}{\sqrt{N}})\nonumber
	\end{equation}
	where $q_1$ is defined in \eqref{def:pq} with $\varphi_t$ a solution of $\mathrm{CSS}(R,\varphi_0)$.
\end{lemma}

\begin{proof}
Before entering the proof, note that obtaining a control in terms of $\mathcal N_+(t)$ is enough to close the Gr\"onwall on $\sqrt{\mathcal N_+}(t)$ for the reason that $\mathcal N_+(t)\leq \sqrt{\mathcal N_+}(t)$, as we can see from the definition \eqref{def:mxi}.
	We first expand the term we want to handle
	\begin{align}
		\dot{M}_V :&= -\im\bra{\Psi_N (t)}\ket{\left[V, \hat{m}(\tfrac12)\right]\Psi_N (t)}\nonumber\\
		&=-\im\bra{\Psi_N (t)}\ket{\Big[ \frac{\beta}{N} \sum_{k=1}^{N}\sum_{m> k}(v(x_k-x_m)+v(x_m-x_k))-\sum_{j=1}^Nv'(x_j), \hat{m}(\tfrac12)\Big]\Psi_N (t)}\nonumber\\
		&=-\frac{\im\beta}{2}\bra{\Psi_N (t)}\ket{\Big[ (N-1)(v(x_1-x_2)+ v(x_2-x_1))-Nv'(x_1)-Nv'(x_2), \hat{m}(\tfrac12)\Big]\Psi_N (t)}\nonumber
	\end{align}
	where we used the bosonic symmetry of $\Psi_N(t)$ and the symmetry of $\hat{m}(\tfrac12)$. We now split the space into the condensate part and the excited part
	\[
	(p_1 +q_1)(p_2 +q_2) =\one_{L^2(\R^4)}
	\] 
	to get three terms and their conjugates
	\begin{align}
		\dot{M}_V
		=&-\im\beta\bra{p_1p_2\Psi_N (t)}\ket{\Big[ (N-1)(v(x_1-x_2)+ v(x_2-x_1))-Nv'(x_1)-Nv'(x_2), \hat{m}(\tfrac12)\Big]q_1p_2\Psi_N (t)}\nonumber\\
		&-\frac{\im\beta}{2}\bra{p_1p_2\Psi_N (t)}\ket{\Big[ (N-1)(v(x_1-x_2)+ v(x_2-x_1))-Nv'(x_1)-Nv'(x_2), \hat{m}(\tfrac12)\Big]q_1q_2\Psi_N (t)} \nonumber\\
		&-\im\beta\bra{q_1p_2\Psi_N (t)}\ket{\Big[ (N-1)(v(x_1-x_2)+ v(x_2-x_1))-Nv'(x_1)-Nv'(x_2), \hat{m}(\tfrac12)\Big]q_1q_2\Psi_N (t)}\nonumber\\
		&\quad +\hc\nonumber\\
		:=&\dot{M}^{(1)}_V+\dot{M}^{(2)}_V+\dot{M}_V^{(3)}+\hc
	\end{align}
	Note that we only get six terms because, by Lemma~\ref{lem:tau}, all term with the same number of factors $q$ on both sides of the bracket vanished.
	We then only have three terms to really treat
	and proceed one by one. We start with $\dot{M}^{(1)}_V$ which is the most important term of the three because its smallness comes from the mean-field cancellation, i.e., the fact that
	\begin{align}
		& p_2(v(x_1-x_2)+ v(x_2-x_1))p_2\nonumber\\
		&\quad\quad\quad\quad=p_2\bigg(\int_{\R^2}|\varphi_t(x_2)|^2(-\im \nabla_1\cdot\nabla^{\perp}w_R(x_1 -x_2)+\hc)\nonumber\\
		&\quad\quad\quad\quad\quad\quad\quad\quad\quad\quad\quad\quad+\overline{\varphi_t}(x_2)(-\im \nabla_2\cdot\nabla^{\perp}w_R(x_2 -x_1)+\hc)\varphi_t(x_2)\,\dd x_2 \bigg)p_2\nonumber\\
		&\quad\quad\quad\quad=p_2\bigg(\int_{\R^2}(-\im \nabla_1\cdot\nabla^{\perp}w_R\ast |\varphi_t|^2+\hc)\nonumber\\
		&\quad\quad\quad\quad\quad\quad\quad\quad\quad\quad\quad\quad+\nabla^{\perp}w_R(x_2 -x_1)\cdot(\overline{\varphi_t}(x_2)( -\im\nabla_2)\varphi_t(x_2)+\hc)\,\dd x_2 \bigg)p_2\nonumber\\
		&\quad\quad\quad\quad=p_2v'(x_1)p_2.
	\end{align}
	We can apply the above in
	\begin{align}
		\dot{M}^{(1)}_V
		=&-\im\beta  \bra{p_2p_1}\ket{\Big[ (N-1)(v(x_1-x_2)+ v(x_2-x_1))-Nv'(x_1), \hat{m}(\tfrac12)\Big]q_1p_2}_{\Psi_N(t)}\nonumber\\
		=&-\im\beta \bra{p_2p_1}\ket{\Big[ v'(x_1), \hat{m}(\tfrac12)\Big]q_1p_2}_{\Psi_N(t)}\nonumber\\
		=&-\im\beta \bra{p_2p_1}\ket{v'(x_1)( \hat{m}(\tfrac12)-\widehat{\tau_{-1}m}(\tfrac12))q_1p_2}_{\Psi_N(t)}\nonumber\\
		=&-\frac{\im\beta}{N} \bra{p_2p_1}\ket{v'(x_1)^2p_1p_2}_{\Psi_N(t)}^{1/2}\bra{p_2q_1}\ket{\hat{m}(-\tfrac12)\hat{m}(-\tfrac12)q_1p_2}_{\Psi_N(t)}^{1/2}
		\label{eq:av1K}
	\end{align}
	using that $p_1q_1v'(x_2)=0$ in the first line of above, the inequality \eqref{ine:diffn} to bound
	\[
	\hat{m}(\tfrac12)-\widehat{\tau_{-1}m}(\tfrac12)\leq \frac{1}{N} \hat{m}(-\tfrac12) 
	\] 
	and a Cauchy-Schwarz inequality.
	We obtain
	\begin{align}
		|\dot{M}^{(1)}_V|&\leq\frac{2|\beta|}{N}\norm{v'(x_1)p_2p_1\Psi_N (t)}\norm{\hat{m}(-1)q_1\Psi_N (t)}\nonumber\\
		&\leq\frac{2|\beta|}{N}\norm{v'(x_1)p_2p_1\Psi_N (t)}\norm{\hat{m}(-1)\hat{m}(1)\Psi_N (t)}\nonumber\\
		&\leq \frac{C}{N}\norm{v'\varphi_t}_2
		\leq CN^{-1}\label{eq:E1K}
	\end{align}
	using \eqref{ine:v'} to bound $\norm{v'\varphi_t}_2$.\\
	For the $\dot{M}^{(2)}_V$ term, we follow the same strategy to the difference that the commutation with a two-body operator makes a shift $\tau_{-2}$ appear, we indeed get
	\begin{align*}
		\dot{M}^{(2)}_V
		=&-\im\frac{\beta}{2}  \bra{p_2p_1}\ket{\left((N-1)(v(x_1-x_2)+ v(x_2-x_1))-Nv'(x_1)-Nv'(x_2)\right)( \hat{m}(\tfrac12)-\widehat{\tau_{-2}m}(\tfrac12))q_1q_2}_{\Psi_N(t)}\\
		=&-\im\frac{\beta(N-1)}{2}  \bra{p_2p_1}\ket{(v(x_1-x_2)+ v(x_2-x_1))\hat{m}(\tfrac14)\hat{m}(-\tfrac14)( \hat{m}(\tfrac12)-\widehat{\tau_{-2}m}(\tfrac12))q_1q_2}_{\Psi_N(t)}
	\end{align*}
	in which we used Lemma \ref{lem:tau}, the relation \eqref{eq.diffm} and the fact that $p_2 v'(x_1)q_2 =0=p_1 v'(x_2)q_1$. We have also used the property \eqref{eq:m2} $\hat{m}(\xi)\hat{m}(-\xi)=\one -P_0$ with $P_0q_j =0$.

	By Cauchy-Schwarz we get
	\begin{align*}
		|\dot{M}^{(2)}_V|^2
		&\leq CN^2 \norm{v(x_1-x_2)\widehat{\tau_2 m}(\tfrac14)p_2p_1\Psi_N (t)}^2\norm{\hat{m}(-\tfrac14)( \hat{m}(\tfrac12)-\widehat{\tau_{-2}m}(\tfrac12))q_2q_1\Psi_N (t)}^2\\
		&=C\bra{p_2p_1}\ket{\widehat{\tau_2 m}(\tfrac14)v^2(x_1-x_2)\widehat{\tau_2 m}(\tfrac14)p_2p_1}_{\Psi_N(t)}\bra{q_2q_1}\ket{\hat{m}(-\tfrac12)\hat{m}^2(-1/2)q_2q_1}_{\Psi_N(t)}\\
		&\leq C\norm{v(x_1-x_2)\varphi_t(x_1)\varphi_t(x_2)}_{L^2(\R^4)}^2\norm{\widehat{\tau_2 m}(\tfrac14)\Psi_N (t)}^2\braket{\Psi_N(t)}{\hat{m}(\tfrac12)\Psi_N(t)}(1+\frac{1}{N})
	\end{align*}
	where we apply formula \eqref{ine:symmhat} to bound $\bra{q_2q_1}\ket{\hat{m}(-\tfrac32)q_2q_1}_{\Psi_N(t)}$ and formula \eqref{ine:diffn} to bound the difference $ \hat{m}(\tfrac12)-\widehat{\tau_{-2}m}(\tfrac12)$. We now use that  
	\[
	\bra{\Psi_N(t)}\ket{\widehat{\tau_2 m}(\tfrac12)\Psi_N(t)}\leq\bra{\Psi_N(t)}\ket{\hat{m}(\tfrac12)\Psi_N(t)}+\frac{C}{\sqrt{N}}=\sqrt{\mathcal N_+}(t)+\frac{C}{\sqrt{N}}
	\]
	to get
	\begin{align*}
		|\dot{M}^{(2)}_V|^2
		&\leq C\norm{v(x_1-x_2)\varphi_t(x_1)\varphi_t(x_2)}_{L^2(\R^4)}^2(\sqrt{\mathcal N_+}^2(t)+\frac{1}{N})
		\leq C|\log R|^2(\sqrt{\mathcal N_+}^2(t)+\frac{1}{N})
	\end{align*}
	where we concluded with \eqref{ine:vphiphi}.\newline
	The last term $\dot{M}^{(3)}_V$ provides
	\begin{align}\label{eq:K32}
		\dot{M}^{(3)}_V
		=&-\im\beta  \bra{p_1q_2}\ket{\Big[ (N-1)(v(x_1-x_2)+ v(x_2-x_1))-Nv'(x_1), \hat{m}(\tfrac12)\Big]q_1q_2}_{\Psi_N(t)}\\
		=&-\im\beta  \bra{p_1q_2}\ket{\left( (N-1)(v(x_1-x_2)+ v(x_2-x_1))-Nv'(x_1)\right)(\hat{m}(\tfrac12)-\widehat{\tau_{-1}m}(\tfrac12))q_1q_2}_{\Psi_N(t)}\nonumber
	\end{align}
	by Lemma~\ref{lem:tau}. By the Cauchy-Schwarz inequality
	\begin{align*}
		|\dot{M}^{(3)}_V|
		\leq &CN\norm{(\hat{m}(\tfrac12)-\widehat{\tau_{-1}m}(\tfrac12))q_1q_2\Psi_N(t)}_2\\
		&\qquad\times\bigg(\norm{v(x_1-x_2)p_1q_2\Psi_N(t)}_2+\norm{ v(x_2-x_1)p_1q_2\Psi_N}_2+\norm{v'(x_1)p_1q_2\Psi_N(t)}_2\bigg)\\
		\leq &C|\log R|\norm{\hat{m}(-\tfrac12)q_1q_2\Psi_N(t)}_2(\norm{\nabla_1 q_1\Psi_N(t)}+\norm{v'(x_1)\varphi_t(x_1)}_2)\\
		\leq &C|\log R|(\mathcal N_+(t))^{1/2}(\norm{\nabla_1 q_1\Psi_N(t)}+\norm{v'(x_1)\varphi_t(x_1)}_2)\\
		\leq &C|\log R|(\mathcal N_+(t)+\norm{\nabla_1 q_1\Psi_N(t)}^2)
	\end{align*}
	where we use \eqref{ine:v'}, Lemma~\ref{ine:vpq} and $\mathcal N_+(t)\leq \sqrt{\mathcal N}_+(t)$ to conclude. \end{proof}

\subsection{The Three-Body Term $W$}
Here, we estimate the $W$ term of \eqref{eq:BigSPm172}.
\begin{lemma}\label{lem:Wm12}
	Let $\Psi_N(t)$ denote the solution of the Schrödinger equation \eqref{def:schro} with the initial data $\Psi_N(0) = \varphi_0^{\otimes N}$, where $\varphi \in H^2(\mathbb{R}^2)$. Let $W$ be as defined in \eqref{def:W}, and let $\hat{m}(\tfrac12)$ be as defined in \eqref{def:mxi}. Then there exist two constants $C,c>0$ such that for any $\beta$, $0<R\leq c$ and any $0 \leq t\leq T$ we have
	\begin{equation}
		\left|\bra{\Psi_N (t)}\ket{\left[W, \hat{m}(\tfrac12)\right]\Psi_N (t)}\right|\leq C|\log R|^4\sqrt{\mathcal N_+}(t)+C|\log R|\norm{\nabla_1q_1\Psi_N(t)}^2+\frac{C|\log R|}{\sqrt{N}}\nonumber
	\end{equation}
	for $q_1$ as in \eqref{def:pq} with $\varphi_t$ solving $\mathrm{CSS}(R,\varphi_0)$.
\end{lemma}

\begin{proof}
	Before entering the proof, note thatthat obtaining a control in terms of $\mathcal N_+(t)$ is enough to close the Gr\"onwall on $\sqrt{\mathcal N_+}(t)$ for the reason that $\mathcal N_+(t)\leq \sqrt{\mathcal N_+}(t)$, as we can see from the definition \eqref{def:mxi}
	We denote
	\begin{align}
		\dot{M}_W :&=-\im\bra{\Psi_N(t)}\ket{\Big[\frac{\beta^2 }{N^2} \sum_{k=1}^{N}\sum_{\substack{m\neq k\\n\neq m\neq k}}w(x_k-x_m,x_k -x_n)-\sum_{j=1}^Nw'(x_j), \hat{m}(\tfrac12)\Big]\Psi_N(t) }\nonumber\\
		&=-\im\frac{\beta^2}{3}\left\langle\Big[ \frac{(N-1)(N-2)}{N}\tilde{W}(x_1 , x_2, x_3)-N(w'(x_1)+w'(x_2)+w'(x_3)), \hat{m}(\tfrac12)\Big]\right\rangle_{\Psi_N (t)}\nonumber
	\end{align}
	where we have introduced $$\tilde{W}(x_1 , x_2, x_3):=w(x_1-x_2,x_1-x_3)+w(x_2-x_1,x_2 -x_3)+w(x_2-x_3,x_2 -x_1).$$  Note that $\tilde{W}$ is symmetric under permutations of the particles and positive by Lemma \ref{lem:3body}.
	We split $\one_{L^2(\R^{6})} =(p_1+q_1)(p_2 +q_2)(p_3+q_3)$ on each side of the scalar product and obtain
	six types of terms:
	\begin{align}
		\dot{M}_W
		=&-\im\beta^2 \bra{p_3p_2p_1}\ket{\Big[ \frac{(N-1)(N-2)}{N}\tilde{W}(x_1 , x_2, x_3)-N(w'(x_1)+w'(x_2)+w'(x_3)), \hat{m}(1/2)\Big]q_1p_2p_3}_{\Psi_N(t)}\nonumber
		\\
		&-\im\beta^2  \bra{p_3p_2q_1}\ket{\Big[ \frac{(N-1)(N-2)}{N}\tilde{W}(x_1 , x_2, x_3)-N(w'(x_1)+w'(x_2)+w'(x_3)), \hat{m}(1/2)\Big]q_1q_2q_3}_{\Psi_N(t)}\nonumber
		\\
		&-\im\frac{\beta^2}{3}  \bra{p_3p_2p_1}\ket{\Big[ \frac{(N-1)(N-2)}{N}\tilde{W}(x_1 , x_2, x_3)-N(w'(x_1)+w'(x_2)+w'(x_3)), \hat{m}(1/2)\Big]q_1q_2q_3}_{\Psi_N(t)}\nonumber
		\\
		&-\im\beta^2 \bra{p_3p_2p_1}\ket{\Big[ \frac{(N-1)(N-2)}{N}\tilde{W}(x_1 , x_2, x_3)-N(w'(x_1)+w'(x_2)+w'(x_3)), \hat{m}(1/2)\Big]p_1q_2q_3}_{\Psi_N(t)}\nonumber
		\\
		&-\im\beta^2 \bra{q_3q_2p_1}\ket{\Big[ \frac{(N-1)(N-2)}{N}\tilde{W}(x_1 , x_2, x_3)-N(w'(x_1)+w'(x_2)+w'(x_3)), \hat{m}(1/2)\Big]q_1q_2q_3}_{\Psi_N(t)}\nonumber
		\\
		&-3\im\beta^2 \bra{q_3p_2p_1}\ket{\Big[ \frac{(N-1)(N-2)}{N}\tilde{W}(x_1 , x_2, x_3)-N(w'(x_1)+w'(x_2)+w'(x_3)), \hat{m}(1/2)\Big]p_1q_2q_3}_{\Psi_N(t)}\nonumber\\
		&\quad +\hc \nonumber\\
		&=:\sum_{i=1}^6\dot{M}_W^{(i)} +\hc .\label{eq:sixterms}
	\end{align}
	Similarly to the two-body case we used Lemma~\ref{lem:tau} to discard all the term containing the same number of $q$ factors on the left and on the right side of the scalar product and the fact that $\Psi_N$ as well as the commutator are symmetric. 
	We now estimate these terms one by one.
	We start by noting that we have the mean-field cancellation
	\begin{align}
		p_3 p_2 \tilde{W}(x_1 , x_2, x_3)p_2 p_3&= p_2 p_3\bigg(\int_{\R^4}|\varphi_t (x_2)|^2|\varphi_t (x_3)|^2\bigg(\nabla^{\perp}w_R(x_1 -x_2)\cdot\nabla^{\perp}w_R(x_1 -x_3)\nonumber\\
		&\quad\quad\quad\quad\quad\quad\quad\quad +2\nabla^{\perp}w_R(x_2 -x_1)\cdot\nabla^{\perp}w_R(x_2 -x_3)\bigg)\dd x_1\dd x_2\bigg)p_2p_3\nonumber\\
		&=p_3p_2 w'(x_1)p_2p_3
	\end{align}
	implying that
	\begin{align}
		|\dot{M}^{(1)}_W|
		&\leq C \left | \bra{p_3p_2p_1}\ket{w'(x_1)(\hat{m}(\tfrac12)-\widehat{\tau_{-1}m}(\tfrac12))q_1p_2p_3}_{\Psi_N(t)}\right|\nonumber\\
		&\leq \frac{C}{N} \norm{w'(x_1)\varphi_t(x_1)}_2
		\leq CN^{-1}\label{eq:KW11}
	\end{align}
	where we used the bound \eqref{ine:diffn} and \eqref{ine:w'}. To treat $\dot{M}^{(4)}_W$, we use that $\hat{m}(-\tfrac14)\hat{m}(\tfrac14)q_1=q_1$ to get, 
	\begin{align}
		\dot{M}^{(4)}_W&=-\im\beta^2 \frac{(N-1)(N-2)}{N}\bra{p_3p_2p_1}\ket{\tilde{W}(x_1 , x_2, x_3) \hat{m}(\tfrac14)\hat{m}(-\tfrac14)(\hat{m}(\tfrac12)-\widehat{\tau_{-2}m}(\tfrac12))p_1q_2q_3}_{\Psi_N(t)}\nonumber\\
		&\leq C 
		\bra{p_3p_2p_1}\ket{ \widehat{\tau_2 m}(\tfrac14)\tilde{W}^2(x_1 , x_2, x_3) \widehat{\tau_2 m}(\tfrac14)p_1p_2p_3}^{1/2}_{\Psi_N(t)}\norm{\hat{m}(-\tfrac14)\hat{m}(-\tfrac12)q_2q_3\Psi_N}\nonumber\\
		&\leq C\norm{\tilde{W}\varphi_t(x_1)\varphi_t(x_2)\varphi_t(x_3)}_{L^2(\R^6)}\norm{\widehat{\tau_2 m}(\tfrac12)\Psi_N}\norm{\hat{m}(\tfrac14)\Psi_N}\nonumber\\
		&\leq C|\log R|\,\sqrt{\mathcal N_+}(t)\nonumber
	\end{align}
	via a Cauchy-Schwarz inequality,
	together with  \eqref{ine:symmhat}  and \eqref{eq:m2} on $\norm{\hat{m}(-\tfrac34)q_1q_3\Psi_N}$ and
	where we applied \eqref{ine:wphiphiphi} to bound the last norm.
	The next term is
	\begin{align}
		|\dot{M}^{(6)}_W|
		\leq &C \left| \bra{q_3p_2p_1}\ket{\left(\frac{(N-1)(N-2)}{N}\tilde{W}(x_1 , x_2, x_3)-Nw'(x_1) \right)(\hat{m}(\tfrac12)-\widehat{\tau_{-1}m}(\tfrac12))q_1p_2q_3}_{\Psi_N(t)}\right|\nonumber\\
		\leq& CN\norm{(\tilde{W}(x_1 , x_2, x_3)-w'(x_1))q_3p_2p_1\Psi_N}\norm{(\hat{m}(\tfrac12)-\widehat{\tau_{-1}m}(\tfrac12))q_1p_2q_3\Psi_N}\nonumber\\
		\leq& C\norm{(\tilde{W}(x_1 , x_2, x_3)-w'(x_1))q_3p_2p_1\Psi_N}\sqrt{\mathcal N_+(t)}\nonumber\\
		\leq& C(\mathcal N_+(t))^{1/2}\left(\norm{w'\varphi_t}_2(\mathcal N_+(t))^{1/2}+\norm{\tilde{W}(x_1 , x_2, x_3)q_3p_2p_1\Psi_N}\right).\label{ine:K6}
	\end{align}
	Note that in the above $(\mathcal N_+(t))^{1/2}$ really means the square-root of $\mathcal N_+(t)$ and not the value of $\sqrt{\mathcal N_+}(t)$ that satisfies $\mathcal N_+(t)\leq \sqrt{\mathcal N_+}(t)$. 
	The term $\norm{w'\varphi_t}_{2}$ has been calculated in \eqref{ine:w'}. For the last norm of above, we expand $\tilde{W}^2$ and use for each of the terms
	\begin{equation}
		|\nabla^{\perp}w_R(x-y)\cdot\nabla^{\perp}w_R(x-z)|^2\leq |\nabla^{\perp}w_R(x-y)|^2|\nabla^{\perp}w_R(x-z)|^2\nonumber
	\end{equation}
	in order to distribute the laplacians of Inequality \eqref{eq:sig2bd}, we use that
	\[
	|\nabla_x^{\perp}w_R(x-y)|^2=|\nabla_y^{\perp}w_R(y-x)|^2.
	\]
	We get
	\begin{align}
		|\dot{M}^{(6)}_W|\leq &C\mathcal N_+(t)+ C|\log R|^2\sqrt{\mathcal N_+(t)}\braket{q_3p_1p_2}{(1-\Delta_1)(1-\Delta_2)q_3p_1p_2}^{1/2}_{\Psi_N}\nonumber\\
		\leq & C|\log R|^2\mathcal N_+(t).
		\label{eq:KW13}
	\end{align}
	We now treat
	\begin{align}
		|\dot{M}^{(3)}_W|
		\leq &CN \left| \bra{p_3p_2p_1}\ket{\tilde{W}(x_1 , x_2, x_3) \hat{m}(\tfrac14)\hat{m}(-\tfrac14) (\hat{m}(\tfrac12)-\widehat{\tau_{-3}m}(\tfrac12))q_1q_2q_3}_{\Psi_N(t)}\right|\nonumber\\
		\leq &CN\norm{\hat{m}(-\tfrac14) (\hat{m}(\tfrac12)-\widehat{\tau_{-3}m}(\tfrac12))q_3q_2q_1\Psi_N(t)}\norm{\tilde{W}\widehat{\tau_3m}(\tfrac14)p_3 p_2 p_1\Psi_N(t)}\nonumber\\
		\leq &C\norm{\hat{m}(-\tfrac14) \hat{m}(-\tfrac12)q_3q_2q_1\Psi_N(t)}\norm{\tilde{W}\widehat{\tau_{3}m}(\tfrac12)p_3 p_2 p_1\Psi_N(t)}\nonumber\\
		\leq &C\sqrt{\mathcal N_+}^{1/2}\norm{\tilde{W}(x_1 , x_2, x_3)\varphi_t(x_1)\varphi_t(x_2)\varphi_t(x_3)}_{L^2(\R^6)}\norm{\widehat{\tau_{3}m}(\tfrac14)p_3 p_2 p_1\Psi_N(t)}\nonumber\\
		\leq & C\sqrt{\mathcal N_+}^{1/2}\norm{\tilde{W}(x_1 , x_2, x_3)\varphi_t(x_1)\varphi_t(x_2)\varphi_t(x_3)}_{L^2(\R^6)}(\sqrt{\mathcal N_+}^{1/2}+\frac{C}{\sqrt{N}})\nonumber\\
		\leq& C|\log R|(\sqrt{\mathcal N_+}(t)+N^{-1})
	\end{align}
	by \eqref{ine:wphiphiphi}.
	We continue with
	\begin{align}
		|\dot{M}^{(2)}_W|
		\leq&CN\left|\bra{q_3p_2p_1}\ket{\tilde{W}(x_1 , x_2, x_3)(\hat{m}(\tfrac12)-\widehat{\tau_{-2}m}(\tfrac12))q_1q_2q_3}_{\Psi_N(t)}\right|\nonumber\\
		\leq& C\norm{\tilde{W}q_3 p_2 p_1\Psi_N(t)}\norm{\hat{m}(-\tfrac12)q_3q_2q_1\Psi_N}\nonumber\\
		\leq&C\norm{\tilde{W}q_3 p_2 p_1\Psi_N(t)}(\mathcal N_+(t))^{1/2}\nonumber\\
		\leq& C|\log R|^2\mathcal N_+(t)
	\end{align}
	where we concluded using the previous calculation \eqref{ine:K6}.
	It remains to estimate
	\begin{align}
		\dot{M}^{(5)}_W=&-2\im\beta^2 \Im \bra{q_3q_2p_1}\ket{\left(\frac{(N-1)(N-2)}{N}\tilde{W}(x_1 , x_2, x_3)-Nw'(x_1) \right)(\hat{m}(\tfrac12)-\widehat{\tau_{-1}m}(\tfrac12))q_1q_2q_3}_{\Psi_N(t)} . \nonumber
	\end{align}
	With the help of Lemma~\ref{lem:sing2bd} to bound the following $\norm{\tilde{W}q_3 q_2 p_1\Psi_N(t)}$ term as previously, we have
	\begin{align}
		|\dot{M}^{(5)}_W|\leq&C\norm{w'\varphi_t}_2\norm{\hat{m}(-\tfrac12)q_3q_2q_1\Psi_N}+C\norm{\tilde{W}q_3 q_2 p_1\Psi_N(t)}\norm{\hat{m}(-\tfrac12)q_3q_2q_1\Psi_N}\nonumber\\
		\leq& C(\mathcal N_+(t))^{1/2}(\norm{w'\varphi_t}_2 +|\log R|^2\norm{\nabla_1q_1\Psi_N(t)})\nonumber\\
		\leq& C|\log R|^4(\mathcal N_+(t)+\norm{\nabla_1q_1\Psi_N(t)}^2)
	\end{align}
	where we used \eqref{ine:w'} to conclude. Gathering the previous estimates and bounding $\mathcal N_+(t)\leq \sqrt{\mathcal N}_+(t)$  provides the proof.
\end{proof}

\subsection{The Term $X$}
Here we estimate the term $X$ to finish proving Theorem~\ref{thm:K}. 
\begin{lemma}\label{lem:X}
	Let $\Psi_N(t)$ denote the solution of the Schrödinger equation \eqref{def:schro} with the initial data $\Psi_N(0) = \varphi_0^{\otimes N}$, where $\varphi \in H^2(\mathbb{R}^2)$. Let $X$ be as defined in \eqref{def:X}, and let $\hat{m}(\xi)$ be as defined in \eqref{def:mxi}. Then there exists a constant $C>0$ such that for any $\beta$, $R$, $\xi\geq 1/2$ and any $0\leq t\leq T$ we have
	\begin{equation}
		\left|\bra{\Psi_N (t)}\ket{\left[X, \hat{m}(\xi)\right]\Psi_N (t)}\right|\leq C\frac{1}{N^{\xi}R^2}.\nonumber
	\end{equation}
\end{lemma}
\begin{proof}
	Recall that
	\begin{align}
		X:=&\beta^2N^{-2}\sum_{j=1}^{N}\sum_{k\neq j}\left| \nabla^{\perp}w_R(x_j -x_k)\right|^2,\nonumber
	\end{align}
	and using that $\hat{m}(\xi)=\sum_k\left(\frac{k}{N}\right)^{\xi}P_k\leq \frac{1}{N^{\xi}}\sum_kkP_k=  \frac{1}{N^{\xi}}\sum_{j=1}^N q_j$ we have
	\begin{align}
		\left|\bra{\Psi_N (t)}\ket{\left[X, \hat{m}(\xi)\right]\Psi_N (t)}\right|&\leq\frac{CN}{N^2N^{\xi}} \sum_{j=2}^N\bra{\Psi_N (t)}\ket{\left[\left| \nabla^{\perp}w_R(x_j-x_1)\right|^2,q_1\right]\Psi_N (t)}_{L^2(\R^{2N})}\nonumber\\
		&\leq CN^{-\xi}R^{-2}\nonumber
	\end{align}
	using the bosonic symmetry of $\Psi_N(t)$ and the bound $\sup_{\R^2}|\nabla w_R|\leq C/R$ of \eqref{eq:sup-nabwR}.
\end{proof}

\section{Evolution of the Kinetic Energy}\label{sec:kinetic}
This section is dedicated to the control of the kinetic energy of excited particles. The idea is to use energy conservation combined to the convergence of the ground states of the microscopic and effective models. 
\begin{remark}\label{rem:positivity}
	In our proof, the main obstacle to obtain faster convergence rates in $R$ is located in Lemma~\ref{lem:kinetic} when we are trying to show the negativity of \eqref{eq:posterm}. In \cite{KnoPic-10}, the equivalent of \eqref{eq:posit} was holding by assumption on two-body interaction potential, which is not guaranteed in our case. While the rest of the proof only deals with powers of $|\log R|$ divergences, at this specific step we cannot avoid a $R^{-2}$ which limits us to deal with $R\geq \rateR$ and makes the final result to be as in \eqref{eq:main}. See more details about this technical issue in Remarks~\ref{rem:posproof} and \ref{rem:posproof2}.
\end{remark}
\subsection{Bound on the Kinetic Energy}
Our aim in this section is to control the $\norm{\nabla_1 q_1\Psi_N(t)}$ term appearing in Theorem~\ref{thm:K}. We will do it through the conservation of the energies
\begin{equation}
	E_N(t):=\frac{1}{N}\braket{\Psi_N(t)}{\HNR\Psi_N(t)}\quad\text{and}\quad\mathcal{E}_R^{\mathrm{af}}[u]:=\int_{\R^2}\left|(-\im\nabla +\beta \bAR[|u|^2])u\right|^2 \nonumber
\end{equation}
combined with the fact that $E_N(0)\simeq \mathcal{E}_R^{\mathrm{af}}[\varphi_0]+O(N^{-1}R^{-2})$.
We define $\sqrt{\mathcal N_+}(t):=\braket{\Psi_N(t)}{\hat{m}(\tfrac12)\Psi_N(t)}$ and the lemma is given as follows:
\begin{theorem}[\textbf{Control of the kinetic energy of excited particles}]\mbox{}\label{lem:kinetic}
	Let $\Psi_N(t)$ denote the solution of the Schrödinger equation \eqref{def:schro} with the initial data $\Psi_N(0) = \varphi_0^{\otimes N}$, where $\varphi_0 \in H^2(\mathbb{R}^2)$. Let $\hat{m}(\tfrac12)$ be as defined in \eqref{def:mxi}. Then there exist two constants $C,c,c'>0$ such that for any $|\beta| \leq c'$, $0<R\leq c$ and any $0\leq t\leq T$ we have
	\begin{equation}
		\norm{\nabla_1 q_1\Psi_N}^2\leq \left| E_N(t)-\mathcal{E}_R^{\mathrm{af}}[\varphi_t] \right| +CR^{-2}\sqrt{\mathcal N_+}(t)+C\frac{|\log R|}{\sqrt{N}}
	\end{equation}
	where $q_1$ is defined as in \eqref{def:pq} with $\varphi_t$ solving $\mathrm{CSS}(R,\varphi_0)$.
	Moreover we have that
	\begin{equation}\label{ine:ener}
		\left| E_N(t)-\mathcal{E}_R^{\mathrm{af}}[\varphi_t] \right| \leq C|\log R|^2N^{-1}	
	\end{equation}
	as a consequence of energy conservation.
\end{theorem}
In this Lemma appears the quantity $\hat{m}(\tfrac12)$. We will obtain a control on it later in Theorem~\ref{thm:K} through a Gr\"onwall-type estimate similar to the one used in Section~\ref{sec:pickl}.

\begin{proof}
	Recalling definitions \eqref{def:v} and \eqref{def:w},  we introduce the notation 
	\[
	\tilde{v}(x):=-\im\nabla_x\cdot \bAR[|\varphi|^2](x)+\mathrm{h.c}\quad \text{and}  \quad\tilde{w}(x):= (\bAR[|\varphi|^2](x))^2
	\] 
	to write
	\begin{align}
		\mathcal{E}_{R}^{\mathrm{af}}[\varphi] = & \langle\varphi,-\Delta\varphi\rangle+\beta\langle\varphi,\tilde{v}\varphi\rangle+\beta^2\langle\varphi,\tilde{w}\varphi\rangle\nonumber\\
		E[\Psi] := & \langle\Psi,-\Delta_{x_1}\Psi\rangle+\beta\langle\Psi,v(x_1,x_2)\Psi\rangle+\frac{(N-2)}{(N-1)}\beta^2\langle\Psi,w(x_1-x_2, x_1-x_3)\Psi\rangle\label{eq:EPsi}.
	\end{align}
	Note that by symmetry of $\Psi_N(t)$ and Lemma~\ref{lem:smCp} we have that 
	\begin{align}
		E[\Psi_N(t)]&= E_N(t)+\frac{\beta^2}{N-1} \langle\Psi_N|\nabla^{\perp}w_R(x_1-x_2)|^2 \Psi_N\rangle\nonumber\\
		&=E_N(t)+O(R^{-2}N^{-1})\label{eq:EptoEN}.
	\end{align}
	Using that $p_{j},q_{j}\leq\one(x_j)_{L^2(\R^2)}$ are projections, applying the Cauchy-Schwarz inequality
	and the symmetry of $\Psi_N$, we have
	\begin{align*}
		\|\nabla_{1}q_{1}\Psi_{N}(t)\| & =\|\nabla_{1}(1-p_{1}(p_{2}+q_{2})(p_{3}+q_{3}))\Psi_{N}(t)\|\\
		& \leq\|\nabla_{1}(1-p_{1}p_{2}p_{3})\Psi_{N}(t)\|+3\|\nabla_{1}p_{1}\|\|q_{2}\Psi_{N}(t)\|\\
		& \leq\|\nabla_{1}(1-p_{1}p_{2}p_{3})\Psi_{N}(t)\|+3\|\nabla_{1}p_{1}\|\|q_{2}\Psi_{N}(t)\|\\
		& \leq\|\nabla_{1}(1-p_{1}p_{2}p_{3})\Psi_{N}(t)\|+3\|\nabla\varphi_{t}\|\langle\Psi_{N}(t),q_{2}\Psi_{N}(t)\rangle^{1/2}\\
		& \leq\|\nabla_{1}(1-p_{1}p_{2}p_{3})\Psi_{N}(t)\|+3\|\nabla\varphi_{t}\|\left\langle\Psi_{N}(t),\hat{m}(1)\Psi_{N}(t)\right\rangle^{1/2}\\
		& \leq\|\nabla_{1}(1-p_{1}p_{2}p_{3})\Psi_{N}(t)\|+3\|\nabla\varphi_{t}\|\left\langle\Psi_{N}(t),\hat{m}(\tfrac12)\Psi_{N}(t)\right\rangle^{1/2}.
	\end{align*}

	The rest of the proof uses the same techniques as Section~\ref{sec:pickl} in order to bound $\|\nabla_{1}(1-p_{1}p_{2}p_{3})\Psi_{N}(t)\|$. As we assumed $|\beta| \leq c$, we will systematically treat $\norm{\varphi_t}_{H^1}$ and  $\norm{\varphi_t}_{H^2}$ as constants for a better readability.
	For the rest of the proof, we write the $\Psi$ for $\Psi_{N}(t)$ to have a shorter notation.
	Inserting
	\[
	\one_{L^2(\R^6)}=p_{1}p_{2}p_{3}+(\one_{L^2(\R^6)}-p_{1}p_{2}p_{3}),
	\]
	into \eqref{eq:EPsi}, and denoting $\Psi^{'}:=(\one_{L^2(\R^6)}-p_{1}p_{2}p_{3})\Psi$ and $h=-\Delta$ and $h_1=-\Delta_{x_1}$, 
	we have
	\begin{align*}
		E[\Psi]= & \langle\Psi,p_{1}p_{2}p_{3}h_{1}p_{1}p_{2}p_{3}\Psi\rangle
		+\langle\Psi,p_{1}p_{2}p_{3}h_{1}\Psi^{'}\rangle+\langle\Psi^{'},h_{1}p_{1}p_{2}p_{3}\Psi\rangle
		+\langle\Psi^{'}h_{1}\Psi^{'}\rangle\\
		& +\beta\Big(\langle\Psi_N,p_{1}p_{2}p_{3}v(x_1,x_2)p_{1}p_{2}p_{3}\Psi\rangle
		+\langle\Psi,p_{1}p_{2}p_{3}v(x_1,x_2)\Psi^{'}\rangle+\langle\Psi^{'}_N,v(x_1,x_2)p_{1}p_{2}p_{3}\Psi\rangle\\
		& +\langle\Psi^{'},v(x_1,x_2)\Psi^{'}\rangle\Big)\\
		& +\beta^2\frac{N-2}{N-1}\Big(\langle\Psi,p_{1}p_{2}p_{3}w(x_1-x_2,x_1-x_3)p_{1}p_{2}p_{3}\Psi\rangle\\
		& +\langle\Psi,p_{1}p_{2}p_{3}w(x_1-x_2,x_1-x_3)\Psi^{'}\rangle+\langle\Psi^{'},w(x_1-x_2,x_1-x_3)p_{1}p_{2}p_{3}\Psi\rangle\\
		& +\langle\Psi^{'},w(x_1-x_2,x_1-x_3)\Psi^{'}\rangle\Big).
	\end{align*}
	Then, using some further short-hand notation, $v:=v(x_1,x_2)$ and $w:=w(x_1-x_2,x_1-x_3)$, we obtain for any $0<\kappa<1$ that
	\begin{align}
		\kappa\langle\Psi, & (1-p_{1}p_{2}p_{3})h_{1}(1-p_{1}p_{2}p_{3})\Psi\rangle\nonumber\\
		= & E[\Psi]-\mathcal{E}_{R}^{\mathrm{af}}[\varphi_t]\nonumber\\
		& -\langle\Psi,p_{1}p_{2}p_{3}h_{1}p_{1}p_{2}p_{3}\Psi\rangle+\langle\varphi_t,h\varphi_t\rangle\label{eq:h1}\\
		& -\langle\Psi,p_{1}p_{2}p_{3}h_{1}(1-p_{1}p_{2}p_{3})\Psi\rangle-\langle\Psi,p_{1}p_{2}p_{3}h_{1}(1-p_{1}p_{2}p_{3})\Psi\rangle\label{eq:h1ppp}\\
		& -\beta\Big(\langle\Psi,p_{1}p_{2}p_{3}vp_{1}p_{2}p_{3}\Psi\rangle-\langle\varphi_t,\tilde{v}\varphi_t\rangle\Big)\label{eq:v}\\
		& -\beta\Big(\langle\Psi,p_{1}p_{2}p_{3}v(1-p_{1}p_{2}p_{3})\Psi\rangle+\langle\Psi,(1-p_{1}p_{2}p_{3})vp_{1}p_{2}p_{3}\Psi\rangle\Big)\label{eq:vppp}\\
		& -\beta^2\Big(\langle\Psi,p_{1}p_{2}p_{3}wp_{1}p_{2}p_{3}\Psi\rangle-\langle\varphi_t,\tilde{w}\varphi_t\rangle\Big)\label{eq:w}\\
		& -\beta^2\Big(\langle\Psi,p_{1}p_{2}p_{3}w(1-p_{1}p_{2}p_{3})\Psi\rangle+\langle\Psi,(1-p_{1}p_{2}p_{3})wp_{1}p_{2}p_{3}\Psi\rangle\Big)\label{eq:wppp}\\
		& +\frac{\beta^2}{N-1}\Big(\langle\Psi,p_{1}p_{2}p_{3}wp_{1}p_{2}p_{3}\Psi\rangle
		\nonumber\\&\hspace{2cm}
		+\langle\Psi,p_{1}p_{2}p_{3}w(1-p_{1}p_{2}p_{3})\Psi\rangle+\langle\Psi,(1-p_{1}p_{2}p_{3})wp_{1}p_{2}p_{3}\Psi\rangle\Big)\label{eq:wrest}\\
		& -\Big(\beta\langle\Psi,(1-p_{1}p_{2}p_{3})v(1-p_{1}p_{2}p_{3})\Psi\rangle\nonumber\\
		& \qquad+\beta^2\frac{N-2}{N-1}\langle\Psi,(1-p_{1}p_{2}p_{3})w(1-p_{1}p_{2}p_{3})\Psi\rangle
		\nonumber\\& \qquad
		+(1-\kappa)\langle\Psi,(1-p_{1}p_{2}p_{3})h_{1}(1-p_{1}p_{2}p_{3})\Psi\rangle\Big).\label{eq:posterm}
	\end{align}

	We now estimate the terms line by line, using the same use of the properties of Lemmas~\ref{lem:mxi} and \ref{lem:tau} as previously.

	\noindent{\it Line of expression} \eqref{eq:h1} is
	\begin{align*}
		& \left|\langle\Psi,p_{1}p_{2}p_{3}h_{1}p_{1}p_{2}p_{3}\Psi\rangle-\langle\varphi,h\varphi\rangle\right|\\
		& =\langle\varphi,h\varphi\rangle\left|\langle\Psi,(1-p_{1}p_{2}p_{3})\Psi\rangle\right|\\
		& =\langle\varphi,h\varphi\rangle\left|\langle\Psi,(q_{1}p_{2}p_{3}+p_{1}q_{2}p_{3}+p_{1}p_{2}q_{3}+p_{1}q_{2}q_{3}+q_{1}p_{2}q_{3}+q_{1}q_{2}p_{3}+q_{1}q_{2}q_{3})\Psi\rangle\right|\\
		& \leq7\langle\varphi,h\varphi\rangle \langle\Psi,\hat{m}(1)\Psi\rangle\\
		& \leq C \langle\Psi,\hat{m}(\tfrac12)\Psi\rangle.
	\end{align*}

	\noindent{\it Line of expression} \eqref{eq:h1ppp} gives
	\begin{align*}
		& \langle\Psi,p_{1}p_{2}p_{3}h_{1}(1-p_{1}p_{2}p_{3})\Psi\rangle\\
		& =\langle\Psi,p_{1}p_{2}p_{3}h_{1}(q_{1}p_{2}p_{3}+p_{1}q_{2}p_{3}+p_{1}p_{2}q_{3}+p_{1}q_{2}q_{3}+q_{1}p_{2}q_{3}+q_{1}q_{2}p_{3}+q_{1}q_{2}q_{3})\Psi\rangle\\
		& =\langle\Psi,p_{1}p_{2}p_{3}h_{1}q_{1}p_{2}p_{3}\Psi\rangle\\
		& =\langle\Psi,q_{1}h_{1}p_{1}p_{2}p_{3}\Psi\rangle\\
		& =\langle\Psi,q_{1}\hat{m}(\tfrac12)^{-1/2}\hat{m}(\tfrac12)^{1/2}h_{1}p_{1}p_{2}p_{3}\Psi\rangle\\
		& =\langle\Psi,q_{1}\hat{m}(\tfrac12)^{-1/2}h_{1}\widehat{\tau_{1} m}(\tfrac12)^{1/2}h_{1}p_{1}p_{2}p_{3}\Psi\rangle.
	\end{align*}
	Thus,
	\begin{align*}
		& \left|\langle\Psi,p_{1}p_{2}p_{3}h_{1}(1-p_{1}p_{2}p_{3})\Psi\rangle\right|\\
		& \leq\sqrt{\langle\Psi,q_{1}\hat{m}(\tfrac12)^{-1}\Psi\rangle}\sqrt{\langle\Psi,p_{1}p_{2}p_{3}\widehat{\tau_{1} m}(\tfrac12)^{1/2}h_{1}^{2}\widehat{\tau_{1} m}(\tfrac12)^{1/2}p_{1}p_{2}p_{3}\Psi\rangle\rangle}\\
		& =\sqrt{\langle\Psi,\hat{m}(\tfrac12)\Psi\rangle}\sqrt{\langle\varphi,h^{2}\varphi\rangle}\sqrt{\langle\Psi,\widehat{\tau_{1} m}(\tfrac12)p_{1}p_{2}p_{3}\Psi\rangle\rangle}\\
		& \leq C\sqrt{\sqrt{\mathcal N_+}(t)}\left(\sqrt{\sqrt{\mathcal N_+}(t)}+\frac{1}{N^{1/4}}\right)\\
		& \leq C\left(\sqrt{\mathcal N_+}(t)+\frac{1}{\sqrt{N}}\right).
	\end{align*}

	\noindent{\it Line of expression}  \eqref{eq:v} provides
	\begin{align*}
		& \left|\langle\Psi,p_{1}p_{2}p_{3}vp_{1}p_{2}p_{3}\Psi\rangle-\langle\varphi,\tilde{v}\varphi\rangle\right|
		\leq\langle\varphi,\tilde{v}\varphi\rangle\left|\langle\Psi,p_{1}p_{2}p_{3}\Psi\rangle-1\right|
		\leq C\left|\langle\Psi,p_{1}p_{2}p_{3}\Psi\rangle-1\right|
		\leq C \sqrt{\mathcal N_+}(t)
	\end{align*}
	where we used the estimate \eqref{eq:N4} to get 
	\begin{align*}
		\langle\varphi,\tilde{v}\varphi\rangle&=\norm{\varphi (-\im\nabla\varphi)\cdot\bAR[|\varphi|^2]}_1
		\leq \norm{\nabla\varphi}_2\norm{\varphi}_4\norm{\bAR[|\varphi|^2]}_4
		\leq C.
	\end{align*}

	\noindent{\it Line of expression} \eqref{eq:vppp} is
	\begin{align*}
		& \left|\langle\Psi,p_{1}p_{2}p_{3}v(1-p_{1}p_{2}p_{3})\Psi\rangle\right|\\
		& =\left|\langle\Psi,p_{1}p_{2}p_{3}v(q_{1}p_{2}p_{3}+p_{1}q_{2}p_{3}+p_{1}p_{2}q_{3}+p_{1}q_{2}q_{3}+q_{1}p_{2}q_{3}+q_{1}q_{2}p_{3}+q_{1}q_{2}q_{3})\Psi\rangle\right|\\
		& \leq\left|\langle\Psi,p_{1}p_{2}p_{3}vq_{1}p_{2}p_{3}\Psi\rangle\right|+\left|\langle\Psi,p_{1}p_{2}p_{3}vp_{1}q_{2}p_{3}\Psi\rangle\right|+\left|\langle\Psi,p_{1}p_{2}p_{3}vp_{1}p_{2}q_{3}\Psi\rangle\right|\\
		& \quad+\left|\langle\Psi,p_{1}p_{2}p_{3}vp_{1}q_{2}q_{3}\Psi\rangle\right|+\left|\langle\Psi,p_{1}p_{2}p_{3}vq_{1}q_{2}p_{3}\Psi\rangle\right|\left|\langle\Psi,p_{1}p_{2}p_{3}vq_{1}p_{2}q_{3}\Psi\rangle\right|\\
		& \quad+\left|\langle\Psi,p_{1}p_{2}p_{3}vq_{1}q_{2}q_{3}\Psi\rangle\right|\\
		& =:\mathrm{(a)}+\mathrm{(b)}+\mathrm{(b')}+\mathrm{(c)}+\mathrm{(d)}+\mathrm{(d')}+\mathrm{(e)}.
	\end{align*}
	We have $\mathrm{(c)}=\mathrm{(e)}=0$ because $p_3q_3=0$. 

	Now, for $\mathrm{(a)}$, we have
	\begin{align*}
		& \left|\langle\Psi,p_{1}p_{2}p_{3}vq_{1}p_{2}p_{3}\Psi\rangle\right|\\
		& =\left|\langle\Psi,p_{1}p_{2}p_{3}\tilde{v}(x_1)q_{1}\Psi\rangle\right|\\
		& =\left|\langle\Psi,p_{1}p_{2}p_{3}\tilde{v}(x_1)\hat{m}(\tfrac12)^{1/2}\hat{m}(\tfrac12)^{-1/2}q_{1}\Psi\rangle\right|\\
		& =\left|\langle\Psi,p_{1}p_{2}p_{3}\widehat{\tau_{1} m}(\tfrac12)^{1/2}\tilde{v}(x_1)\hat{m}(\tfrac12)^{-1/2}q_{1}\Psi\rangle\right|\\
		& \leq\sqrt{\langle\varphi,\tilde{v}^{2}(x_1)\varphi\rangle}\sqrt{\langle\Psi,\widehat{\tau_{1} m}(\tfrac12)p_{1}p_{2}p_{3}\Psi\rangle\rangle}\sqrt{\langle\Psi,\hat{m}(\tfrac12)^{-1}q_{1}\Psi\rangle}\\
		& \leq C\left(\sqrt{\mathcal N_+}(t)+\frac{1}{\sqrt{N}}\right)
	\end{align*}
	by the estimate $\norm{\tilde{v}\varphi_t}_2\leq C\norm{\nabla\varphi\cdot \bAR[|\varphi|^2]}_2\leq C$ as in the proof of the estimate \eqref{def:v'}.
	For $\mathrm{(b)}$ and  $\mathrm{(b')}$, similarly we obtain
	\begin{align*}
		\left|\langle\Psi,p_{1}p_{2}p_{3}vp_{1}q_{2}p_{3}\Psi\rangle\right|+\left|\langle\Psi,p_{1}p_{2}p_{3}vp_{1}p_{2}q_{3}\Psi\rangle\right|
		\leq C\left(\sqrt{\mathcal N_+}(t)+\frac{1}{\sqrt{N}}\right).
	\end{align*}

	For $\mathrm{(d)}$,
	\begin{align*}
		& \left|\langle\Psi,p_{1}p_{2}p_{3}vq_{1}q_{2}p_{3}\Psi\rangle\right|\\
		& =\left|\langle\Psi,p_{1}p_{2}p_{3}v\hat{m}(\tfrac12)\hat{m}(\tfrac12)^{-1}q_{1}q_{2}\Psi\rangle\right|\\
		& =\left|\langle\Psi,p_{1}p_{2}p_{3}\widehat{\tau_{1} m}(\tfrac12) v\hat{m}(\tfrac12)^{-1}q_{1}q_{2}\Psi\rangle\right|\\
		& \leq\norm{v(x-y)\varphi(x)\varphi(y)}_{L^2(\R^4)}\sqrt{\langle\Psi,\widehat{\tau_{1} m}(\tfrac12)p_{1}p_{2}p_{3}\Psi\rangle\rangle}\sqrt{\langle\Psi,\hat{m}(\tfrac12)^{-1}q_{1}q_{2}\Psi\rangle}\\
		& \leq C|\log R|\left(\langle\Psi,\hat{m}(1)\Psi\rangle+\frac{1}{\sqrt{N}}\right)\\
		& \leq C|\log R| \left(\langle\Psi,\hat{m}(\tfrac12)\Psi\rangle+\frac{1}{\sqrt{N}}\right)
	\end{align*}
	where we used the estimate \eqref{ine:vphiphi}. The same holds true for $\mathrm{(d')}$.\newline
	\noindent{\it Line of expression} \eqref{eq:w}
	\begin{align*}
		\left|\langle\Psi,p_{1}p_{2}p_{3}wp_{1}p_{2}p_{3}\Psi\rangle-\langle\varphi,\tilde{w}\varphi\rangle\right|
		& \leq\langle\varphi,\tilde{w}\varphi\rangle\left|\langle\Psi,p_{1}p_{2}p_{3}\Psi\rangle-1\right|
		\leq C\norm{|\varphi|^2}_2\norm{\bAR[|\varphi|^2]}_4^2 \sqrt{\mathcal N_+}(t)
		\leq C \sqrt{\mathcal N_+}(t)
	\end{align*}
	using the estimate \eqref{eq:N4}.
	\noindent{\it Line of expression} \eqref{eq:wppp}
	\begin{align*}
		& \langle\Psi,p_{1}p_{2}p_{3}w(1-p_{1}p_{2}p_{3})\Psi\rangle\\
		& =\langle\Psi,p_{1}p_{2}p_{3}w(q_{1}p_{2}p_{3}+p_{1}q_{2}p_{3}+p_{1}p_{2}q_{3}+p_{1}q_{2}q_{3}+q_{1}p_{2}q_{3}+q_{1}q_{2}p_{3}+q_{1}q_{2}q_{3})\Psi\rangle\\
		& \leq\left|\langle\Psi,p_{1}p_{2}p_{3}wq_{1}p_{2}p_{3}\Psi\rangle\right|+
		2\left|\langle\Psi,p_{1}p_{2}p_{3}wp_{1}p_{2}q_{3}\Psi\rangle\right|+\left|\langle\Psi,p_{1}p_{2}p_{3}wp_{1}q_{2}q_{3}\Psi\rangle\right|\\
		&\quad+2\left|\langle\Psi,p_{1}p_{2}p_{3}wq_{1}q_{2}p_{3}\Psi\rangle\right|+\left|\langle\Psi,p_{1}p_{2}p_{3}wq_{1}q_{2}q_{3}\Psi\rangle\right|\\
		& =:\mathrm{(a)}+2\mathrm{(a)'}+\mathrm{(b)}+2\mathrm{(b')}+\mathrm{(c)}.
	\end{align*}

	For $\mathrm{(a)}$,
	\begin{align*}
		& \left|\langle\Psi,p_{1}p_{2}p_{3}wq_{1}p_{2}p_{3}\Psi\rangle\right|\\
		& =\left|\langle\Psi,p_{1}p_{2}p_{3}\tilde{w}(x_1)q_{1}\Psi\rangle\right|\\
		& =\left|\langle\Psi,p_{1}p_{2}p_{3}\tilde{w}(x_1)\hat{m}(\tfrac12)^{1/2}\hat{m}(\tfrac12)^{-1/2}q_{1}\Psi\rangle\right|\\
		& =\left|\langle\Psi,p_{1}p_{2}p_{3}\widehat{\tau_{1} m}(\tfrac12)^{1/2}\tilde{w}(x_1)\hat{m}(\tfrac12)^{-1/2}q_{1}\Psi\rangle\right|\\
		& \leq\sqrt{\langle\varphi,\tilde{w}^2\varphi\rangle}\sqrt{\langle\Psi,\widehat{\tau_{1} m}(\tfrac12)p_{1}p_{2}p_{3}\Psi\rangle\rangle}\sqrt{\langle\Psi,\hat{m}(\tfrac12)^{-1}q_{1}\Psi\rangle}\\
		& \leq C\left(\sqrt{\mathcal N_+}(t)+\frac{1}{\sqrt{N}}\right)
	\end{align*}
	using that $\langle\varphi,\tilde{w}^2\varphi\rangle \leq \norm{|\varphi|^2}_2\norm{\bAR[|\varphi|^2]}_8^4\leq C$ by Estimate \eqref{eq:N4}. We bound $\mathrm{(a')}$ the same way using that by \eqref{eq:WY} we have
	\begin{align*}
		\norm{(\nabla^{\perp}w_R\ast\bAR[|\varphi|^2]|\varphi|^2)^2|\varphi|^2}_1&\leq \norm{|\varphi|^2}_2\norm{\nabla^{\perp}w_R\ast\bAR[|\varphi|^2]|\varphi|^2}_4^2\\
		&\leq  \norm{|\varphi|^2}_2  \norm{\nabla^{\perp}w_R}_{2,w}\norm{\bAR[|\varphi|^2]|\varphi|^2}_{\frac43}^2
		\leq C.
	\end{align*}
	For $\mathrm{(b)}$ and $\mathrm{(b')}$, 
	\begin{align*}
		& \left|\langle\Psi,p_{1}p_{2}p_{3}wq_{1}q_{2}p_{3}\Psi\rangle\right|+ \left|\langle\Psi,p_{1}p_{2}p_{3}wp_{1}q_{2}q_{3}\Psi\rangle\right|\\
		& =\left|\langle\Psi,p_{1}p_{2}p_{3}\widehat{\tau_{1} m}(\tfrac12)^ {}w\hat{m}(\tfrac12)^{-1}q_{1}q_{2}p_3\Psi\rangle\right|+\left|\langle\Psi,p_{1}p_{2}p_{3}\widehat{\tau_{1} m}(\tfrac12)^ {}w\hat{m}(\tfrac12)^{-1}p_{1}q_{2}q_3\Psi\rangle\right|\\
		& \leq\sqrt{\langle\varphi^{\otimes 3},w^{2}\varphi^{\otimes 3}\rangle}\sqrt{\langle\Psi,\widehat{\tau_{1} m}(\tfrac12)p_{1}p_{2}p_{3}\Psi\rangle\rangle}(\sqrt{\langle\Psi,\hat{m}(\tfrac12)^{-1}q_{1}q_{2}\Psi\rangle}+\sqrt{\langle\Psi,\hat{m}(\tfrac12)^{-1}q_{2}q_{3}\Psi\rangle})\\
		& \leq C |\log R| \left(\langle\Psi,\hat{m}(1)\Psi\rangle+\frac{1}{\sqrt{N}}\right)\\
		& \leq C |\log R|\left(\langle\Psi,\hat{m}(\tfrac12)\Psi\rangle+\frac{1}{\sqrt{N}}\right)
	\end{align*}
	where we concluded using \eqref{ine:wphiphiphi}.\newline
	For $\mathrm{(c')}$, the same argument provides
	\begin{align*}
		& \left|\langle\Psi,p_{1}p_{2}p_{3}wq_{1}q_{2}q_{3}\Psi\rangle\right| \leq C|\log R|\left(\sqrt{\mathcal N_+}(t)+\frac{1}{\sqrt{N}}\right).
	\end{align*}
	\noindent{\it Line of expression} \eqref{eq:wrest} is bounded, using \eqref{ine:wphiphiphi} and that $\snorm{\Psi'_N}\leq 1$.
	We obtain via the Cauchy-Schwarz inequality on $\braket{wp_1p_2p_3\Psi_N}{\Psi_N'}$ that  \eqref{eq:wrest} is controlled by
	\begin{align}
		C\frac{\beta^2}{N}\norm{wp_1p_2p_3\Psi_N}&\leq C\frac{\beta^2}{N}\langle\varphi^{\otimes 3},w^{2}\varphi^{\otimes 3}\rangle^{1/2}
		\leq  C\frac{|\log R|}{N}.
	\end{align}

	\noindent{\it Line of expression} \eqref{eq:posterm}. This is one of the main technical points of the paper. Following \cite[Proof of Theorem 4.1]{KnoPic-10}, we would like to prove that
	\begin{equation}\label{eq:posit}
		0\leq (1-\kappa)(-\Delta_{x_1})+ \beta v(x_1,x_2) + \frac{(N-2)}{(N-1)} \beta^2 w(x_1-x_2,x_1-x_3)
	\end{equation}
	using that our $\beta$ is small.

	More precisely, we drop the $w$ term of above which is positive, see Lemma~\ref{lem:3body}, and use a weighted Cauchy-Schwarz for operators to get
	\begin{equation}\label{eq:CSv}
		\beta v(x_1,x_2)\geq -(-\frac{1}{2}\Delta_{x_1}+2\beta^{2}|\nabla^{\perp}w_R(x_1 -x_2)|^2)
	\end{equation}
	such that 
	\begin{align}
		\beta\braket{\Psi'_N}{v(x_1,x_2)\Psi'_N}
		&\geq -\frac{1}{2} \braket{\Psi'_N}{-\Delta_{x_1}\Psi'_N}-CR^{-2}\braket{\Psi'_N}{\Psi'_N}\nonumber\\
		&\geq  -\frac{1}{2} \braket{\Psi'_N}{-\Delta_{x_1}\Psi'_N}-CR^{-2}\sqrt{\mathcal N_+}(t).
	\end{align}
	We now choose $\kappa=\frac{1}{2}$ such that $1-\kappa-\frac{1}{2} \geq 0$ and include the $CR^{-2}\hat{m}(\tfrac12)$ in the final bound.
	All in all, we have proven that
	\begin{align*}
		\langle\Psi, & (1-p_{1}p_{2}p_{3})h_{1}(1-p_{1}p_{2}p_{3})\Psi\rangle\\
		& \leq |E[\Psi]-\mathcal{E}_{R}^{\mathrm{af}}[\varphi]|+R^{-2}\sqrt{\mathcal N_+}(t)+\frac{|\log R|}{\sqrt{N}}\\
		&\leq |E_N(t)-\mathcal{E}_{R}^{\mathrm{af}}[\varphi]|+CR^{-2}N^{-1}+R^{-2}\sqrt{\mathcal N_+}(t)+\frac{|\log R|}{\sqrt{N}}
	\end{align*}
	by \eqref{eq:EptoEN}.
	We conclude the proof by showing that
	\begin{align}
		\left| E_N(t)-\mathcal{E}_R^{\mathrm{af}}[\varphi_t] \right| &=\left| E_N(0)-\mathcal{E}_R^{\mathrm{af}}[\varphi_0] \right|\nonumber\\
		&\leq \frac{\beta^2}{N-1}\braket{\varphi_0^{\otimes 3}}{(|\nabla^{\perp}w_R(x_1-x_2)|^2 +\nabla^{\perp}w_R(x_1-x_2)\cdot\nabla^{\perp}w_R(x_1-x_3))\varphi_0^{\otimes 3}}|\nonumber\\
		&\leq CN^{-1}|\log R|^2\braket{\varphi_0}{(1-\Delta)\varphi_0}|\nonumber\\
		&\leq C|\log R|^2N^{-1}
	\end{align}
	where we used the conservation of energy Theorem~\ref{Thm:cons} together with the bounds \eqref{eq:sig2bd} and \eqref{eq:W3C}.
\end{proof}
\begin{remark}\label{rem:posproof}
		Note that the above proof, the operator acts on $\Psi_N'=(\one_{L^2(\R^6)}-p_{1}p_{2}p_{3})\Psi$ which is symmetric in the particles $(x_1,x_2,x_3)$ and in $(x_4,\cdots ,x_N)$ but not in $(x_1,\cdots,x_N)$. Because of this lack of symmetry, we do not manage to prove the above inequality using the completion of the square and are forced to use the rough bound \eqref{eq:sup-nabwR} providing a $R^{-2}$ instead of, ideally, a $|\log R|$. This strongly deteriorates the rate of the final Theorem~\ref{thm:main} as explained in Remark~\ref{rem:positivity}. 
	\end{remark}
	\begin{remark}\label{rem:posproof2}
		Note that a naive completion of the square in \eqref{eq:posit}, only using the symmetry of $\Psi_N'$ in the variables $(x_1,x_2,x_3)$ provides the same $|\nabla^{\perp}w_R(x_1 -x_2)|^2$ type of term as in \eqref{eq:CSv}. We would need to use the symmetry in all the particles to complete the square and be left with the usual $N^{-1}|\nabla^{\perp}w_R(x_1 -x_2)|^2$.
	\end{remark}

\subsection{Conclusion}\label{sec:con}
\begin{proof}[Proof of Theorem~\ref{thm:main}]
	We can now prove Theorem~\ref{thm:main}. We work as usual under the assumption that $\varphi_0\in H^2(\R^2)$, $|\beta|\leq c$, $R\leq c'$ and that $0\leq t\leq T$ to make use of Theorem~\ref{thm:K}. There exists a constant $C>0$ such that
	\begin{align}
		\mathcal N_+ (t)\leq \sqrt{\mathcal N}_+(t)\leq \frac{C}{\sqrt{N}R^2}e^{CT\frac{|\log R|}{R^2}}.\nonumber
	\end{align}
	We now split
	\begin{equation}
		\Tr\left|\gamma_N^{(k)}(t)-\sket{\varphi_t^{\otimes k}}\sbra{\varphi_t^{\otimes k}}\right|
		\leq 	
		\Tr\left|\sket{\varphi_t^{R}}\sbra{\varphi_t^{R}}^{\otimes k}-\sket{\varphi_t^{\otimes k}}\sbra{\varphi_t^{\otimes k}}\right|
		+
		\Tr\left|\gamma_N^{(k)}(t)-\sket{\varphi_t^{R}}\sbra{\varphi_t^{R}}^{\otimes k}\right|.\nonumber
	\end{equation}
	To bound the first term of above, we apply Lemma~\ref{lem:Conv_pilot}, which gives
	\begin{align}\label{ine:convpilo}
		\Tr\left|\sket{\varphi_{t}^{R}} \sbra{\varphi_{t}^{R}}^{\otimes k}- \sket{\varphi_{t}^{\otimes k}} \sbra{\varphi_{t}^{\otimes k}}\right|
		\leq
		2k\norm{\varphi_{t}^{R}-\varphi_{t}}_{2}
		\leq
		C Re^{CT} \,.
	\end{align}
	For the second, we use \eqref{ine:EAB} to obtain
	\[
	\Tr\left|\sket{\varphi_t^{R}}\sbra{\varphi_t^{R}}^{\otimes k}-\sket{\varphi_t^{\otimes k}}\sbra{\varphi_t^{\otimes k}}\right|=R_N^{(k)}(t) \leq \sqrt{8 E_N^{(k)}(t)} \leq  \frac{Ck}{RN^{1/4}}e^{CT\frac{|\log R|}{R^2}}.
	\]
	We then get the final result by choosing $R\geq (\log N)^{-\frac{1}{2}+\varepsilon}$ for any $\varepsilon >0$.
	We then have
	\begin{equation}
		\Tr\left|\gamma_N^{(k)}(t)-\sket{\varphi_t^{\otimes k}}\sbra{\varphi_t^{\otimes k}}\right|
		\leq
		\frac{Ce^{CT}}{(\log N)^{\frac{1}{2}-\varepsilon}}+ \frac{C}{N^{1/4}} N^{\frac{CT\log \log N}{(\log N)^{\varepsilon}}}\to 0\nonumber
	\end{equation}
	as $N\to \infty$.
\end{proof}

\vspace{2em}

\appendix
\newpage

\section{Useful Estimates}
Here we repeat some useful generic bounds well-known in the literature.
\begin{lemma}[\textbf{Bound on the magnetic term}]\label{lem:af_three_body}
	We have for any $u \in L^2(\R^2)$ and $R\geq 0$, that
	\[
	\int_{\R^2} \left| \bAR[|u|^2] \right|^2 |u|^2 
	\le \frac{3}{2} \|u\|_{2}^4 \int_{\R^2} \left| \nabla |u| \right|^2
	\le \frac{3}{2} \|u\|_{2}^4 \int_{\R^2} \left| \nabla |u| \right|^2.
	\]
\end{lemma}
\begin{proof}
	The proof follows the one of \cite[Lemma A.1]{LunRou} but with $\nabla^{\perp}w_R$ instead of $\nabla^{\perp}w_0$.
	We have
	\begin{align*}
		& \int_{\R^2} \left| \bAR[|u|^2](x) \right|^2 |u(x)|^2 \,\mathrm{d}x 
		= \iiint_{\R^6} \frac{x-y}{|x-y|_R^2} \cdot \frac{x-z}{|x-z|_R^2} |u(x)|^2 |u(y)|^2 |u(z)|^2 \,\mathrm{d}x\mathrm{d}y\mathrm{d}z \\
		&\leq \frac{1}{6} \int_{\R^6} \frac{1}{\rho(X)^2} \left| |u|^{\otimes 3} \right|^2 \mathrm{d}X 
		\le \frac{1}{2} \int_{\R^6} \left| \nabla_X |u|^{\otimes 3} \right|^2 \mathrm{d}X
		= \frac{3}{2} \int_{\R^2} \left| \nabla|u(x)| \right|^2 \mathrm{d}x
		\left( \int_{\R^2} |u(x)|^2 \mathrm{d}x \right)^2.
	\end{align*}
	Here, we used the symmetry and that 
	\begin{equation}
		0\leq\frac{x-y}{|x-y|_R^2} \cdot \frac{x-z}{|x-z|_R^2}+\frac{y-x}{|y-x|_R^2} \cdot \frac{y-z}{|y-z|_R^2}+\frac{z-y}{|z-y|_R^2} \cdot \frac{z-x}{|z-x|_R^2}\leq \frac{C}{\rho^2(x,y,z)}
	\end{equation}
	as shown in \cite[Proof of Lemma 2.4, inequality (2.13)]{LunRou} with $\rho^2(x,y,z)=|x-y|^2+|y-z|^2+|z-x|^2$ and where we concluded applying Hardy' inequality \cite[Lemma 2.5]{LunRou}.
\end{proof}

\begin{lemma}[\textbf{Diamagnetic inequality}]\label{lem:af_smooth_ineqs}
	We have for $u \in H^1(\R^2)$ and $R\geq 0$, that
	\begin{equation}\label{eq:af_diamagnetic}		
		\int_{\R^2} \left| (\nabla + \im\beta \bAR[|u|^2])u \right|^2 
		\ge \int_{\R^2} \left| \nabla |u| \right|^2.
	\end{equation}
\end{lemma}
This is the usual diamagnetic inequality, the proof can be founded in \cite[Lemma 1.4]{LarLun-16}.

\begin{proposition}\label{prop:conv}
	The functional $\mathcal{E}_R^{\mathrm{af}}[u_0]$ converges pointwise to $\mathcal{E}^{\mathrm{af}}[u_0]$ as $R\to 0$
	\begin{equation}
		\left|\mathcal{E}_R^{\mathrm{af}}[u_0]-\mathcal{E}^{\mathrm{af}}[u_0]\right|\leq C(1+\mathcal{E}^{\mathrm{af}}[u_0]^{\frac32})R.
	\end{equation}

\end{proposition}
This proposition can be found in \cite[Proposition A.6]{LunRou}.

\section{BBGKY Hierarchy}\label{app:BBGKY}
Here we show how the effective equation \eqref{eq:pilotu} can be formally inferred from the BBGKY hierarchy. 
We start from \eqref{expanded_H} and express each term of \eqref{def:Schro2} in terms of the reduced density matrices to get
\begin{align}\label{eq:bbgky}
	\im \partial_{t}\gamma^{(k)}_{N}
	&=\sum_{j=1}^{k}\left [-\Delta_{j},\gamma^{(k)}_{N}\right ]\nonumber\\
	&+\alpha \sum_{\substack {j,l =1 \\ j \neq l }}^{k}\left[\left(-\im\nabla_{j}.\nabla^{\perp}w_{R}(x_{j}-x_{l})+\nabla^{\perp}w_{R}(x_{j}-x_{l}).-\im\nabla_{j}\right),\gamma^{(k)}_{N}\right]\nonumber\\
	&+\alpha(N-k) \sum_{j=1}^{k}\Tr_{k+1}\left[\left(-\im\nabla_{j}.\nabla^{\perp}w_{R}(x_{j}-x_{k+1})+\nabla^{\perp}w_{R}(x_{j}-x_{k+1}).-\im\nabla_{j}\right),\gamma^{(k+1)}_{N}\right]\nonumber\\
	&+\alpha(N-k) \sum_{l=1}^{k}\Tr_{k+1}\left[\left(-\im\nabla_{k+1}.\nabla^{\perp}w_{R}(x_{k+1}-x_{l})+\nabla^{\perp}w_{R}(x_{k+1}-x_{l}).-\im\nabla_{k+1}\right),\gamma^{(k+1)}_{N}\right]\nonumber\\
	&+2\alpha^{2} \sum_{\substack {j,l =1 \\ j < l }}^{k}\left[\left|\nabla^{\perp}w_{R}(x_{j}-x_{l})\right|^{2},\gamma^{(k)}_{N}\right]\nonumber\\
	&+2(N-k)\alpha^{2} \sum_{j=1}^{k}\Tr_{k+1}\left[\left|\nabla^{\perp}w_{R}(x_{j}-x_{k+1})\right|^{2},\gamma^{(k+1)}_{N}\right]\nonumber\\
	&+\alpha^{2} \sum_{\substack {j,l,m =1 \\ j \neq l\neq m }}^{k}\left[\nabla^{\perp}w_{R}(x_{j}-x_{l}).\nabla^{\perp}w_{R}(x_{j}-x_{m}),\gamma^{(k)}_{N}\right]\nonumber\\
	&+2(N-k)\alpha^{2} \sum_{\substack {j,m =1 \\ j \neq  m }}^{k}\Tr_{k+1}\left[\nabla^{\perp}w_{R}(x_{j}-x_{k+1}).\nabla^{\perp}w_{R}(x_{j}-x_{m}),\gamma^{(k+1)}_{N}\right]\nonumber\\
	&+(N-k)\alpha^{2} \sum_{\substack {l,m =1 \\  l\neq m }}^{k}\Tr_{k+1}\left[\nabla^{\perp}w_{R}(x_{k+1}-x_{l}).\nabla^{\perp}w_{R}(x_{k+1}-x_{m}),\gamma^{(k+1)}_{N}\right]\nonumber\\
	&+2(N-k)(N-k-1)\alpha^{2} \sum_{m=1}^{k}\Tr_{k+1,k+2}\left[\nabla^{\perp}w_{R}(x_{k+1}-x_{k+2}).\nabla^{\perp}w_{R}(x_{k+1}-x_{m}),\gamma^{(k+2)}_{N}\right]\nonumber\\
	&+2(N-k)(N-k-1)\alpha^{2} \sum_{j=1}^{k}\Tr_{k+1,k+2}\left[\nabla^{\perp}w_{R}(x_{j}-x_{k+1}).\nabla^{\perp}w_{R}(x_{j}-x_{k+2}),\gamma^{(k+2)}_{N}\right].
\end{align}

\newpage
The limit can be taken formally with $\alpha=\beta/(N-1)$ %and $R=1/N^{\eta}$ 
to obtain
\begin{align}
	\im \partial_{t}\gamma^{(k)}_{\infty}=&\sum_{j=1}^{k}\left [-\Delta_{j},\gamma^{(k)}_{\infty}\right ]\nonumber\\
	&+ \sum_{j=1}^{k}\Tr_{k+1}\left[\left(-\im\nabla_{j}.\nabla^{\perp}w_{R}(x_{j}-x_{k+1})+\nabla^{\perp}w_{R}(x_{j}-x_{k+1}).-\im\nabla_{j}\right),\gamma^{(k+1)}_{\infty}\right]\nonumber\\
	&+ \sum_{l=1}^{k}\Tr_{k+1}\left[\left(-\im\nabla_{k+1}.\nabla^{\perp}w_{R}(x_{k+1}-x_{l})+\nabla^{\perp}w_{R}(x_{k+1}-x_{l}).-\im\nabla_{k+1}\right),\gamma^{(k+1)}_{\infty}\right]\nonumber\\
	&+ 2\sum_{m=1}^{k}\Tr_{k+1,k+2}\left[\nabla^{\perp}w_{R}(x_{k+1}-x_{k+2}).\nabla^{\perp}w_{R}(x_{k+1}-x_{m}),\gamma^{(k+2)}_{\infty}\right]\nonumber\\
	&+ \sum_{j=1}^{k}\Tr_{k+1,k+2}\left[\nabla^{\perp}w_{R}(x_{j}-x_{k+1}).\nabla^{\perp}w_{R}(x_{j}-x_{k+2}),\gamma^{(k+2)}_{\infty}\right]\nonumber
\end{align}
under the formal assumption that $\gamma^{(k)}_{\infty}$ exists. If we also assume that it takes the form of a product state $\gamma^{(k)}=\sket{\varphi^{\otimes k}}\sbra{ \varphi^{\otimes k}} $ in the limit, we obtain
\begin{align}
	\im \partial_{t}\gamma^{(1)}=&\left [-\Delta_{1},\gamma^{(1)}\right ]\nonumber\\
	&+ \beta\Tr_{2}\left[\left(-\im\nabla_{1}.\nabla^{\perp}w_{R}(x_{1}-x_{2})+\nabla^{\perp}w_{R}(x_{1}-x_{2}).-\im\nabla_{1}\right),\gamma^{(2)}\right]\nonumber\\
	&+ \beta\Tr_{2}\left[\left(-\im\nabla_{2}.\nabla^{\perp}w_{R}(x_{2}-x_{1})+\nabla^{\perp}w_{R}(x_{2}-x_{1}).-\im\nabla_{2}\right),\gamma^{(2)}\right]\nonumber\\
	&+ 2\beta^{2}\Tr_{2,3}\left[\nabla^{\perp}w_{R}(x_{2}-x_{3}).\nabla^{\perp}w_{R}(x_{2}-x_{1}),\gamma^{(3)}\right]\nonumber\\
	&+ \beta^{2}\Tr_{2,3}\left[\nabla^{\perp}w_{R}(x_{1}-x_{2}).\nabla^{\perp}w_{R}(x_{1}-x_{3}),\gamma^{(3)}\right].
\end{align}
By taking the limit $R\to 0$, we get
\begin{align}
	\im \partial_{t}\gamma^{(1)}=&\left [\left (-\im\nabla +\beta \bA\left [| \varphi_t|^{2}\right ]\right )^{2},\gamma^{(1)}\right ]\nonumber\\
	&-\beta\left [\nabla^{\perp}w_{R}\ast\left (2\beta \bA\left [| \varphi_t|^{2}\right ]|\varphi_t|^{2}+\im\left (\varphi_t\nabla\bar{\varphi_t}-\bar{\varphi_t}\nabla\varphi_t\right ) \right ),\gamma^{(1)}\right ].\nonumber
\end{align}
The above is equivalent to the CSS effective equation
\begin{align}\label{eq:pilot}
	\im \partial_{t}\varphi_t=&\left (-\im\nabla +\beta \bA\left [|\varphi_t|^{2}\right ]\right )^{2}\varphi_t\nonumber\\
	&- \beta\left [\nabla^{\perp}w_{0}\ast\left (2\beta \bA\left [| \varphi_t|^{2}\right ]|\varphi_t|^{2}+\im\left (\varphi_t\nabla\overline{\varphi_t}-\overline{\varphi_t}\nabla\varphi_t \right ) \right )\right ]\varphi_t
\end{align}
in which
\begin{equation}\label{app:CSS}
	\bA\left [|\varphi_t|^{2}\right ]:=\bA^{\!R=0}\left [|\varphi_t|^{2}\right ]
	=\nabla^{\perp}w_0\ast|\varphi_t|^{2}.
\end{equation}

\section{Remark on the Dynamics with Two-Body Interaction}\label{sec:twobody}

By following \cite{Pickl2015RMP}, adding a two-body interaction $\mathsf{W}_{N,\delta}$ to the energy should be strait-forward. The new Hamiltonian would be
\begin{equation}\label{def:HN+2bdint}
	H_{N}=\sum_{j=1}^{N}\left(-\im\nabla_j+\alpha \bAR_j\right)^{2}-\frac{g}{N-1}\sum_{i\neq j=1}^{N}\mathsf{W}_{N,\nu}(x_i-x_j),
\end{equation}
where $\alpha = \beta/(N-1)$ and
\[
\mathsf{W}_{N,\nu}(x)=N^{2\nu}\mathsf{w}(N^\nu x) \quad\text{for some }\mathsf{w}\in C_c^\infty \text{ with } \|\mathsf{w}\|_1 = 1 \text{ and } \nu>0.
\]
Note that the ground state of such a Hamiltonian is studied in detail in \cite{Ngu_24, Ata_lun_Ngu_24} for many interesting regimes of $R$ and $\nu$.

The new CSS equation would include a nonlinear interaction with coupling constant $g$
\begin{align}
	\im \partial_{t}u =&\left (-\im\nabla +\beta \bA\left [| u|^{2}\right ]\right )^{2}u - g|u|^2 u\nonumber\\
	&- \beta\left [\nabla^{\perp}w_{R}\ast\left (2\beta \bA\left [| u|^{2}\right ]| u|^{2}+\im\left (u\nabla\overline{u}-\overline{u}\nabla u \right ) \right )\right ]u.
\end{align}
It is clear from \cite{Pickl2015RMP} that the derivation of this new effective equation holds with the interacting $H_N$.
We would then have, for $R = N^{-\eta}$, $0<\eta<c$ and some restricted $\nu>0$, that
\begin{equation}
	\Tr\left|\gamma_N^{(k)}(t)-\sket{\varphi_t^{\otimes k}}\sbra{\varphi_t^{\otimes k}}\right| \to 0
\end{equation}
as $N\to \infty$.
The only point that would need to be checked is the $R$-dependent well-posedness and the control of $\snorm{u}_{H^2}$ of Theorem~\ref{thm:wellpose}. Some work would be needed to show that similar results still hold when adding the term $g|u|^2 u$ to the effective equation. We refer to \cite[Theorem 2]{Ngu_24} for detail about the specific possible interesting regimes.

\section{Global Welposedness of $\mathrm{CSS}(R,\varphi_0)$}
We give here an other version of Theorem~\ref{thm:wellpose} in which we allow some divergence in the parameter $R$. It provides the global well-posedness for any $R>0$, $T(R)=+\infty$, and does not require the constraint $|\beta|\leq c$. The price to pay to relax this constraint on $\beta$ is the diverging control we obtain in \eqref{ine:H2norm2} compared to \eqref{ine:normH1u}.
\begin{theorem}[Global well-posedness in $H^2$]\label{thm:u-GWP}
	For any $0\leq R<1$ and any initial data $u_0\in H^2(\R^2)$ there exists a time $T(R)>0$ only depending on $R$ and on $\norm{u_0}_{H^2}$ such that $\mathrm{CSS}(R,u_0)$ has a unique solution
	\[
	u\in C([0,T(R)],H^2(\R^2)).
	\] 
	Moreover, when $R>0$ we have the control
	\begin{align}\label{ine:H2norm2}
		\norm{-\Delta u_t}_2&\leq  CR^{-Ct}
	\end{align}
	implying that $T(R)=+\infty$.
\end{theorem}

\begin{proof}
	Note that the abstract well-posedness is provided in the proof of Theorem~\ref{thm:wellpose}.
	Thus, it is enough to bound $\norm{-\Delta u_t}_2$ globally in time.
	\begin{align}
		\norm{-\Delta u_t}_2&=\norm{(\nabla  -\im \mathcal{A}_R +\im \mathcal{A}_R)^2 u_t}_2\nonumber\\
		&\leq \norm{(\nabla  -\im \mathcal{A}_R )^2u_t}_2+\norm{\beta\bAR[|u_t|^2](-\im\nabla+\beta\bAR[|u_t|^2])u_t}_2+|\beta|^2\norm{\left|\bAR[|u_t|^2]\right|^2u_t}_2\nonumber\\
		&\leq \norm{(\im\partial_t -\mathcal{A}_R^0)u_t}_2 +|\beta|\norm{\bAR[|u_t|^2]}_{\infty}(\mathcal{E}_R^{\mathrm{af}}[u_0]+|\beta|\norm{\bAR[|u_t|^2]u_t}_2)\nonumber\\
		&\leq   \norm{\partial_t u_t}_2+ \norm{\mathcal{A}_R^0u_t}_2 + C\sqrt{|\log R|}\,\norm{u_0}^3_{H^1}\nonumber\\
		&\leq   \norm{\partial_t u_t}_2+ C\sqrt{|\log R|}\,\norm{u_t}^3_{H^1}.\label{ine:H2ofu}
	\end{align}
	Here we replaced the constraint on $\beta$ of Theorem~\ref{thm:wellpose} by a $|\log R|$ divergence.
	We now use \eqref{ine:magtopartu1} to compute
	\begin{align}
		\partial_t \norm{\partial_t u_t}_2^2=\partial_t \norm{\mathcal{A} u_t}_2^2=2\im\Im\braket{(\im\partial_t \mathcal{A})u_t}{\mathcal{A}u_t}+2\im\Im\braket{\mathcal{A}^2u_t}{\mathcal{A}u_t}\nonumber.
	\end{align}	
	The second term in the above is the imaginary part of a real number and cancels. We then have to bound
	\begin{align}
		2\im\Im\braket{(\im\partial_t \mathcal{A})u_t}{\mathcal{A}u_t}&=2\im \Im \braket{\im\partial_t (-\im\nabla+\beta\bAR[|u_t|^2])^2u_t+(\im\partial_t\mathcal{A}_R^0)u_t}{\mathcal{A}u_t}
		:=E^{(1)}+E^{(2)}.
	\end{align}
	To treat $E^{(1)}$, we first calculate
	\begin{align}
		\im\partial_t \bAR[|u_t|^2]&=2\nabla^{\perp}w_R\ast\Re[\overline{u}_t\im\partial_t u_t]
	\end{align}	
	and get
	\begin{align}
		|E^{(1)}|&=8\left|\Im \braket{\beta\nabla^{\perp}w_R\ast\Re[\overline{u}_t\im\partial_t u_t]\cdot(-\im\nabla+\beta\bAR[|u_t|^2]) u_t}{\partial_tu_t}\right|
	\end{align}
	with
	\begin{align}
		|E^{(1)}|&\leq C\snorm{\nabla^{\perp}w_R\ast\Re[\overline{u}_t\im\partial_t u_t]}_{\infty}\snorm{(-\im\nabla+\beta\bAR[|u_t|^2])u_t}_2\norm{\partial_tu_t}_2\nonumber\\
		&\leq C\snorm{\partial_tu_t}_2\norm{u_0}_{H^1}\snorm{\nabla^{\perp}w_R}_{2+\varepsilon}\norm{\overline{u}_t\im\partial_t u_t}_{\frac{2+\varepsilon}{1+\varepsilon}}\nonumber\\
		&\leq C\sqrt{|\log R|}\snorm{\partial_tu_t}_2\norm{u_0}_{H^1}\norm{\im\partial_t u_t}_{2}\snorm{u_t}_{\frac{2+2\varepsilon}{\varepsilon}}\nonumber\\
		&\leq C|\log R|\snorm{\partial_tu_t}^2_2\norm{u_0}^2_{H^1}\label{ine:EAR1}
	\end{align}
	where we used the inequality \eqref{ine:magtopartu1} to get the last line but also Young inequality for the convolution, the conservation of energy \eqref{eq:conservation} and the estimates \eqref{ine:normH1u} with $\varepsilon=|\log R|^{-1}$. 
	We now treat
	\begin{align}
		|E^{(2)}|&=2|\im \Im \braket{(\im\partial_t\mathcal{A}_R^0)u_t}{\mathcal{A}u_t}|\nonumber\\
		&\leq C\snorm{\partial_tu_t}_2\norm{u_t}_4\snorm{\im\partial_t\mathcal{A}_R^0}_4\nonumber\\
		&\leq C\snorm{\partial_tu_t}_2\snorm{u_t}_4\snorm{\nabla^{\perp}w_R}_{2,w}\norm{4\beta\bAR[\Re(\overline{u_t}\im\partial_t u_t)]|u_t|^2+4\beta\bAR[|u_t|^2]\Re(\overline{u_t}\im\partial_t u_t)}_{\frac43}\nonumber\\
		&\qquad+ C\snorm{\partial_tu_t}_2\snorm{u_t}_4\snorm{\nabla^{\perp}w_R\ast \mathbf{J}[\im\partial_t u_t]}_4.\nonumber
	\end{align}
	Here we have two $\frac43$-norms to treat and the last term. The first gives 
	\begin{align}
		\snorm{4\beta\bAR[\Re(\overline{u_t}\im\partial_t u_t)]|u_t|^2}_{\frac43}&\leq C\snorm{|u_t|^2}_2\snorm{\nabla^{\perp}w_R\ast \Re(\overline{u_t}\im\partial_t u_t)}_4\nonumber\\
		&\leq  C\snorm{|u_t|^2}_2\snorm{\nabla^{\perp}w_R}_{2,w}\snorm{\overline{u_t}\im\partial_t u_t}_{\frac43}\nonumber\\
		&\leq Cs\norm{|u_t|^2}_2\snorm{\nabla^{\perp}w_R}_{2,w}\snorm{\im\partial_t u_t}_{2}\snorm{u_t}_4.
	\end{align}
	Here we used \eqref{eq:WY}.
	The second provides
	\begin{align}
		\norm{4\beta\bAR[|u_t|^2]\Re(\overline{u_t}\im\partial_t u_t)}_{\frac43}&\leq C\norm{\partial_tu_t}_2\norm{\bAR[|u_t|^2]}_8\norm{u_t}_8
		\leq  C\norm{\partial_tu_t}_2\norm{u_0}^2_{H^1}.
	\end{align}
	The last term to treat is
	\begin{align}
		\norm{\nabla^{\perp}w_R\ast \mathbf{J}[\im\partial_t u_t]}_4&= \norm{\nabla^{\perp}w_R\ast \im[(\im\partial_tu_t)\nabla\overline{u}_t+u_t\nabla(\im\partial_t\overline{u}_t)-(\im\partial_t\overline{u}_t)\nabla u_t-\overline{u}_t\nabla (\im\partial_tu_t)]}_4\nonumber\\
		&\leq  4\norm{\nabla^{\perp}w_R\ast (\im\partial_tu_t)\nabla\overline{u}_t}_4\nonumber\\
		&\leq C\snorm{\nabla^{\perp}w_R}_{2,w}\norm{ (\im\partial_tu_t)\nabla\overline{u}_t}_{\frac43}\nonumber\\
		&\leq C\snorm{\nabla^{\perp}w_R}_{2,w}\norm{ \partial_tu_t}_{2}\norm{u_t}_4
	\end{align}
	where we first used that $\nabla^{\perp}\cdot\nabla =0$ to integrate by parts in the convolution product the two $u_t\nabla(\im\partial_t\overline{u}_t)$ terms and then \eqref{eq:WY}.
	We then conclude that
	\begin{align}\label{ine:E2}
		|E^{(2)}|
		&\leq C\norm{\partial_tu_t}^2_2(\norm{u_t}^4_{H^1}+\norm{u_t}^3_{H^1}+\norm{u_t}^2_{H^1}).
	\end{align}
	If we gather \eqref{ine:EAR1} and \eqref{ine:E2} we get 
	\begin{align}\label{ine:dtu}
		\partial_t \norm{\partial_t u_t}_2^2&\leq  C|\log R|\norm{\partial_t u_t}_2^2.
	\end{align}
	We can apply Gr\"onwall's lemma~\ref{lem:Gron} to the above line and obtain
	\begin{align}
		\norm{\partial_t u_t}_2^2&\leq C\norm{\mathcal{A}u_0}_2^2 \, e^{C|\log R|t}\nonumber\\
		&\leq C(\norm{u_0}^2_{H^2}+\norm{u_0}^6_{H^1}) \, e^{C|\log R|t}
		\leq CR^{-Ct}.
	\end{align}
	This last inequality plugged into \eqref{ine:H2ofu}  allows us to state that $T(R)=+\infty$ and concludes the proof.
\end{proof}

\section{Sharpness of Lemma \ref{lem:sing2bd}}\label{app:counterexample}

	The logarithmic loss in Lemma \ref{lem:sing2bd} is optimal. The factor $|\log R|^2$ cannot be improved to $|\log R|$. Here we show a counter example example:
	For $0<R<1$, set $L:=|\log R|,$ and define the radial function
	\[
	f_R(x)=
	\begin{cases}
		\sqrt L, & |x|\le R,\\[1mm]
		\dfrac{\log(1/|x|)}{\sqrt L}, & R<|x|<1,\\[2mm]
		0, & |x|\ge1.
	\end{cases}
	\]
	Then $f_R\in H^1(\R^2)$ and
	\[
	\|\nabla f_R\|_{L^2}^2
	=
	2\pi\int_R^1
	\frac1{r^2L}\,r\,\dd r
	=
	2\pi.
	\]
	Moreover,
	\[
	\|f_R\|_{L^2}^2
	=
	\pi R^2L
	+
	\frac{2\pi}{L}
	\int_R^1
	\log^2(1/r)\,r\,\dd r
	\le C,
	\]
	so that
	\[
	\|f_R\|_{H^1(\R^2)}
	\le C
	\]
	uniformly in $R$.
	On the other hand,
	\[
	\begin{aligned}
		\int_{\R^2}
		\frac{|f_R(x)|^2}{\max(|x|,R)^2}\,\dd x
		&\ge
		\frac{2\pi}{L}
		\int_R^1
		\log^2(1/r)\,\frac{\dd r}{r}.
	\end{aligned}
	\]
	After the change of variables $t=\log(1/r)$,
	\[
	\frac{2\pi}{L}
	\int_R^1
	\log^2(1/r)\,\frac{\dd r}{r}
	=
	\frac{2\pi}{L}
	\int_0^L t^2\,\dd t
	=
	\frac{2\pi}{3}L^2.
	\]
	Consequently,
	\[
	\int_{\R^2}
	\frac{|f_R(x)|^2}{\max(|x|,R)^2}\,\dd x
	\gtrsim
	\log^2\frac1R\,
	\|f_R\|_{H^1(\R^2)}^2.
	\]
	Therefore the estimate
	\[
	\int_{\R^2}
	\frac{|f(x)|^2}{\max(|x|,R)^2}\,\dd  x
	\le
	C\left(1+\log^2\frac1R\right)
	\|f\|_{H^1(\R^2)}^2
	\]
	is sharp in the scale of $H^1(\R^2)$.
\bigskip

\section*{Acknowledgment}
We thank Paolo Antonelli, Kihyun Kim, and Sung-Jin Oh for helpful discussion about Chern--Simons--Schr\"odinger equation.
We also appreciate Serena Cenatiempo, Younghun Hong, Peter Pickl, and Nicolas Rougerie for insightful discussion.
T.G. is supported by the mathematics area of GSSI, L'Aquila, Italy.
J.L. is supported by the Swiss National Science Foundation through the NCCR SwissMAP, the SNSF Eccellenza project PCEFP2\_181153, 
by the Swiss State Secretariat for Research and Innovation through the project P.530.1016 (AEQUA), and
Basic Science Research Program through the National Research Foundation of Korea (NRF) funded by the Ministry of Education (RS-2024-00411072).
Part of this work was carried out during a visit to GSSI. We would like to express our gratitude to GSSI for hosting and funding J.L.'s visit to L'Aquila.

\bibliographystyle{siam}
\bibliography{biblio_Dec2024}
\end{document}